\documentclass[oneside, a4paper,11pt,reqno]{amsart}

\usepackage[T1]{fontenc}
\usepackage{verbatim}
\usepackage{amsmath}
\usepackage{amsfonts}
\usepackage{color}
\usepackage{multicol}
\usepackage{graphicx}
\usepackage{amssymb}
\usepackage{ulem}

\newenvironment{enonce}[1]
    { \par \medskip  \noindent {\bf #1}
    }
    { \\
    }
    
\newtheorem{theo}{Theorem}
\newtheorem{prop}[theo]{Proposition}

\newtheorem{lem}[theo]{Lemma}

\newcommand {\pare}[1] {\left( {#1} \right)}
\newcommand {\cro}[1] {\left[ {#1} \right]}
\newcommand {\acc}[1] {\left\{ {#1} \right\}}
\newcommand {\nor}[1] { \left\| {#1} \right\|}

\newcommand {\bra}[1] { \left\langle {#1} \right\rangle}

\newcommand {\abs}[1] {\left\lvert {#1} \right\rvert}
\newcommand {\va}[1] {\left| {#1} \right|}
\newcommand {\sous}[2] { \begin{array}{c} \scriptstyle{{#1}} \\[-.2cm] \scriptstyle{{#2}} \end{array}}
\newcommand {\cal}{\mathcal}

\def \E {\mathbb{E}}
\def \N  {\mathbb{N}} 
\def \P {\mathbb{P}}
\def \R  {\mathbb{R}} 
\def \Z  {\mathbb{Z}}

\def \AA {{\mathcal A}}
\def \EE {{\mathcal E}}
\def \FF {{\mathcal F}}
\def \LL {{\mathcal L}}

\def  \RR {{\mathcal R}}
\def \XX {{\mathcal X}}

\def \UU {{\mathcal U}}
\def  \VV {{\mathcal V}}
\def \SS {{\mathcal S}}

\def \la {\lambda}
\def \ind {{\bf 1}}
\def \bx {\bar{x}}
\def \ba {\bar{\alpha}}
\def \bX {\bar{\XX}}

\def \bw {\bar{w}}
\def \by {\bar{y}}

\def \baf {\bar{f}}
\def \bR {\bar{R}}
\def \br {\bar{r}}

\def \bLL {\bar{\LL}}

\def \bK {\bar{K}}
\def \brX {\breve{\XX}}
\def \brx {\breve{x}}
\def \bry {\breve{y}}
\def \brz {\breve{z}}
\def \brf {\breve{f}}
\def \brr {\breve{r}}
\def \brR {\breve{R}}
\def \brK {\breve{K}}

\def \Id {\mbox{Id}}
\def \Tr {\mbox{Trace}}

\def \Span {\mbox{Span}}
\def \equi {\Longleftrightarrow}

\begin{document}

\title[Multiresolution on graphs]{Intertwining wavelets \\ or Multiresolution analysis on graphs \\ 
through  random forests}
\author{Luca Avena} 
\address{Leiden University.}
\email{l.avena@math.leidenuniv.nl}
\today

\author{Fabienne Castell} 
\address{
 Aix Marseille Univ, CNRS, Centrale Marseille, I2M, Marseille, France.}
\email{fabienne.castell@univ-amu.fr}

\author{Alexandre Gaudilli\`ere} 
\address{
 Aix Marseille Univ, CNRS, Centrale Marseille, I2M, Marseille, France.}
\email{alexandre.gaudilliere@math.cnrs.fr} 

\author{Clothilde M\'elot} 
\address{
Aix Marseille Univ, CNRS, Centrale Marseille, I2M, Marseille, France.}
\email{clothilde.melot@univ-amu.fr}

\subjclass[2010]{94A12, 05C81, 05C85, 15A15, 60J20, 60J28}
\keywords{Graph signal processing, multiresolution analysis, wavelet basis, intertwining, Markov processes, 
random spanning forests, determinantal processes.}

\begin{abstract}
We propose a new method for performing  multiscale analysis of functions defined on the vertices of a finite connected weighted graph.
Our approach relies on a random spanning forest to downsample the set of vertices,
and on approximate solutions of Markov intertwining relation to provide a subgraph structure and a filter bank leading to a wavelet basis of the set of functions.
Our construction involves two parameters $q$ and $q'$. The first one controls the mean number of kept vertices in the downsampling,
while the second one is a tuning parameter between space localization and frequency localization.
We provide an explicit reconstruction formula, bounds on the reconstruction
operator norm and on the error in the intertwining relation, and a Jackson-like inequality. These bounds lead
to recommend a way to choose the parameters $q$ and $q'$. We illustrate the method by numerical experiments.   
\end{abstract}

\maketitle

\section{Introduction}
  Graphs provide a flexible representation of geometric structures of irregular domains, and they are now a commonly used tool in numerous applications including neurosciences, social sciences, biology,  transport, 
  communications, etc. One edge between two vertices models an interaction between them, and 
  one can also associate a weight to each edge to quantify levels of interaction, leading to a weighted graph. A graph signal  associates data to each vertex of the graph. As a measure of the brain activity in distinct functional regions, functional magnetic resonance images are an instance of a such a graph signal. Therefore, signal processing on weighted graphs has raised significant interest in recent years (see \cite{SCHU} for a review on the subject).  
    This paper is concerned with the problem of defining a multiresolution analysis for graph signals.
 To introduce the problem, let us recall the now classical method to perform such a multiresolution analysis on 
 a regular domain. 
 
\subsection{Multiresolution analysis  on regular grids.}\label{sec:intro}
 Let us consider a discrete periodic function  $f:\Z_n=\Z/n\Z\rightarrow \R$, viewed as a vector in $\R^n$.  The multiresolution analysis of $f$  is based on wavelet analysis and the so-called multiresolution scheme. Performing the wavelet analysis    
amounts to compute an "approximation" $f_{1}\in \R^{n/2}$ and a "detail" component $g_{1}\in\R^{n/2}$ through classical operations in signal processing such as "filtering" and "downsampling". 
Roughly speaking one wishes the approximation to give the main trends present in $f$ whereas the detail would contain more refined information. This is done by splitting the frequency content of $f$ into two components:  $f_{1}$ focuses on the low frequency part of $f$ whereas the high frequencies in $f$ are contained in $g_{1}$. 
In signal processing and image processing, "filtering" (i.e computing the convolution $f\star k$ for some well chosen kernel $k$) allows to perform such frequency splittings. In our case $K^l(f)=f \star k^l$ yields $K^l(f)$ as a low frequency component of $f$ and 
$K^h(f)=f\star k^h$ yields $K^h(f)$ as a high frequency version of $f$.
The  vectors $f_1$ and $g_1$ are "downsampled" versions of $K^l(f)$ and $K^h(f)$ by a factor of 2, which means that one keeps one coordinate of $K^l(f)$ and  $K^h(f)$ out of two, to build $f_{1}$ and $g_{1}$ respectively. Thus the total length of  the concatenation of the two vectors $[f_{1},g_{1}] $ is exactly $n$, the length of $f$. To sum up we have
\begin{eqnarray*}
f_{1}(\overline{x})&=&K^l(f)(\overline{x})=\langle \varphi_{\overline{x}},f \rangle\\
g_{1}(\overline{x})&=&K^h(f)(\overline{x})=\langle \psi_{\overline{x}},f\rangle
\end{eqnarray*}
where 

\begin{itemize}
\item $\overline{x}$ belongs to the set of downsamples $\overline{\Z_n} $ isomorphic to $\Z/\frac{n}{2}\Z $.
\item $\{\varphi_{\overline{x}},\overline{x}\in \overline{\Z_n}\}$ is the set of functions such that the equality between linear forms $\langle\varphi_{\overline{x}},\cdot \rangle=K^l(\cdot)(\overline{x}) $ holds for all $\overline{x}\in\overline{\Z_n}$.
\item In the same way $\{\psi_{\overline{x}},\overline{x}\in \overline{\Z_n}\}$ is such that  $\langle\psi_{\overline{x}},\cdot \rangle=K^h(\cdot)(\overline{x}) $ holds for all $\overline{x}\in\overline{\Z_n}$.
\end{itemize}
 The choice of $k^l$ and $k^h$ is clearly crucial and is done in such a way that perfect reconstruction of $f$ is possible from $f_{1}$ and $g_{1}$, so that there is no lost information in the representation $[f_{1},g_{1}] $. One can iterate the scheme to produce a sequence $f_{N_0}\in\R^{n/2^{N_0}}, g_{N_0}\in\R^{n/2^{N_0}},...,g_{1}\in\R^{n/2}  $ such that the total length of the concatenated vectors $[f_{N_0},g_{N_0},...,g_{1}]$ is exactly $n$. This is the reason why this scheme yields a multiresolution representation of $f$. 
Remark that the perfect reconstruction condition amounts to have $\mathcal{B}=\{\varphi_{\overline{x}},\psi_{\overline{x}},\overline{x}\in\overline{\Z_n}\} $ a basis for the signals 
$f$ on $\Z_n$. A famous construction by Ingrid Daubechies  \cite{DAUB} derives several families of orthonormal compactly supported such  basis $\mathcal{B} $. These families combine localization in space around the point $\overline{x}$ and localization properties in frequency due to the filtering step they have been built from. 
Using this space-frequency localization one can derive key properties of the wavelet analysis of a signal which rely on the deep links between the local regularity properties of $f$ and the behavior and decay properties of detail coefficients. Let us finish by saying that wavelet type algorithms led to developments of many powerful algorithms of compression, signal restauration, classification, and other kind of time-frequency transforms etc. We will not go  into more detail and refer the interested reader to one of the numerous books on wavelet methods and their applications such as \cite{DAUB} or \cite{MAL}.  \\

\subsection{Signal processing on graphs.}
Our aim is to define a multiresolution analysis on a  generic weighted finite graph $G=(\XX, w, \mu)$,  
 where:
 \begin{itemize}
 \item $\XX$ is the set of vertices of $G$ with cardinality $\va{\XX}=n$;
 \item $w: \XX \times \XX \mapsto \R^+$ is a weight function, which associates to each pair $(x,y)$  a positive weight $w(x,y)$; 
 \item $\mu$ is a measure on $\XX$ such that $w$   is reversible w.r.t. $\mu$, meaning that 
\begin{equation}
  \label{symmetry}
\forall x, y \in \XX \, ,  \, \,  \mu(x) w(x,y) = \mu(y) w(y,x)  \, .
 \end{equation}
 \end{itemize}
The weight function $w$ provides a graph structure on $\XX$, $(x,y)$ being an edge
of the graph iff $w(x,y)>0$. On such a graph, the operations of translation and downsampling are no longer canonically defined and several attempts to tackle these issues and generalize the wavelet constructions have been proposed: see \cite{SCHU} for a review on this subject and \cite{HAM} for one of the most popular methods. We will describe some of them in section \ref{review}, but 
let us stress that when going to the case of a generic 
weighted graph, there are three main problems one has to face:
\begin{description}
\item[(Q1)] What kind of downsampling should we use? What is the meaning of "keep one point out of two"?
\item[(Q2)] On which weighted graph should the approximation $f_1$ be defined to iterate the procedure?
\item[(Q3)] Which kind of filters should one use? What is a good "local" mean?
\end{description}  

This paper proposes a construction of a multiresolution analysis on a weighted graph which is based on a random downsampling method to answer question (Q1),  and on Markov processes interwining to answer questions (Q2) and (Q3).  
Markov processes enter naturally into the game through the
graph Laplacian $\LL$  defined on $\ell_2( \XX,\mu)$ by 
\begin{equation}
 \label{Generateur.def}  \LL f(x) := \sum_{y \in \XX} w(x,y) (f(y)- f(x)), 
 \end{equation}
where $f: \XX \rightarrow \R$ is an arbitrary function, and $w(x,y)$ is now viewed as the transition rate from $x$ to $y$. For $x \in \XX$, let  
\[ w(x) := \sum_{y \in \XX \setminus \acc{x}} w(x,y) \, . 
\]
Note that $\LL$ acts on functions as the matrix, still denoted by $\LL$:
\[ \LL(x,y)=w(x,y)  \mbox { for } x \neq y \, ; \, \, \LL(x,x)=-w(x) \, . 
\]
$\LL$ is the generator of  a continuous time Markov process $X=(X(t), t \geq 0)$, which jumps from $x$ to a neighboring site $y$ at random exponential times (see section \ref{Markov.sec}).  The law of 
$X$ starting from $x$ is denoted by $P_x$, and $E_x$ is the expectation  w.r.t. $P_x$.
Apart from the symmetry assumption \eqref{symmetry}, we will assume throughout the paper that 
\begin{equation} 
\label{irred.hyp}
\LL \mbox{ is irreducible. }
\end{equation}

Under this assumption, $\mu$ is the unique (up to multiplicative constants) invariant measure of the process $(X(t), t \geq 0)$ 
(i.e $\mu \LL=0$), and $\mu(x) > 0$ for any $x \in \XX$.  Since $\XX$ is finite, we can assume without loss of generality that $\mu$ is a probability measure on $\XX$.

\subsection{Our approach for downsampling, weighting and filtering procedures on graphs.}
 The first step in our proposal of a multiresolution scheme consists in answering (Q1) and constructing a random 
 subset  $\bX$ of $\XX$, whose main feature is to be "well spread" on $\XX$. In this respect, we use an adaptation of Wilson's algorithm (\cite{Wi}) studied in \cite{LA}, whose 
output is a random rooted spanning forest $\Phi$ on $G$. Its law $\pi_q$ depends on the weight function $w$ and on a positive real $q$. This algorithm and the properties of the corresponding random forest 
will be described in section \ref{rsf.sec}.  The set of roots of $\Phi$, denoted by $\rho(\Phi)$,
 has the nice property of being a determinantal process on $\XX$ with kernel given by the Green kernel:
\begin{equation}
\label{Kq.def} 
K_q(x,y) := q (q \Id - \LL)^{-1}(x,y)= P_x\cro{X(T_q)=y}  \, , 
\end{equation}
where $T_q$ is an exponential random variable with parameter $q$, independent of the process $X$.
Being a determinantal process, the roots of $\Phi$ tend to repulse each other, and this repulsiveness property results in the noteworthy fact 
that the mean time needed by the 
 process $(X(t), t \geq 0)$ to reach $\rho(\Phi)$, 
 does not depend on the starting point in $\XX \setminus \rho(\Phi)$. Therefore, this set of roots is a natural choice for $\bX$, especially as the 
 parameter $q$ can be tuned  to control the mean number of points in $\rho(\Phi)$. See Section \ref{foret.sec} for more details.

  Once the downsampling $\bX$ is fixed, it remains to answer questions (Q2) and (Q3), i.e. we have to construct a weight function
  $\bw$ on $\bX \times \bX$, and to  define a filtering method which gives the
  approximation $f_1$ and the detail $g_1$ of a given function $f_0$ in $\ell_2(\XX,\mu)$. 
  This will be achieved by looking for a   solution $(\bLL,\Lambda)$
   to the intertwining relation
 \begin{equation}
\label{DF.eq}
 \bLL  \Lambda = \Lambda  \LL \, , 
  \end{equation}
  where 
  \begin{itemize}
  \item $\bLL$ is a Markov generator on $\bX$;
  \item $\Lambda : \bX \times \XX \mapsto \R^+$ is a positive rectangular matrix. 
  \end{itemize}
An intertwining relation gives a natural link between two Markov processes living on different 
state spaces. It appeared in the context of diffusion processes in the paper by Rogers and Pitman \cite{RP}, as  a tool to state identities in laws, and was later successfully applied 
to many other examples (see for instance \cite{CPY},
\cite{MY}, etc). In the context of Markov chains, intertwining was used by Diaconis and Fill \cite{DF}
 to construct strong stationary times, and to control the rate of convergence to equilibrium. 
 At the time being, applications of intertwining include random matrices \cite{DMDY}, particle systems \cite{War}, etc.
 For our purpose, the intertwining relation can be very useful since 
 \begin{itemize} 
 \item $\bLL$ provides a natural choice for the graph structure to be put on $\bX$;
\item each row $\nu_{\bx}:=\Lambda(\bx, \cdot)$ ($\bx \in \bX$) defines a positive measure on $\XX$, which can serve as a "local mean" around
 $\bx$. $f_0$ will be 
  approximated by the function $f_1$  defined on $\bX$ by
 \[ \forall \bx  \in \bX,   f_1(\bx) := \nu_{\bx}(f_0) = \sum_{x \in \XX} \Lambda(\bx,x) f_0(x) \, ;
 \]
 \item it will serve as a basis for getting a Jackson-type inequality (see Proposition \ref{jackson.prop}). 
 \end{itemize}
 To be of any use in signal processing, the "filters"   $(\nu_{\bx}, \bx \in \bX)$  have 
 to be well localized in space and frequency. In the graph context, frequency localization means that the filters belong to an eigenspace of the graph laplacian $\LL$. Hence, 
  we are interested in solutions to \eqref{DF.eq} such that the measures  $(\nu_{\bx}, \bx \in \bX)$ 
  are linearly independent measures tending
  to be non-overlapping (space localization), and contained in eigenspaces of $\LL$.  
  In addition, in order to iterate the procedure, we also need  $\bLL$ to be reversible  on $\bX$. 

Note that  saying that $(\Lambda, \bLL)$ is an exact
  solution to \eqref{DF.eq} implies that the linear space spanned by the measures $(\nu_{\bx}, 
  \bx \in \bX)$ is stable by $\LL$, and is therefore a direct sum of  eigenspaces of $\LL$, so that these measures provide filters which are frequency localized. Hence the error in the intertwining relation is a measure of frequency localization: the smaller the intertwining error, the better the frequency localization. 
Finding solutions to \eqref{DF.eq} is the purpose of \cite{ACGM1}, where an exact linearly independent solution 
to \eqref{DF.eq} is provided. However this solution
tends to be overlapping. \cite{ACGM1} provides also approximate solutions to  \eqref{DF.eq} with small overlapping,
and it is one of these approximate solutions
 we use in our multiresolution analysis scheme. In order to describe it, 
we first consider the trace process of $(X(t), t \geq 0)$ on $\bX$, 
i.e. the process $(X(t), t \geq 0)$ sampled at the 
passage times on $\bX$. This process is a Markov process on $\bX$, whose generator 
$\bLL$ defines  a weight function $\bw$ on $\bX$, symmetric with respect to 
 the measure $\mu$ restricted on $\bX$. We will recall some general facts about
 Markov processes, and make the previous statements  more precise in section \ref{Markov.sec}.

Turning to the rectangular matrix $\Lambda$, it depends on a parameter $q'$ 
and involves the Green kernel $K_{q'}$ (see definition \eqref{Kq.def}). 
The rectangular matrix $\Lambda$ is just the restriction of $K_{q'}$ to 
$\bX \times \XX$. Note that 
\begin{itemize}
\item when $q'$ goes to $0$,  for any $\bx \in \bX$, 
$K_{q'}(\bx,y)$ goes to $\mu(y)$ so that \eqref{DF.eq} is clearly satisfied. $\mu$ being  
the left-eigenvector of $\LL$ corresponding to the eigenvalue $0$, the $K_{q'}(\bx,y)$ are well 
frequency localized. However, the vectors
$(K_{q'}(\bx, \cdot), \bx \in \bX)$ become linearly dependent and very badly space localized. 
\item when $q'$ goes to $\infty$, $K_{q'}(\bx, \cdot)$ goes to $\delta_{\bx}$. Hence, the space localization is perfect. However, the frequency localization is lost, and  the error in \eqref{DF.eq} tends to grow.
\end{itemize}
Hence, a compromise has  to be made concerning the choice of $q'$, and we will discuss this point later. 

To sum up our proposal, two parameters $q$ and $q'$ being fixed, 
\begin{itemize}
\item take  $\bX :=\rho(\Phi)$, $\Phi$ being sampled with $\pi_q$. Set $\brX = \XX \setminus \bX$.
\item $\bLL$ is the generator of the trace process of $X$ on $\bX$. It can be shown (see Lemma \ref{LbarSchur.lem} section \ref{Lbar.sec}) that 
$\bLL$ is irreducible and 
reversible w.r.t. the probability measure $\mu_{\bX}$, which is the measure $\mu$ conditioned to the set $\bX$: 
$\mu_{\bX}(A)=\mu(A \cap \bX)/\mu(\bX)$. $\bLL$ is actually  the Schur complement in $\LL$, of 
$\LL$ restricted to $\brX$, which was already used for instance in \cite{HAM} . For any $\bx$, $\by$ in $\bX$, such that
$\bx \ne \by$, the weight function $\bw$ on $\bX \times \bX$ is then defined by $\bw(\bx,\by)= \bLL(\bx,\by)$.  
\item $\Lambda(\bx, \cdot) := K_{q'}(\bx, \cdot)$.
\end{itemize}
$f$ (${=}f_0$) being a function in $\ell_2(\XX,\mu)$, we define the approximation of $f$ as the function 
$\baf$ (${=}f_1$) defined on $\bX$ by 
\begin{equation}
\label{barf.def}
 \forall \bx \in \bX \, , \, \,  \baf(\bx) = \Lambda f(\bx)= K_{q'}f(\bx) \, , 
\end{equation}
and its detail function as the function  $\brf$ (${=}g_1$) defined on $\brX$ by 
\begin{equation}
\label{brevef.def}
\forall \brx \in \brX \, , \, \,  \brf(\brx) = (K_{q'}-\Id)f(\brx) \, .  
\end{equation}
We proved in \cite{ACGM1} that some localization property of the  $(\nu_{\bx}=\Lambda(\bx,\cdot) , \bx \in \bX)$ defined by \eqref{barf.def} followed from the fact that $\bX$ is a determinantal process with kernel $K_q$.  Since $\brX= \XX \setminus \bX$ is also a determinantal process,
 with kernel $ \Id-K_q$, this suggested the detail definition \eqref{brevef.def}, in which the sign convention is chosen to have a self-adjoint
  analysis operator $U: f \mapsto (\baf, \brf)$ in $\ell_2(\XX,\mu)$. 
  
The process can then be iterated with $(\bX,\bLL,\mu_{\bX})$ in place of $(\XX, \LL,\mu)$. This leads to a multire\-solution scheme. 

\subsection{Description of the results.}
We give now a description of our main results. Note that whatever the choice of a subset $\bX$  of $\XX$, one can still define 
$\bLL$, $\Lambda$, $\baf$ and $\brf$ as we just did it previously. Our first set of results assumes that 
$\bX$ is any subset of $\XX$, and provides a reconstruction 
formula, and bounds on various operators including approximation, detail  and analysis operators.
They are expressed in terms of the hitting time of $\bX$:
  \[
H_{\bX}:= \inf \acc{t \geq 0, X(t) \in \bX } \, .
\] 
To state them, we introduce some notation. First of all, we define  the maximal rate $\alpha > 0$ by 
\begin{equation}
\label{MaxTaux.eq}
 \alpha = \max_{x \in \XX} w(x) \, . 
\end{equation}
Hence, the matrix $P= \Id + \LL / \alpha$ is a stochastic matrix on $\XX$. $\bX$ being any subset of $\XX$, we define the two positive real numbers $\beta$ and $\gamma$ by 
 \begin{equation}
  \label{beta-gamma-def}
   \frac{1}{\gamma} := \max_{\brx \in \brX} E_{\brx}\cro{H_{\bX}} 
     \, , \,\, 
 \frac{1}{\beta} := \max_{\bx \in \bX} \sum_{z \in \XX} P(\bx, z) E_{z}\cro{H_{\bX}} 
  \, .
  \end{equation}
$\ba$ denotes the maximal rate of $\bLL$:
\[ \ba := \max_{\bx \in \bX} - \bLL(\bx,\bx) \, . 
\]
We stress the obvious fact that $\gamma$, $\beta$ and $\ba$ are functions of the subset $\bX$. 

Finally, when 
$(M(u,v), u \in \UU, v \in \VV)$ is a rectangular matrix,  when $U$ is a subset of  $\UU$, and
$V$ a subset of $\VV$, $M_{UV}$ denotes the submatrix of $M$ obtained by 
restricting the entries of $M$ to $U$ and $V$.

To begin with our results, we get a reconstruction formula:

\begin{enonce}{Reconstruction formula}
(Proposition \ref{reconstruction.prop}.) \\
{\it For any $f \in \ell_2(\XX,\mu)$, let $\baf \in \ell_2(\bX,\mu)$ and $\brf \in \ell_2(\brX,\mu)$ be defined by 
\eqref{barf.def} and \eqref{brevef.def}. Then, 
\[ f = \bR \baf + \brR \brf  \, ,  
\]
where $\bR$ are $\brR$ are rectangular matrices  indexed respectively  by $\XX \times \bX$  and $\XX \times \brX$, whose block decompositions are 
\begin{equation}
\label{bR.def}
 \bR = \begin{pmatrix} \Id_{\bX} - \frac{1}{q'} \bLL 
 \\[.2cm]
  (- \LL_{\brX \brX})^{-1} \LL_{\brX \bX}
\end{pmatrix}  
\, , 
\,\, 
\mbox{ and }
\brR 
=\begin{pmatrix}  \LL_{\bX \brX} (-\LL_{\brX \brX})^{-1} \\  - \Id_{\brX} - q' (- \LL_{\brX \brX})^{-1} 
\end{pmatrix}
 \, . 
\end{equation}
}
\end{enonce}
We provide also upper bounds on norms of various operators: the approximation operator $\bR$, the detail 
operator $\brR$, and the interwining error operator $ \bLL \Lambda - \Lambda  \LL$. These bounds are given here 
when these operators are seen as operators from one $\ell_{\infty}$ space to another one. 
They  are stated in greater generality in Propositions \ref{bR-norm.prop}, \ref{brR-norm.prop} and  \ref{intertwiningp.prop}.
\begin{enonce}{Approximation and detail operator norms}
(Propositions \ref{bR-norm.prop} and  \ref{brR-norm.prop}). \\
{\it
Let $\bX$ be any proper subset of $\XX$ (i.e. $\emptyset \subsetneq \bX \subsetneq \XX$), $\brX = \XX \setminus \bX$, and let $\bR$ and 
  $\brR$ be the operators defined in \eqref{bR.def}.  For any   $\baf \in \ell_{\infty}(\bX, \mu_{\bX})$
and any $\brf \in \ell_{\infty}(\brX, \mu_{\brX})$, 
\begin{equation}
\label{normR.eq}
 \nor{\bR \baf}_{\infty,\XX}
   \leq  
\pare{1+ 2 \frac{\ba}{q'}}  \,   \nor{\baf}_{\infty,\bX} 
  \,  ,  \mbox{ and } 
 \nor{\brR \brf}_{\infty,\XX}
   \leq  
\max \pare{\frac{\alpha}{\beta}; 1+ \frac{q'}{\gamma}}  \,   \nor{ \brf}_{\infty,\brX} 
  \,  .
  \end{equation}
  }
   \end{enonce} 
\begin{enonce}{Intertwining error norm} (Proposition \ref{intertwiningp.prop}). \\
{\it
  Let $\bX$ be any proper subset of $\XX$, $\brX = \XX \setminus \bX$.  For any $f \in \ell_{\infty}(\XX,\mu)$, 
  \begin{equation}
  \label{normInter.eq}
  \nor{\pare{\bLL \Lambda - \Lambda \LL} f}_{\infty,\bX} 
  \leq 2 q' \frac{\alpha }{\beta}  \nor{f}_{\infty,\XX}
  \, . 
  \end{equation}
  }
 \end{enonce}
                         
Note that these bounds reflect the competition between $\nor{\pare{\bLL \Lambda - \Lambda \LL} f}_{\infty,\bX}$ (frequency localization) and 
$ \nor{\bR \baf}_{\infty,\XX}$ (space localization). Actually, the term $ \ba/q'$ appearing in \eqref{normR.eq} is a decreasing function of $q'$ and an increasing function of $\bX$, while the term $q' / \beta$ in \eqref{normInter.eq}  is increasing in $q'$ and decreasing in $\bX$. 

Other results are stated in the paper, including a bound on the norm of the detail $\brf$ in terms of the 
norm of $\LL f$ (Proposition \ref{brf.prop}), and a Jackson's type inequality on the approximation error after
$K$ steps of the mutiresolution scheme (Proposition \ref{jackson.prop}). 

The other set of results focuses on the case where $\bX$ is the set of roots of the random forest $\Phi$. They 
provide estimates on its cardinality $\va{\bX}$, and on the quantities $\ba$, $\beta$ and $\gamma$ involved in 
the previous statements. $\bX$ being random, 
all these quantities are random ones, and the estimates we get
are averaged ones. To state them, we will assume that the random variable $\Phi$ is defined  on some probability space $(\Omega_f, \AA_f, \P_q)$. The corresponding expectation will be denoted 
 by $\E_q$. 

 \begin{enonce}{Cardinality estimates} (Proposition \ref{card.prop}).  \\ 
{\it  Let $\bX= \rho(\Phi)$ and $\brX = \XX \setminus \bX$. Let $r \in ]0;1[$, and define 
$\RR_r = \acc{ x \in\XX; w(x) \geq r \alpha}$ the set of "rapid" points.   Then, 
\[  \va{\XX} \frac{q}{q+\alpha} \leq  \E_q\cro{\va{\bX }} \, ; \, \, 
\E_q\cro{\va{\bX \cap \RR_r }}  \leq \va{\RR_r} \frac{q+(2-r) \alpha}{q + 2 \alpha} 
\, . 
\]
}
 \end{enonce}
 
 These estimates can be used to tune the parameter $q$ in order to target a given proportion of kept points. Next, we obtain bounds on the quantities $\ba$, $\beta$ and $\gamma$ when $\bX= \rho(\Phi)$. Unfortunately, these bounds are just 
 lower bounds, and are moreover averaged ones. Nevertheless, since they are expressed in terms of the cardinality of $\rho(\Phi)$, they
 can easily be estimated by Monte Carlo methods, and can serve as a guide for the choice of $q$ and $q'$. These are then the key estimates where our particular random choice for $\bX$  plays a central role. 
 
 \begin{enonce}{Estimates on $\ba$, $\beta$, $\gamma$} 
(Propositions  \ref{alphabar.prop},  \ref{beta.prop} and \ref{gamma.prop}). \\
{\it  Let $\bX= \rho(\Phi)$ and $\brX = \XX \setminus \bX$.
 \begin{equation}
   \E_q\cro{\ba } \geq
 q \E_q \cro{\frac{\va{\brX}}{\va{\bX} + 1}} 
  \, , \, \,  
    \E_q\cro{\frac{1}{\beta} } \geq
   \E_q \cro{\frac{\va{\brX}}{\alpha \va{\bX}}} 
     \, , \, \,  \E_q\cro{\frac{1}{\gamma} } \geq
  \frac{1}{q}  \E_q \cro{\frac{\va{\XX}}{\va{\brX} + 1}} \, . 
  \end{equation}
 }
 \end{enonce}

\begin{enonce}{Discussion on the choice of $q$ and $q'$}
\end{enonce}
Based on these estimates, we argue in section \ref{discussionq.sec} that the parameter $q$ should be chosen 
in $[\theta_1 \alpha; \theta_2 \alpha]$ in order to ensure that the mean number of vertices of the subgraph is at least 
a given proportion ($\theta_1/(\theta_1+1)$) of the size of the original graph, and that a given proportion of the rapid points are decimated. In addition it should  minimize the function $q \mapsto \E_q(\bar{\alpha}/\beta)$, for 
the approximation operator norm and the intertwining error to be small (see \eqref{normR.eq} and  \eqref{normInter.eq}). The Monte-Carlo estimation of this function is computationally  costly, so that we propose to minimize the function $q \in [\theta_1 \alpha; \theta_2 \alpha]
\mapsto   q \E_q\cro{\va{\brX}/(1+\va{\bX})}  \E_q\cro{\va{\brX}/\va{\bX}}$. This is possible in practice since, from \cite{LA}, we can simultaneously sample a whole continuum of forests $\Phi_q$ for $q  \in [\theta_1 \alpha; \theta_2 \alpha]$, each of them with the correct distribution $\pi_q$.
  Once $q$ has been fixed, $\bX$ is sampled according to $\pi_q$, $\ba$ and $\beta$ are computed 
  and $q'$ is chosen equal to $2 \ba \va{\bX}/\va{\brX}$, which will ensure numerical stability of the algorithm (see Section \ref{discussionq.sec}). 
   
\subsection{Related works.}
\label{review}
  Previous authors have explored multiresolution analysis on graphs, and have proposed answers to questions (Q1-3). Far from being exhaustive, we describe some of these, and refer the interested reader to \cite{SCHU} for a more complete state of the art. 
  
  Concerning the downsampling procedure, also referred to as the graph coarsening problem, many approaches have been investigated. 
  To mention a few,  one can try to decompose the graph into bipartite graphs for which the notion of "one every two points" is clear. 
  This is the way followed in \cite{NO,ND}, using either {\it coloring-based downsampling}, or {\it maximum spanning tree}.  
  In \cite{SFV}, the partitioning of $\XX$ into two subsets,
   is based on the sign of the eigenfunction of $- \LL$ associated to the maximal eigenvalue. Finally, the authors of \cite{TB} use a community 
   detection algorithm maximizing the modularity, to partition the original graph into many connected subgraphs of small size.  The vertices of the
   downsampled  graph are the elements of the partition, and are not properly speaking selected points of the original graph. 
 
Turning to the weighting procedure, bipartite graphs designed methods put an edge between two selected nodes, if they share at least a  
neighbor in the  original graph, and the weight  may be proportional to the number of shared neighbors. 
Starting from a bipartite graph, this leads immediately to a non bipartite subgraph, thus the need to "decompose" a graph by 
bipartite ones. The authors of \cite{SFV} rely on the so-called {\it Kron reduction}, which is the same as computing the Schur complement.  In \cite{TB}, the weight between two communities is the sum of the weight of edges linking these communities. 

Finally, as far as the filtering procedure is concerned, various filter banks have been proposed, 
leading sometimes to orthonormal basis, or just frames. The {\it graph wavelet filtering bank} of \cite{NO2} 
is more specifically designed for bipartite graphs, and exploit the specific features of bipartite graph Laplacian.
It produces an orthonormal basis.  The {\it diffusion wavelets} of \cite{CM} are obtained by constructing 
orthonormal basis of the spaces $(V_{k-1} \ominus V_k, k = 1, \cdots, K)$, where $V_k = \Im((-\LL)^{2^k})$, 
and are thus designed to form an orthonormal basis. The {\it spectral graph wavelets} of \cite{HAM} are 
 functions of type $g(-t \LL) \delta_n$ where $g$ is a function localized around 0 for the low-pass filters, or away from  $0$ for the high-pass ones. Different values of $t$ are used to select different frequencies, thus leading to a frame. Since the computation of   $g(-t \LL)$ requires the knowledge of the spectral decomposition of 
 $\LL$, polynomial approximations are performed. In \cite{TB}, the filters used are the eigenfunctions of the 
 Laplacians restricted to each small community, naturally leading to an orthonormal basis. In \cite{GBC}, assuming that we are given a  
 subsampling procedure encoding the geometry of the original graph, the authors construct a Haar basis adapted to this subsampling.
 
 Compared to these works, our downsampling approach through Wilson's algorithm, is a partitioning method
 in the spirit of \cite{TB}, a community corresponding to a tree of the forest.  The mean number of selected points can be adjusted through the parameter $q$.   
 
   The filters we use are a special case of spectral graph wavelets of \cite{HAM}, the function $g$ being equal to 
 $1/(1+x)$ for the scaling function, and $x/(1+x)$ for the wavelet, and $t$ corresponding to $1/q'$ (note that 
 $K_{q'} = q'(q' \Id - \LL)^{-1}= (\Id - \LL/q')^{-1}$). This choice is very natural in the signal processing 
 variational approach, since $K_{q'}f$ minimizes the function $g \mapsto \bra{g,-\LL g} + q' \nor{g-f}^2$.   Since 
 we use only one value of $t$ at each step, the filter bank we obtain is a basis, which is not orthogonal, 
 but tends to be so when $q' \gg 1$ (in this case $K_{q'}(\bx, \cdot) \simeq \delta_{\bx}$, while 
 $(K_{q'}- \Id)(\brx, \cdot) \simeq \LL(\brx, \cdot)/q'$). 
  
  The weighting procedure  through Kron reduction has also already be used in \cite{SFV}. But in our approach, 
  weighting and filtering are linked together through the intertwining relation. 

  To sum up, the advantages of our approach are: \begin{itemize}
 \item to provide a new partitioning method, allowing to tune the mean number of kept vertices;
 \item to link the subgraph structure and the choice of the filters through the intertwining relation;
 \item to produce a filter bank leading to a  basis;
 \item to mimic the steps of a classical wavelet analysis algorithm;
 \item to allow the computation of various error bounds;
 \item to suggest a systematic method to choose the parameters $q$ and $q'$, and ensure numerical stability;
 \item to have a computational complexity similar to already existing methods. Actually, apart from the 
 sampling of the random forest which is of order $\alpha \va{\XX}/q$, our approach only requires the computation of $K_{q'}$
  and of the Schur complement, already present in \cite{HAM, SFV}. Starting from a sparse matrix $\LL$, $K_{q'}$ can be efficiently approximated
  by polynomials of small order. Due to the good repartition property of our random $\bX$, such a polynomial approximation can  also be implemented  for the  inversion involved in the computation
   of $\bLL$. This results in a sparse $\bLL$, however the level of sparsity  obtained is not enough to  go on with the algorithm for large graphs. Hence, as in \cite{SFV},  a sparsification step can be added after the computation of  $\bLL$. 
   Our proposition for this sparsification will be guided by the intertwining error bounds that provide a Jackson-like inequality (see Section \ref{resume}). 
 \end{itemize}
 Finally, we would like to mention that our multiresolution scheme construct a basis of the space of signals, which can be used to analyze, compress, etc, any signal on the graph. When one wants to handle just one specific signal,   adaptative multiresolution schemes as in 
 \cite{ED} may be more appropriate.
  
\subsection{Organization of the paper}
We begin in section \ref{notation.sec} with notations used throughout the paper. Section \ref{foret.sec} is devoted to the description and the properties of the random forest, and to the downsampling procedure. 
The weighting and the filtering procedures, $\bX$ being any subset of $\XX$, are  discussed in section
\ref{weight-filter.sec}. We prove in this section the bounds on operator norms. We go on with the iterative scheme and Jackson's inequality in section \ref{jackson.sec}. The discussion on the choice of $q$ and $q'$ and the estimates on $\ba$, $\beta$ and $\gamma$, are given in section \ref{choixq.sec}. We 
summarize  the pyramidal algorithm and discuss computational issues in section \ref{resume}, and we finally end with numerical experiments  illustrating the method in section \ref{numeric.sec}.

   
\section{Notations and preliminary results.}
\label{notation.sec}
In this section, we give the notations used throughout the paper, and we state some useful results. 

 \subsection{Sets of functions and measures.}
A function $f$ on $\XX$ will be seen as a column vector, whereas a signed measure on $\XX$ will be seen as a 
row vector. For $p \geq 1$,  $\ell_p(\XX,\mu)$ is the space of functions on $\XX$ 
endowed with the norm 
\[ \nor{f}_p = \pare{\sum_{x \in \XX} \va{f(x)}^p  \mu(x)}^{1/p}\, .
\]
The scalar product of two functions $f$ and $g$ in $\ell_2(\XX,\mu)$ is 
\[ \bra{f,g} = \sum_{x \in \XX} f(x) g(x) \mu(x) \,  
\] 
When $f$ is a function on $\XX$, $f^{\star}$ will denote the signed measure whose density w.r.t $\mu$ is $f$:
$f^{\star}(A)= \sum_{x \in A} \mu(x) f(x)$ for all subset $A$ of $\XX$. Similarly, when $\nu$ is a signed measure
on $\XX$, $\nu^{\star}$ is the density of $\nu$ w.r.t $\mu$: $\nu^{\star}(x)= \nu(x)/\mu(x)$ for any $x \in \XX$.

\subsection{Schur complement}
Let $M$ be a  matrix of size $n=p+r$ and let 
\[ M=  \begin{pmatrix} A & B \\ C & D \end{pmatrix} \]
be its block decomposition, $A$ being a square matrix of size $p$ and $D$ a square matrix of size $r$. If $D$ is
invertible, the Schur complement of $D$ in $M$ is the square matrix of size $p$ defined by
\[ S_M(D) := A - B D^{-1} C  \, .
\] 

We remind the reader the following standard results concerning the Schur complement (see for instance \cite{Zhang}):
\begin{prop}
\label{schur.prop}
Assume that $D$ is invertible.
\begin{enumerate} 
\item $M=\begin{pmatrix} \Id & B \\ 0 & D \end{pmatrix}
				\begin{pmatrix} S_M(D) & 0 \\  D^{-1}C & \Id \end{pmatrix}$.
\item $\det(M) = \det(D) \det(S_M(D))$.
\item $M$ is invertible if and only if $S_M(D)$ is invertible. In that case, 
\[ M^{-1} = \begin{pmatrix} S_M(D)^{-1} & - S_M(D)^{-1} B D^{-1}
					\\ - D^{-1} C S_M(D)^{-1}  & D^{-1} + D^{-1} C  S_M(D)^{-1} B D^{-1}
					 \end{pmatrix} \, .
\]
\item Assume that $M$ defines a positive symmetric operator in $\ell^2(\XX,\mu)$. 
 Let $\lambda_{max}(M)$ (respectively  $\lambda_{min}(M)$)  be the largest (respectively the smallest) eigenvalue of $M$. Then, $S_M(D)$ is also positive symmetric and 
\[\lambda_{min}(M) \leq  \lambda_{min}(S_M(D)) \leq   \lambda_{max}(S_M(D)) \leq \lambda_{max}(M) \, .
\]
\end{enumerate}
\end{prop}

\subsection{Markov process.} 
\label{Markov.sec}
Consider  an  irreducible  continuous time Markov process
$(X(t), t \ge 0)$ on  $\XX$,  with generator 
 $\LL$ given by \eqref{Generateur.def} satisfying \eqref{symmetry} and \eqref{irred.hyp}. 
 We recall the definition \eqref{MaxTaux.eq} of $\alpha > 0$ as the maximal rate, and that 
 of the matrix $P:= \LL/\alpha+\Id$. $P$ is an irreducible stochastic matrix, 
and we denote by $(\hat{X}_k, k \in \N)$ a discrete time Markov
chain with transition matrix $P$. The process $(X(t), t \ge 0)$ can be constructed 
from $(\hat{X}_k, k \in \N)$ and an 
independent Poisson process $(\tau_i, i \geq 0)$ on $\R^+$ with rate $\alpha$. At each event of the Poisson process,  $X$ moves according to the trajectory of $\hat{X}$:
\[ X(t) = \sum_{i=0}^{+\infty} \hat{X}_i \ind_{\tau_i \leq t < \tau_{i+1}} \, . 
\]

By \eqref{symmetry} and \eqref{irred.hyp},  $\mu$ is strictly positive. 
In addition, $- \LL$ defines a positive symmetric operator on 
$\ell_2(\XX,\mu)$, and we denote by $(\lambda_i; i=0, \cdots, n-1)$
the real eigenvalues of $-\LL$ in increasing order. It follows from  the fact that $P$ is  irreducible
 that 
\begin{equation}
\label{vp.eq}
 0 = \la_0 < \la_1 \leq \la_2 \cdots \leq \la_{n-1} \leq  2 \alpha \, .
\end{equation}
The right eigenvector of $-\LL$ associated to $\la_i$ is denoted by $e_i$:
\[ -  \LL e_i = \la_i e_i \, . 
\]
The $(e_i, i=0, \cdots , n-1)$ are normalized to have an $\ell_2$- norm equal to 1, and form an orthonormal basis of $\ell_2(\XX,\mu)$. By construction, $e_0=\ind$. 
 For any subset $A$ of $\XX$, $H_A$ is the hitting time of $A$ by the process $X$:
\[
H_A := \inf \acc{t \geq 0, X(t) \in A } \, ,  
\]
while $H^+_A$ is the return time to $A$ by the process $X$:
\begin{equation}
\label{return.def}
H^+_A := \inf \acc{t \geq \tau_1, X(t) \in A } \, . 
 \end{equation}


\section{Random spanning forests, Wilson's algorithm and downsampling procedure.}
\label{foret.sec}
We spend now some time on  the description and the properties of the spanning random forest $\Phi$ used
in the downsampling procedure .   Let us call $\EE$ the set of unoriented edges 
of $G$, that is the set of pairs $\acc{x,y}$ such that $w(x,y) > 0$ (and $w(y,x)>0$).  A spanning {\it unoriented} forest is a graph without cycles, with $\XX$ as set of vertices, and  a subset of $\EE$ as edge set. An unoriented tree is a connected component of such a  forest. By choosing in each tree one specific vertex, which we call root,  we define a rooted spanning forest.  Note that the number of roots is the same as the number of trees. Orienting each edge of a tree toward its root,  we obtain an spanning oriented forest (s.o.f) 
$\phi$. The set of roots of a spanning oriented forest $\phi$ is denoted by $\rho(\phi)$. 
If $e=(x,y)$ is an oriented edge, we will use $w(e)$ for $w(x,y)$, and say that $e \in \phi$ if $e$ is an edge of a tree
of the forest $\phi$.

\subsection{A probability measure on forests.}

  We introduce now a real parameter $q >0$, and associate to each oriented forest a weight
 \begin{equation}
 \label{poidsforet.def}
 w_q(\phi) := q^{|\rho(\phi)|} \prod_{e \in \phi} w(e) \, ,
 \end{equation}
where $|\rho(\phi)|$ is the cardinality of $\rho(\phi)$, i.e. the number of trees in the forest 
$\phi$. $w_1(\phi)$ will be denoted by $w(\phi)$ so that 
\[  w_q(\phi) := q^{|\rho(\phi)|} w(\phi) \, . 
\]
These weights can be renormalized to define a probability measure on the set of spanning oriented forest:
 \begin{equation}
 \label{probaforet.def}
 \pi_q(\phi) := \frac{ w_q(\phi) }{Z(q)} \, ,
 \end{equation}
 where the partition function $Z(q)$ is given by
 \begin{equation}
 \label{partition.def}
Z(q):= \sum_{\phi \mbox{ s.o.f.}}  w_q(\phi) \, .
\end{equation}

\subsection{Wilson's algorithm.}
A way to sample a random s.o.f. $\Phi$ from $\pi_q$ is given by the following iterative algorithm. 
Let $\Phi_c$ be the current state of the forest being constructed, and let $\VV_c$ be the set of vertices of
$\Phi_c$. 
At the beginning, $\VV_c $ is equal to $\emptyset$. \\
While $\XX \setminus \VV_c \neq \emptyset$, perform the following steps:
\begin{itemize}
\item Choose a point  $x$  at random in $\XX \setminus \VV_c$. 
\item  Let evolve the Markov process $(X(t), t \geq 0)$ from $x$, and stop it either when it 
reaches $\VV_c$, or after an independent exponential time  of parameter $q$. 
\item  Erase the loops of the trajectory drawn by $X$. We obtain a self-avoiding path $C$
starting from $x$ and oriented towards its end-point.
\item Add $C$ to $\Phi_c$.
\end{itemize}
Each  iteration of the "while loop" stopped by the exponential time, gives birth to another tree. 
Wilson's algorithm is not only a way to sample $\pi_q$. It provides also a powerful tool to analyse it,
the reason being that it does not depend on the way we "choose a point" in the first step of the "while" loop. 
In addition there exists a coupling of the probability measures $\pi_q$ for different values of $q$. This 
means that we can construct the random forests for different values of $q$ from the same set of random
variables. This coupling is explained in \cite{LA}.

 \subsection{Properties of the random forest.}
 \label{rsf.sec}
   Using Wilson's algorithm, and the explicit knowledge of the law of a loop-erased Markov process, the following 
   statement,  as all the results  stated in this section,   is proven in \cite{LA}: 
\begin{prop} Partition function. \\
The partition function $Z(q)$ is the characteristic polynomial of $\LL$:
\begin{equation}
\label{part.exp}
Z(q) = \det (q \,  \Id - \LL) \, .
\end{equation}
\end{prop}
Some other important features of the random s.o.f. $\Phi$ are listed below.
 \begin{prop} Number of roots. \\
\label{nombre.prop}
   For all $k \in \acc{0, \cdots, n}$, 
 \[ \P_q\cro{\abs{\rho(\Phi)}=k} =
  \sum_{\begin{matrix}\scriptstyle{ I \subset \acc{0, \cdots, n-1}}\\[-2pt]  \scriptstyle{  \abs{I}=k}  \end{matrix} }
\,  \prod_{i \in I} \frac{q}{q+\lambda_i} \,\,   \prod_{i \notin I} \frac{\lambda_i}{q+\lambda_i} \, . 
 \]
Otherwise stated, the number of roots has the same law as $\sum_{i=0}^{n-1} B_i$ where 
 $B_0, \cdots, B_{n-1}$ are independent, $B_i$ having Bernoulli distribution with parameter $\frac{q}{q+\lambda_i}$.
  \end{prop} 
 Note that since $\lambda_0=0$, $B_0 = 1$ a.s. and we recover the fact that $ \abs{\rho(\Phi)} \ge 1$ a.s.. 
 
 \def \det {\mbox{det}} 
 \begin{prop} 
 \label{racines.prop}
 Set of roots. \\
 $\rho(\Phi)$ is a determinantal process on $\XX$ with kernel $K_q$ (see definition in \eqref{Kq.def}), i.e.
  for any subset  $A$ of $\XX$,
  \[ \P_q( A \subset \rho(\Phi)) = \det_{A}(K_q) \, , 
  \]
  where $\det_A$ is the minor defined by the rows and columns corresponding to $A$.
 \end{prop}
  The following statement involves two independent sources of randomness, the Markov 
 process $X$ starting from $x$ defined on $(\Omega, \AA, P_x)$, and the random forest  $\Phi$ defined on 
 $(\Omega_f,\AA_f,\P_q)$. Integration on the product space  $(\Omega \times \Omega_f, \AA \otimes \AA_f)$
  w.r.t. the product measure $\P_{x,q}:=P_x \otimes \P_q$ is denoted by $\E_{x,q}$.
    \begin{prop} Hitting time of $\rho(\Phi)$.
  \label{hitting.prop}
  \begin{enumerate} 
  \item  For any $x \in \XX$, 
  \begin{equation}
  \E_{x,q}\cro{H_{\rho(\Phi)}} = \frac{\P_q\cro{\abs{\rho(\Phi)}>1}}{q} \, .
  \end{equation}
 \item For any $x \in \XX$, and $m \in \acc{1, \cdots, n}$,
   \begin{equation}
  \E_{x,q}\cro{H_{\rho(\Phi)}|_{ \abs{\rho(\Phi)}=m}} 
  = \frac{\P_q\cro{\abs{\rho(\Phi)}=m+1}}{q \P_q\cro{\abs{\rho(\Phi)}=m}} \, .
  \end{equation}
  \end{enumerate}
  \end{prop}
 Note that these expressions do not depend on the starting point $x$, saying that in some sense, the
 roots of the random forest are "well spread" on $\XX$. 
 
\subsection{The downsampling procedure.}
In view of  the results of section \ref{rsf.sec}, $\bX:=\rho(\Phi)$  is a natural candidate as a 
downsampling of $\XX$. We give now estimates on the mean of $\va{\rho(\Phi)}$ which unlike Proposition \ref{nombre.prop},
do not depend on the knowledge of the eigenvalues $\la_i$. To this purpose, we introduce the mean value of the eigenvalues of 
$-\LL/\alpha$:
\[
m(\LL) := \frac{1}{\va{\XX}} \sum_{i=0}^{\va{\XX} - 1} \frac{\lambda_i}{\alpha} = \frac{\Tr(-\LL)}{\alpha \va{\XX}}\ \, . 
\]
Note that by \eqref{vp.eq}, $m(\LL) \in ]0,2]$.
We get then 
 \begin{prop} 
 \label{card.prop}
 Let $\bX= \rho(\Phi)$ and $\brX = \XX \setminus \bX$. Then
 \begin{enumerate}  
 \item $\E_q\cro{\va{\bX }} \geq \va{\XX} \frac{q}{q+\alpha} $ . 
  \item $\E_q \cro{\va{\brX}}  \geq  \va{\XX} \frac{ \alpha}{q + 2 \alpha} m(\LL)$. 
\item  For any $r \in [0,1]$, set $\RR_r := \acc{x \in \XX \mbox{ s.t. } w(x) \geq r \alpha}$. 
\[ \E_q\cro{\va{\brX \cap \RR_r}} \geq \va{\RR_r} \frac{r \alpha}{q + 2\alpha} \, . 
\]
 \end{enumerate} 
 \end{prop}
\noindent
{\it Remark:} Hence, in a loose sense, $\bX$ contains a given proportion of points in $\XX$, and  $\brX$  contains a given proportion of the "rapid" points in $\XX$, i.e. points for which the rate of escape is high.  
Since they do not depend on the spectral decomposition of $\LL$, these estimates can be helpful
concerning the choice of $q$. Taking for instance $q \in [2\alpha/3;\alpha]$ ensures that the mean proportion of
sampled points is greater than $2/5$, and that  the mean proportion of decimated points is greater than
$m(\LL)/3$.  

\medskip

The proof of Proposition \ref{card.prop} relies on the following lemma.
\begin{lem}
\label{Kq.lem}
For any $x \in \XX$, and any $q>0$,
\begin{equation}
\label{min-Kq.eq}
K_q(x,x) \geq \frac{q}{q+w(x)} \, .
\end{equation}
\begin{equation}
\label{maj-Kq.eq}
1-K_q(x,x) \geq  \frac{w(x)}{q+\lambda_{n-1}} \geq \frac{w(x)}{q+2\alpha} \, .
\end{equation}
\end{lem}
\begin{proof}
Remind that $(e_i,i=0, \cdots n-1)$ is an orthonormal basis
of eigenfunctions of  $-\LL$ in $\ell_2(\XX,\mu)$, and let $(\ind_x, x \in \XX)$ be the canonical basis of 
$\R^{\XX}$ ($\ind_x(y) = 0$ for $y \ne x$ and $\ind_x(x)=1$). Note that 
\[
w(x)=- \LL(x,x) = - \LL \ind_x (x)=\sum_{i=0}^{n-1} \lambda_i \bra {e_i,\ind_x} e_i(x)
= \sum_{i=0}^{n-1} \lambda_i e_i(x)^2 \mu(x) \, .
\]
In the same way, 
\[ K_q(x,x) = \sum_{i=0}^{n-1} \frac{q}{q+\lambda_i} e_i(x)^2 \mu(x)
\, . 
\]
Note that for any $x \in \XX$, 
\[ \sum_{i=0}^{n-1} e_i(x)^2 \mu(x) =  \sum_{i=0}^{n-1} \bra{\ind_x,e_i}^2 \frac{1}{\mu(x)}
= \frac{\nor{\ind_x}^2}{\mu(x)} = 1 
\, ;
\]
Hence, for any $x \in \XX$, $(e_i^2(x) \mu(x); i=0, \cdots, n-1)$ is a probability  distribution
on $\acc{0,1, \cdots, n-1}$. The function $\lambda \in \R^+ \mapsto q/(q+\lambda)$ 
being convex, it follows 
from Jensen's inequality, that for any $x \in \XX$, 
\[ K_q(x,x) \geq \frac{q}{q + \sum_{i=0}^{n-1} \lambda_i e_i(x)^2 \mu(x)} = \frac{q}{q+w(x)}
\, . 
\]
Moreover,  
\begin{align*}
1-K_q(x,x) & = \sum_{i=0}^{n-1} \frac{\lambda_i}{q+\lambda_i} e_i(x)^2 \mu(x)
\, ,
\\
& \geq  \frac{1}{q+\lambda_{n-1}} \sum_{i=0}^{n-1} \lambda_i e_i(x)^2 \mu(x)
\, ,
\\
& = \frac{w(x)}{q+\lambda_{n-1}}
\, .
\end{align*}
This ends the proof of Lemma \ref{Kq.lem}, since $\lambda_{n-1}\leq 2 \alpha$.

\end{proof}
We return now to the proof of Proposition \ref{card.prop}. 
\begin{proof}
\begin{enumerate}
\item It follows from Proposition \ref{racines.prop} that 
$ \E_q\cro{\va{\bX}} = \sum_{x \in \XX} \P_q\cro{x \in \rho(\Phi)}
= \sum_{x \in \XX} K_q(x,x)$. Using \eqref{min-Kq.eq} and the definition of $\alpha$, we obtain 
that $ \E_q\cro{\va{\bX}} \geq \frac{q}{q+ \alpha} \va{\XX}$.
\item By Proposition \ref{nombre.prop}, 
$ \E_q\cro{\va{\brX}}  = \sum_{i=0}^{n-1} \frac{\la_i}{q+\la_i} \geq \frac{1}{q+\la_{n-1}} \Tr(-\LL)$ .
\item By \eqref{maj-Kq.eq}, $\E_q\cro{\va{\brX \cap \RR_r}} = \sum_{x \in \RR_r} \P_q\cro{x \notin \rho(\Phi)} = \sum_{x \in \RR_r} (1-K_q(x,x))
\geq \sum_{x \in \RR_r} \frac{w(x)}{q+ \la_{n-1}}   \geq \va{\RR_r} \frac{r \alpha}{q + \la_{n-1}}$. 
\end{enumerate}
\end{proof}


\section{The weighting and filtering  procedures. }
\label{weight-filter.sec}
 This section aims to describe the weighting and the filtering procedures we use in the multiresolution scheme,
 once $\bX$ has been chosen. Hence, we will assume 
 throughout this section that $\bX$ is any  proper subset of $\XX$. 

\subsection{The weighting procedure.}
\label{Lbar.sec}
   Let us  define for any  $\bx,\by \in \bX$, 
 \begin{equation}
 \label{Pbar.def}
 \tilde{P}(\bx,\by):=P_{\bx}\cro{X(H_{\bX}^+)=\by} \, . 
 \end{equation}
 $\tilde{P}$ is a stochastic matrix on $\bX$, and one can associate to $\tilde{P}$ a weight function 
 $\bar{w}$ and a Markov generator
 $\bLL$ on $\bX$, which are unique up to multiplicative constants. The next 
 lemma explains how to compute $\bLL$ from $\LL$. 
  
 \begin{lem} 
 \label{LbarSchur.lem} Let $\bX$ be any  proper subset of $\XX$, and set $\brX=\XX \setminus \bX$. 
 Let  $\tilde{P}$ be the stochastic matrix on $\bX$  defined by \eqref{Pbar.def}. 
Then,  $\LL_{\brX \brX}$ is invertible. Let $\bar{\LL}$ be the Schur complement of $\LL_{\brX \brX}$ in $\LL$. $\bar{\LL}$ is an irreducible Markov generator on $\bX$, 
 $\tilde{P} = \Id + \bLL/\alpha$, and for any $\bx, \by \in \bX$, 
 $\mu(\bx) \bLL(\bx,\by) = \mu(\by)  \bLL(\by,\bx) $. 
 \end{lem}
 
 \begin{proof}
 Let $Q(x, \bx) := P_x \cro{X(H^+_{\bX}) = \bx}$, so that $\tilde{P}=Q_{\bX \bX}$. Using Markov property at time $\tau_1$
 leads to
 \[ Q(x, \bx) = \sum_{z \in \XX} P(x,z) P_z \cro{X(H_{\bX}) = \bx}
 = P(x, \bx) + \sum_{z \in \brX} P(x,z) Q(z,\bx). \] 
 Hence, since $P= \Id + \LL/\alpha$,
\begin{equation}
\label{LbarSchur.eq}
\left\{
 \begin{array}{lcl} 
  Q_{\bX \bX} & = & P_{\bX \bX} +  P_{\bX \brX} Q_{\brX \bX} 
 \\
 Q_{\brX \bX} & = & 
 P_{\brX \bX} +  P_{\brX \brX} Q_{\brX \bX} 
 \end{array}
 \right.
 \Longleftrightarrow  
 \left\{
 \begin{array}{rcl} 
\tilde{P} & = & P_{\bX \bX} +  P_{\bX \brX} Q_{\brX \bX} 
 \\
-  \LL_{\brX \brX} Q_{\brX \bX} & = & \LL_{\brX \bX} 
\end{array}
\right. 
\, .
\end{equation}
We are now going to prove that 
$- \LL_{\brX \brX}$ is invertible. Let $\phi_{\brX}$ 
be a vector of $\mbox{Ker}(- \LL_{\brX \brX})$, and let $\psi:= (0_{\bX}, \phi_{\brX})$. 
We get on one hand $\bra{\psi, -\LL \psi}=0$. On the other hand,  
$\bra{\psi, -\LL \psi} = \sum_{i=0}^{n-1} \lambda_i \bra{\psi,e_i}^2$. 
Therefore, for any $i \in  \acc{0, \cdots,n-1}$, $\lambda_i \bra{\psi,e_i}^2$. 
$\lambda_0=0$ being simple, $ \bra{\psi,e_i} = 0$ for all $i \geq 1$. Hence, 
$\psi$ is colinear to $e_0 =\ind$. If $\brX \neq \emptyset$ and $\bX \neq \emptyset$, this implies that 
$\phi_{\brX} \equiv 0$. Hence $- \LL_{\brX \brX}$ is invertible. 
Going back to \eqref{LbarSchur.eq}, we obtain that 
 \begin{equation}
 \label{Q.eq}
 Q_{\brX \bX} =   (-  \LL_{\brX \brX})^{-1} \LL_{\brX \bX} \, , 
 \end{equation}
and  \begin{equation}
  \tilde{P} =  \Id_{\bX} + \frac{1}{\alpha} \LL_{\bX \bX} - \frac{1}{\alpha} \LL_{\bX \brX}  
  			( \LL_{\brX \brX})^{-1} \LL_{\brX \bX} = \Id_{\bX} + \frac{\bLL}{\alpha} \, .
\end{equation}
Therefore, $\forall \bx, \by \in \bX$, $\bx \neq \by$, $\bLL(\bx,\by) = \alpha \tilde{P}(\bx, \by) \geq 0$, 
and $\bLL \ind_{\bX}= 0$. $\bLL$ is thus a Markov generator on $\bX$. 
The fact that $\bLL$ is reversible w.r.t. $\mu$ is a direct consequence of the fact that 
$\bLL$ is a Schur complement, and of the reversibility of $\LL$ w.r.t. $\mu$. Concerning the irreducibility statement, if $\bx$ and $\by$ are two distinct elements of $\bX$, the irreducibility of $\LL$ implies that there exists a path in $\XX$ going from 
$\bx$ to $\by$ with positive probability for  $P$. Sampling this path on the passage times on $\bX$, we obtain a path in $\bX$ going from $\bx$ to $\by$ with positive probability for $\tilde{P}$.  
\end{proof}
 
 Once $\bLL$ is defined, $\bw$ and $\ba$ are defined in the same way that $w$ and $\alpha$ were defined from $\LL$, i.e.
 \begin{equation}
 \label{bar-alpha.eq}
  \forall \bx \in \bX \, , \, \, \bw(\bx) := - \bLL(\bx,\bx) \, ; \, \, 
 \ba := \max_{\bx \in \bX}  \bw(\bx) \, . 
 \end{equation}
 
 \subsection{The filtering procedure.}
 \label{filtering.sec}
We have now   to define the "approximation" and the "detail"
of a given function $f $ on $\XX$ , and to say how to reconstruct $f$ from the approximation and the detail. This is the aim of this section.

\subsubsection{The analysis operator.}
Let $q' > 0$ be fixed and define the collection $(\nu_{\bx}, \bx \in \bX)$ of probability measures
on $\XX$ by 
\[ \nu_{\bx}(x) = K_{q'}(\bx,x) \, . 
\]
  Let $(\rho_{\brx}, \brx \in \brX)$ be the collection  of signed measures on $\XX$ defined by 
\begin{equation}
\label{rho.def}
\rho_{\brx}(x) = (K_{q'}-Id)(\brx, x) \, . 
\end{equation}
We associate to the collection of measures $(\nu_{\bx}, \rho_{\brx})_{\bx \in \bX, \brx \in \brX}$ the corresponding
collection of functions $(\phi_{\bx}, \psi_{\brx})_{\bx \in \bX, \brx \in \brX}$:
\begin{equation}
\label{ondelettes.def}
\phi_{\bx}(x) := \nu_{\bx}^{\star}(x)=\frac{\nu_{\bx}(x)}{\mu(x)}  \, \, ;
 \psi_{\brx}(x) := \rho^{\star}_{\brx}(x)  = \frac{\rho_{\brx}(x)}{ \mu(x)} \, . 
\end{equation}
The functions $(\psi_{\brx}, \brx \in \brX)$ play the role of "wavelets" in our proposition of a multiresolution analysis.  Note that the functions $\phi_{\bx}$ are naturally normalized in $\ell_1(\XX,\mu)$, but this is not the case for the functions  $\psi_{\brx}$ ($\nor{ \psi_{\brx}}_1 = 2 (1-K_{q'}(\brx,\brx))$). Note also that
for any  $\brx \in \brX$,  $ \bra{\psi_{\brx},\ind} = 0$.
In addition, we get the following result:
\begin{lem}
The family $(\phi_{\bx}, \psi_{\brx})_{\bx \in \bX, \brx \in \brX}$ is a basis of $\ell_2(\XX,\mu)$.
\end{lem}
\begin{proof}
Since $\mu(x) > 0$ for all $x \in \XX$, it is equivalent to prove
that the matrix $M$ whose row vectors are given by $(\nu_{\bx}, \rho_{\brx})_{\bx \in \bX, \brx \in \brX}$ is invertible.
$M$ can be rewritten in terms of its block matrices according to the rows and columns indexed by $\bX$ and 
$\brX$:
\[ M = \begin{pmatrix} K_{q'}|_{\bX,\bX} & K_{q'}|_{\bX, \brX} \\
						K_{q'}|_{\brX, \bX} & K_{q'}|_{\brX, \brX} - \Id_{\brX} \end{pmatrix} \, .
\]
Note that the eigenvalues of $K_{q'}-\Id$ are the non positive real numbers $(-\frac{\lambda_j}{q'+ \lambda_j}; j=0, \cdots 
n-1)$. Therefore $\lambda_{max}(K_{q'}|_{\brX, \brX} -\Id_{\brX}) \leq \lambda_{max}(K_{q'} -\Id)= 0$. Equality holds
iff there exists an eigenfunction of $K_{q'} -\Id$ associated to the eigenvalue 0, which is contained in 
$\Span(\ind_{\acc{\brx}}; \brx \in \brX)$. But, 0 is a simple eigenvalue of $K_{q'} -\Id$  associated to the function $\ind$. 
Since $\brX \neq \XX$, we get that $\lambda_{max}(K_{q'}|_{\brX, \brX} -\Id_{\brX}) < 0$. 
Thus, $K_{q'}|_{\brX, \brX} -\Id_{\brX}$  is invertible and by Proposition \ref{schur.prop}, 
\begin{equation}
\label{detM.eq}
 \det(M) = \det(K_{q'}|_{\brX, \brX} - \Id_{\brX}) \, \det(S_{M}(K_{q'}|_{\brX, \brX} - \Id_{\brX})) \, .
\end{equation}

Concerning the Schur complement  in \eqref{detM.eq}, note that 
\begin{align*}
S_{M}(K_{q'}|_{\brX, \brX} - \Id_{\brX}) 
& = K_{q'}|_{\bX,\bX} - K_{q'}|_{\bX, \brX} (K_{q'}|_{\brX, \brX} - \Id_{\brX})^{-1} K_{q'}|_{\brX, \bX} 
\\ & = \Id_{\bX} + S_{K_{q'}-\mbox{\scriptsize{Id}}}(K_{q'}|_{\brX, \brX} - \Id_{\brX}) \, .
\end{align*}
$K_{q'}-\Id$ is a symmetric negative operator in $\ell^2(\XX, \mu)$ such that  $\lambda_{min}(K_{q'}-\Id)= - \frac{\lambda_{n-1}}{q'+\lambda_{n-1}} > -1$. Using Proposition \ref{schur.prop}, we deduce that 
\[ \lambda_{min}(S_{K_{q'}-\mbox{\scriptsize{Id}}}(K_{q'}|_{\brX, \brX} - \Id_{\brX}))
	\geq \lambda_{min}(K_{q'}-\Id) > - 1 \, .
\]
Hence $\det(S_{M}(K_{q'}|_{\brX, \brX} - \Id_{\brX})) > 0$. We have thus proven that 
$\det(M) \neq 0$.
\end{proof}

 Once we have defined a basis $(\xi_x)_{x \in \XX}:= (\phi_{\bx},\psi_{\brx})_{\bx \in \bX, \brx \in \brX}$ of $\ell_2(\XX,\mu)$, 
 we are able to define the analysis operator 
 \[ U: \begin{array}[t]{lcl} \ell_2(\XX,\mu) & \rightarrow & \R^{\XX} 
 	\\ f & \mapsto & (\bra{\xi_x,f})_{x\in \XX}=(\bra{\phi_{\bx},f}, \bra{\psi_{\brx},f})_{\bx \in \bX, \brx \in \brX}
	\end{array} \, . 
	\]
The hope is that when $f$ is "regular" (a notion still to be defined), the coefficients $(\bra{f,\psi_{\brx}})_{\brx \in \brX}$
are "small". Actually, when $f$ is a constant function, these coefficients are equal to zero. 	We stress the fact that these coefficients are not normalized. 
$f$ can be reconstructed from the coefficients $U(f)$. Let $(\tilde{\xi}_x)_{x \in \XX}$ be the dual basis of 
$(\xi_x)_{x \in \XX}$ in $\ell_2(\XX,\mu)$, defined by $\langle \tilde{\xi}_x,\xi_y \rangle = \delta_{xy}$. Then we get the obvious reconstruction formula
\[ 
f = \sum_{x \in \XX} \bra{\xi_x,f} \tilde{\xi}_x \, . 
\]
The approximation of $f$ is then $\sum_{\by \in \bX} \bra{\phi_{\by},f} \tilde{\xi}_{\by}$, while
the detail of $f$ is $\sum_{\bry \in \brX} \bra{ \psi_{\bry},f}  \tilde{\xi}_{\bry}$. 
The reconstruction of $f$ from its coefficients $U(f)$ can here be made explicit. This 
is the content of the following proposition.

\subsubsection{The reconstruction formula}
\begin{prop}[Reconstruction formula]
\label{reconstruction.prop}
For any $f \in \ell_2(\XX,\mu)$, let $\baf \in \ell_2(\bX,\mu)$ and $\brf \in \ell_2(\brX,\mu)$ be defined by 
\[
\forall \bx \in \bX \, , \,\, \baf(\bx) = K_{q'}f(\bx)=U(f)(\bx) \, ; 
\]
\[ 
\forall \brx \in \brX \, , \,\, \brf(\brx) =( K_{q'}-\Id)f(\brx)=U(f)(\brx) \, . 
\]
Then, 
\[ f = \begin{pmatrix} \Id_{\bX} - \frac{1}{q'} \bLL &  \LL_{\bX \brX} (-\LL_{\brX \brX})^{-1}
\\[.2cm]
  (-\LL_{\brX \brX})^{-1} \LL_{\brX \bX} & q'  \LL_{\brX \brX}^{-1} - \Id_{\brX}
\end{pmatrix}
\begin{pmatrix} \baf \\ \brf
\end{pmatrix} 
= \bR \baf + \brR \brf  \, ,  
\]
where
\begin{equation}
 \bR = \begin{pmatrix} \Id_{\bX} - \frac{1}{q'} \bLL 
 \\[.2cm]
  (- \LL_{\brX \brX})^{-1} \LL_{\brX \bX}
\end{pmatrix}  
\, , 
\,\, 
\mbox{ and }
\brR 
=\begin{pmatrix}  \LL_{\bX \brX} (-\LL_{\brX \brX})^{-1} \\  - \Id_{\brX} - q' (- \LL_{\brX \brX})^{-1} 
\end{pmatrix}
 \, . 
\end{equation}
\end{prop}
\begin{proof}
Note that 
$\begin{pmatrix} \baf \\ \brf \end{pmatrix}
= \cro{K_{q'} - \begin{pmatrix} 0 & 0 \\ 0 & \Id_{\brX} \end{pmatrix}} f$ . Since $K_{q'}=(\Id
- \frac{1}{q'} \LL)^{-1}$, we get
\[
 \pare{\Id - \frac{1}{q'} \LL}  \begin{pmatrix} \baf \\ \brf \end{pmatrix}
  = \cro{\Id - \pare{\Id - \frac{1}{q'} \LL} \begin{pmatrix} 0 & 0 \\ 0 & \Id_{\brX} \end{pmatrix}} f
  =   \begin{pmatrix} \Id_{\bX} & \LL_{\bX \brX}/q'
 \\ 0   &    \LL_{\brX \brX}/q' \end{pmatrix} f \, . 
 \]
 The inverse of $ \begin{pmatrix} \Id_{\bX} & \LL_{\bX \brX}/q'
 \\ 0   &  \LL_{\brX \brX}/q' \end{pmatrix}$ is 
 $  \begin{pmatrix} \Id_{\bX} &   \LL_{\bX \brX} (-\LL_{\brX \brX})^{-1}
 \\[.5pt]  0   &   q' (\LL_{\brX \brX})^{-1} \end{pmatrix}$. 
 Hence, 
\begin{align*}
f & = \begin{pmatrix} \Id_{\bX} &   \LL_{\bX \brX} (-\LL_{\brX \brX})^{-1}
 \\[.5pt]  0   &   q' (\LL_{\brX \brX})^{-1} \end{pmatrix} 
 \begin{pmatrix}  \Id_{\bX} -  \frac{1}{q'}  \LL_{\bX \bX} &  -  \frac{1}{q'}  \LL_{\bX \brX}
 \\[.2cm]
  -  \frac{1}{q'}  \LL_{\brX \bX} & \Id_{\brX} -  \frac{1}{q'}  \LL_{\brX \brX}
 \end{pmatrix}
 \begin{pmatrix} \baf \\ \brf \end{pmatrix}  \, , 
 \\
 & = 
  \begin{pmatrix} \Id_{\bX} -  \frac{1}{q'}  \bLL &   \LL_{\bX \brX}  (-\LL_{\brX \brX})^{-1}
  \\[.2cm] 
  (-  \LL_{\brX \brX})^{-1} \LL_{\brX \bX}  & q' (\LL_{\brX \brX})^{-1} - \Id_{\brX}
  \end{pmatrix} 
  \begin{pmatrix} \baf \\ \brf \end{pmatrix}  \, . 
 \end{align*}
  \end{proof}
  In the sequel, we will talk about $\bR$ as the approximation operator, and about 
  $\brR$ as the detail operator. In order to get stable numerical results in the analysis and the  reconstruction of the signal $f$, it is important to control the operator norm of these operators. This is easy for the analysis operator $U$ since for any $f \in \ell_p(\XX,\mu)$, 
 \[\nor{U(f)}_{p,\XX} = \nor{K_{q'}(f) - f \ind_{\brX}}_{p,\XX} \leq \nor{K_{q'}(f)}_{p,\XX}  + \nor{f}_{p,\XX} \, . 
 \]
 $K_{q'}$ being a probability kernel, symmetric w.r.t. $\mu$, we get by Jensen's inequality, 
\begin{align}
\nor{K_{q'}f}_{p,\XX}^p 
& = \sum_{x \in \XX} \mu(x) \va{K_{q'} f (x)}^p
\\
& \leq \sum_{x \in \XX} \mu(x) K_{q'}( \va{f}^p)(x)
= \bra{\ind, K_{q'}(\va{f}^p)} = \bra{K_{q'}(\ind) , \va{f}^p} 
=  \nor{f}_{p,\XX}^p \,  .
\label{Kq-norm.eq}
\end{align}
Hence for any $f \in \ell_p(\XX,\mu)$, 
\begin{equation} 
\nor{U(f)}_{p, \XX} \leq 2 \nor{f}_{p,\XX} \, . 
\label{U-norm.eq}
\end{equation}

 The aim of the two following sections is to provide bounds on the norm of the reconstruction operator.

  \subsection{Detail operator norm.}
  We give in this section a control on the norm of the detail operator $\brR$, in terms of 
  $\beta$ and $\gamma$ defined in  \eqref{beta-gamma-def}. Using Markov property at time $\tau_1$, 
  one can express $1/\beta$ as
  \[  \frac{1}{\beta} = \max_{\bx \in \bX} E_{\bx}\cro{H^+_{\bX} - \tau_1}  \, .
  \]
  In the sequel, when $A$ is a subset of $\XX$, and $f$ is a function on $A$, 
  $\nor{f}_{p,A}$ is the norm of $f$ in $\ell_p(A, \mu_A)$, where $\mu_A$ is the conditional 
  probability $\mu$ on $A$: for any $B \subset A$, $\mu_A(B)= \mu(B)/\mu(A)$.

  \begin{prop}[Detail operator norm] 
  \label{brR-norm.prop}
  Let $\bX$ be any proper subset of $\XX$, $\brX = \XX \setminus \bX$, and let 
  $\brR$ be the operator defined in \eqref{bR.def}. 
    For any $p \geq 1$,  for any $f \in \ell_{p}(\brX, \mu_{\brX})$, 
  \begin{equation}
    \label{brR-norm.eq}
 \nor{\brR f}_{p,\XX}
   \leq  
\cro{\pare{\frac{\alpha}{\beta}}^{p/p^*}
      + \pare{1+ \frac{q'}{\gamma}}^p}^{1/p} \,  
\mu(\brX)^{1/p}  \nor{ f}_{p,\brX} 
  \,  , 
  \end{equation}
  where  $p^*$ is the conjugate exponent of $p$. 
   \end{prop} 
   The proof is a consequence of Lemmas \ref{green.lem}-\ref{Opernorm.lem} given below.

\begin{lem}
\label{green.lem}
For any $x \in \XX$, and $\brx \in \brX$, we define 
\[ G_{\bX}(x,\brx)= E_x\cro{\int_0^{H^+_{\bX}} \ind_{X(s)=\brx} \, ds} \, . 
\]
Then,
\begin{equation}
\label{green1.eq}
 (-\LL_{\brX \brX})^{-1} =( G_{\bX})_{\brX \brX} \, \mbox{ , and }
\LL_{\bX \brX} (-\LL_{\brX \brX})^{-1} = \alpha  (G_{\bX})_{\bX \brX} \, . 
\end{equation}
The adjoint operator of $$(-\LL_{\brX \brX})^{-1}: \ell_p(\brX,\mu) \mapsto \ell_q(\brX,\mu)$$ is 
the operator $$(-\LL_{\brX \brX})^{-1}: \ell_{q^*}(\brX,\mu) \mapsto \ell_{p^*}(\brX,\mu) \, .$$ 
The adjoint operator of $$\LL_{\bX \brX} (-\LL_{\brX \brX})^{-1}:   \ell_p(\brX,\mu)  \mapsto \ell_q(\bX,\mu)$$
is the operator $$(-\LL_{\brX \brX})^{-1} \LL_{\brX \bX}:   \ell_{q^*}(\bX,\mu) \mapsto \ell_{p^*}(\brX,\mu) \, ,$$ with
\begin{equation}
\label{green.eq}(-\LL_{\brX \brX})^{-1} \LL_{\brX \bX}(\brx,\bx) = 
P_{\brx}\cro{X(H_{\bX})=\bx} \, , \brx \in \brX \, , \,\, \bx \in \bX \, .
\end{equation}
\end{lem}  
\begin{proof}
\begin{align*}
G_{\bX}(x,\brx) & = E_x\cro{\int_0^{\tau_1} \ind_{X(s)=\brx} \, ds 
 + \int_{\tau_1}^{H^+_{\bX}} \ind_{X(s)=\brx} \, ds} 
 \, , 
 \\
 & = \delta_x(\brx) E_x(\tau_1)
 + \sum_{z \in \XX} P(x,z) E_z\cro{\int_0^{H_{\bX}} \ind_{X(s)=\brx} \, ds} 
\mbox{ by Markov property at time } \tau_1 \, , 
\\
& = \frac{1}{\alpha} \delta_x(\brx) + \sum_{z \in \brX} P(x,z) 
 E_z\cro{\int_0^{H^+_{\bX}} \ind_{X(s)=\brx} \, ds}  \, . 
\end{align*}
This can be rewritten as 
\begin{align*}
\left\{
\begin{array}{l}
(G_{\bX})_{\bX \brX} = P_{\bX \brX} (G_{\bX})_{\brX \brX}
\\
(G_{\bX})_{\brX \brX} = \frac{1}{\alpha} \Id_{\brX} +  P_{\brX \brX} (G_{\bX})_{\brX \brX}
\end{array}
\right .
&
\Longleftrightarrow
\left\{
\begin{array}{l} 
-\LL_{\brX \brX} (G_{\bX})_{\brX \brX} = \Id_{\brX}
\\
(G_{\bX})_{\bX \brX} = \frac{1}{\alpha}  \LL_{\bX \brX} (G_{\bX})_{\brX \brX}
\end{array}
\right. 
\\
& \Longleftrightarrow
\left\{
\begin{array}{l} 
(-\LL_{\brX \brX})^{-1} = (G_{\bX})_{\brX \brX} 
\\
\LL_{\bX \brX} (-\LL_{\brX \brX})^{-1} = \alpha (G_{\bX})_{\bX \brX} 
\end{array}
\right. 
\,. 
\end{align*}
Using the symmetry of $\LL$ with respect to $\mu$, for any $g \in  \ell_{q^{\star}}(\brX,\mu)$ and any
$f \in \ell_p(\brX,\mu)$, 
\begin{align*}
 \bra{(-\LL_{\brX \brX})^{-1}f,g }_{\brX}
& =  \sum_{\brx,\bry \in \brX} \mu(\brx) 
(-\LL_{\brX \brX})^{-1}(\brx,\bry) f(\bry) g(\brx) 
 \\
 & = \sum_{\brx,\bry \in \brX} \mu(\bry) 
(-\LL_{\brX \brX})^{-1}(\bry,\brx) f(\bry) g(\brx)
 = \bra{f,(-\LL_{\brX \brX})^{-1}g }_{\brX}
\, .
\end{align*}
In the same way,   for all $g \in  \ell_{q^{\star}}(\bX,\mu)$ and all
$f \in \ell_p(\brX,\mu)$, 
\begin{align*}
 \bra{\LL_{\bX \brX}(-\LL_{\brX \brX})^{-1}f,g }_{\bX}
 &  =  \sum_{\bx \in \bX;\bry, \brz\in \brX} \mu(\bx)  
\LL(\bx,\bry)(-\LL_{\brX \brX})^{-1}(\bry,\brz) f(\brz) g(\bx) 
 \\
 &= \sum_{\bx \in \bX;\bry ,\brz\in \brX} \mu(\brz)
\LL(\bry,\bx)(-\LL_{\brX \brX})^{-1}(\brz,\bry) f(\brz) g(\bx) 
\\
 & = \bra{f,(-\LL_{\brX \brX})^{-1}\LL_{\brX \bX}g }_{\brX}
\, . 
\end{align*}
\eqref{green.eq} was already stated in the proof of Lemma \ref{LbarSchur.lem} (see \eqref{Q.eq}). 
\end{proof}

We are now able to give bounds on the norms of the operators involved in the definition of 
$\bR$ and $\brR$. 
\begin{lem} 
\label{Opernorm.lem}
 Let $\bX$ be any proper subset of $\XX$, $\brX = \XX \setminus \bX$. 
  $\beta$ and $\gamma$ being defined by 
 \eqref{beta-gamma-def}, we get
\begin{equation}
\label{Opernorm1.eq}
\nor{(-\LL_{\brX \brX})^{-1} \LL_{\brX \bX} f}_{p,\brX} 
\leq \pare{\frac{\alpha}{\beta}  \frac{\mu(\bX)}{\mu(\brX)} }^{1/p} \nor{f}_{p,\bX} \, . 
\end{equation}
\begin{equation}
\label{Opernorm2.eq}
\nor{ \LL_{\bX \brX} (-\LL_{\brX \brX})^{-1} f}_{p,\bX} 
\leq\pare{ \frac{ \alpha}{\beta} }^{1/p^*} \pare{\frac{\mu(\brX)}{\mu(\bX)}}^{1/p} \nor{f}_{p,\brX} \, . 
\end{equation}
\begin{equation}
\label{Opernorm3.eq}
\nor{  (-\LL_{\brX \brX})^{-1} f}_{p,\brX} 
\leq  \frac{1}{\gamma} \,  \nor{f}_{p,\brX} \, . 
\end{equation}
\end{lem}

\begin{proof}
\begin{align*}
\mu(\brX) \nor{(-\LL_{\brX \brX})^{-1} \LL_{\brX \bX} f}^p_{p,\brX}
& = \sum_{\brx \in \brX} \mu(\brx) \va{(-\LL_{\brX \brX})^{-1} \LL_{\brX \bX} f (\brx)}^p
\, , 
\\
& =  \sum_{\brx \in \brX} \mu(\brx) \va{E_{\brx}\cro{f(X(H_{\bX}))}}^p 
\mbox{ by  \eqref{green.eq},} 
\\
& \leq  \sum_{\brx \in \brX} \mu(\brx)  E_{\brx}\cro{\va{f(X(H_{\bX}))}^p} 
\mbox{ by Jensen's inequality ,} 
\\
& =  \sum_{\brx \in \brX, \by \in \bX} \mu(\brx)  (( -\LL_{\brX \brX})^{-1} \LL_{\brX \bX})(\brx,\by)
\va{f}^p(\by)
\mbox{ by  \eqref{green.eq},} 
\\
& =  \sum_{\brx \in \brX, \by \in \bX} \mu(\by) (\LL_{\bX \brX}( -\LL_{\brX \brX})^{-1} )(\by,\brx)
\va{f}^p(\by) 
\mbox{ by symmetry,} \\
& 
= \alpha \sum_{\by \in \bX} \mu(\by) \va{f(\by)}^p  E_{\by}\cro{\int_0^{H^+_{\bX}} 
\ind_{\brX}(X(s)) \, ds }
\mbox{ by  \eqref{green1.eq},} 
\\
& 
= \alpha \sum_{\by \in \bX} \mu(\by) \va{f(\by)}^p  E_{\by}\cro{H^+_{\bX}-\tau_1}
\, , 
\\
& 
\leq \alpha \max_{\bx \in \bX} E_{\bx}\cro{H^+_{\bX}-\tau_1} \mu(\bX) \nor{f}^p_{p,\bX} 
\, . 
\end{align*}
This gives \eqref{Opernorm1.eq}.  \eqref{Opernorm2.eq} follows from  \eqref{Opernorm1.eq}, since the operator $\LL_{\bX \brX} (-\LL_{\brX \brX})^{-1} $ from $\ell_p(\brX,\mu)$ to 
$\ell_p(\bX,\mu)$ is the adjoint operator of the operator $(-\LL_{\brX \brX})^{-1} \LL_{\brX \bX}  $ from $\ell_{p^*}(\bX,\mu)$ to  $\ell_{p^*}(\brX,\mu)$. 
Concerning \eqref{Opernorm3.eq}, note that
by Lemma \ref{green.lem}, for any $\brx, \bry \in \brX$, $(-\LL_{\brX \brX})^{-1}(\brx,\bry) \geq 0$. 
\begin{align*}
\mu(\brX) & \nor{(-\LL_{\brX \brX})^{-1} f}^p_{p,\brX}\\
&= \sum_{\brx \in \brX} \mu(\brx) \va{(-\LL_{\brX \brX})^{-1} f (\brx)}^p
\, , 
\\
& = \sum_{\brx \in \brX} \mu(\brx) 
\va{\sum_{\bry \in \brX} (-\LL_{\brX \brX})^{-1}(\brx,\bry) f (\bry)}^p
\, , 
\\
& \leq  \sum_{\brx \in \brX} \mu(\brx)  
\pare{\sum_{\bry \in \brX} (-\LL_{\brX \brX})^{-1}(\brx,\bry)}^{p/p^*}
\pare{\sum_{\bry \in \brX}  (-\LL_{\brX \brX})^{-1}(\brx,\bry) \va{f(\bry)}^p}
\\
& \hspace*{5cm} \mbox{ by H\"older's inequality ,} 
\\
& \leq \pare{\max_{\brx \in \brX} \sum_{\bry \in \brX} (-\LL_{\brX \brX})^{-1}(\brx,\bry)}^{p/p^*}
\pare{\sum_{\brx \in \brX, \bry \in \brX} \mu(\brx)    (-\LL_{\brX \brX})^{-1}(\brx,\bry) \va{f(\bry)}^p}
\, ,
\\
& \leq \pare{\max_{\brx \in \brX} \sum_{\bry \in \brX} (-\LL_{\brX \brX})^{-1}(\brx,\bry)}^{p/p^*}
\pare{\sum_{\bry \in \brX} \mu(\bry) \va{f(\bry)}^p  \, \sum_{\brx \in \brX}  (-\LL_{\brX \brX})^{-1}(\bry,\brx) }
\\
& 
\hspace*{5cm} \mbox{ by symmetry ,} 
\\
& \leq \pare{\max_{\brx \in \brX} \sum_{\bry \in \brX} (-\LL_{\brX \brX})^{-1}(\brx,\bry)}^{p} 
\mu(\brX) \nor{f}^p_{p,\brX}
\end{align*}
This ends the proof of \eqref{Opernorm3.eq} since by  \eqref{green1.eq}, 
\[ \sum_{\bry \in \brX} (-\LL_{\brX \brX})^{-1}(\brx,\bry) 
= E_{\brx} \cro{\int_0^{H^+_{\bX}} \ind_{\brX}(X(s)) \, ds}
= E_{\brx} \cro{ H^+_{\bX}} 
= E_{\brx} \cro{ H_{\bX}} \, . 
\]

\end{proof}

 \noindent
 {\it Proof of Proposition \ref{brR-norm.prop}.}
For any $f \in \ell_p(\brX,\mu_{\brX})$, 
 \begin{align*}
\nor{\brR f}^p_{p,\XX}
& 
 = \nor{\begin{pmatrix} \LL_{\bX \brX}(-\LL_{\brX \brX})^{-1} 
 	\\ -\Id_{\brX} -q' ( -\LL_{\brX \brX})^{-1} 
	\end{pmatrix} f}^p_{p,\XX}
 \\
& = \mu(\bX) \nor{\LL_{\bX \brX}(-\LL_{\brX \brX})^{-1} f}^p_{p,\bX}
 + \mu(\brX)  \nor {\pare{\Id_{\brX} + q' ( -\LL_{\brX \brX})^{-1}} f}^p_{p,\brX}
 \\
 & \leq  \cro{\pare{\frac{\alpha}{\beta}}^{p/p^*} \mu(\brX)
      + \pare{1+ \frac{q'}{\gamma}}^p \mu(\brX)} 
      \nor{f}^p_{p,\brX} 
 \mbox{ by } \eqref{Opernorm2.eq} \mbox{ and }
      				\eqref{Opernorm3.eq} \, . 
\end{align*}
\qed
 
 \subsection{Approximation operator norm.}
   This section aims to control the operator norm of $\bR$ defined in \eqref{bR.def}. 
   \begin{prop}[Approximation operator norm] 
  \label{bR-norm.prop}
  Let $\bX$ be any proper subset of $\XX$, $\brX = \XX \setminus \bX$, and let 
  $\bR$ be the operator defined in \eqref{bR.def}. 
    For any $p \geq 1$,  for any $f \in \ell_{p}(\bX, \mu_{\bX})$, 
  \begin{equation}
  \label{bR-norm.eq}
 \nor{\bR f}_{p,\XX}
   \leq  
\cro{\pare{1+ 2 \frac{\ba}{q'}}^p + \frac{\alpha}{\beta}}^{1/p} \,  
\mu(\bX)^{1/p}  \nor{ f}_{p,\bX} 
  \,  , 
  \end{equation}
  where $\ba$ is defined by \eqref{bar-alpha.eq}.
  \end{prop}
 \begin{proof}
 For any $f \in \ell_p(\bX,\mu_{\bX})$, 
  \begin{align*}
\nor{\bR f}^p_{p,\XX}
& 
 = \nor{\begin{pmatrix} \Id_{\bX} - \frac{\bLL}{q'} 
  	\\  ( -\LL_{\brX \brX})^{-1}  \LL_{\brX \bX}
	\end{pmatrix} f}^p_{p,\XX}
 \\
& = \mu(\bX) \nor{\pare{\Id_{\bX} - \frac{\bLL}{q'}} f}^p_{p,\bX}
 + \mu(\brX) \nor { ( -\LL_{\brX \brX})^{-1}  \LL_{\brX \bX} f}^p_{p,\brX}
 \\
 & \leq  \mu(\bX) \nor{\pare{\Id_{\bX} - \frac{\bLL}{q'}} f}^p_{p,\bX} 
      +  \mu(\bX) \frac{\alpha}{\beta}  \nor{f}^p_{p,\bX} 
 \mbox{ by } \eqref{Opernorm1.eq} \, . 
\end{align*}
It just  remains to prove that for any  $f \in \ell_{p}(\bX, \mu_{\bX})$,  
\[ \nor{\bLL f}_{p,\bX} \leq 2 \ba \nor{f}_{p,\bX} \, . 
\]
 By definition of $\ba$, the matrix $\tilde{P} = \Id_{\bX} + \frac{\bLL}{\ba}$ is a stochastic
 matrix, which is symmetric w.r.t $\mu_{\bX}$. Therefore, it is contraction operator from
  $\ell_p(\bX,\mu_{\bX})$ to  $\ell_p(\bX,\mu_{\bX})$. Hence, 
 \[ \nor{\bLL f}_{p,\bX}  =  \ba \nor{(\Id_{\bX} - \tilde{P}) f}_{p,\bX}
 \leq 2 \ba \nor{f}_{p,\bX} \, . 
 \] 
  \end{proof}
 
 \subsection{Size of the detail.}
 In this section, we provide  control of the unnormalized detail coefficients $\brf$ in terms of the regularity of the signal $f$. 
    
\begin{prop}
   \label{brf.prop}
   For any $p \geq 1$ and any $f \in \ell_{p}(\XX,\mu)$, 
\[  \nor{\brf}_{p,\brX} =  \nor{(K_{q'}-\Id) f}_{p,\brX}
 \leq \frac{\max_x K_{q'}(x,\brX)^{1/p}}{q' \mu(\brX)^{1/p}} \nor{\LL f}_{p, \XX} \, . 
\]
   \end{prop}
\begin{proof}
Note that by definition of $K_{q'}$, $K_{q'} - \Id = \frac{1}{q'} K_{q'} \LL $. Hence, 
\begin{align*}
\mu(\brX) \nor{ (K_{q'} - \Id) f}^p _{p,\brX} 
& =  \frac{1}{q'^p} \sum_{\brx\in \brX}\mu(\brx) \va{E_{\brx}\cro{ \LL f (X(T_{q'})}}^p 
\leq  \frac{1}{q'^p} \sum_{\brx\in \brX}\mu(\brx)  E_{\brx}\cro{ \va{ \LL f (X(T_{q'})}^p}
\, , 
\\
& \leq   \frac{1}{q'^p} \sum_{\brx\in \brX, x \in \XX}\mu(\brx) K_{q'}(\brx,x) \va{\LL f(x)}^p
\, , 
\\
& \leq  \frac{1}{q'^p} \sum_{\brx\in \brX, x \in \XX}\mu(x) K_{q'}(x,\brx) \va{\LL f(x)}^p 
\mbox{ by symmetry  } 
\, , 
\\
& 
\leq  \frac{\max_x K_{q'}(x,\brX)}{q'^p} \sum_{ x \in \XX}\mu(x) \va{\LL f(x)}^p
\, . 
\end{align*}
\end{proof}

  \subsection{Link between $(\XX,\LL)$ and $(\bX, \bLL)$.}
  We give bounds on the error in the interwining relation.
  \begin{prop}
  \label{intertwiningp.prop}
  Let $\bX$ be any proper subset of $\XX$, $\brX = \XX \setminus \bX$. Recall the 
definition of   $\beta$  given in \eqref{beta-gamma-def}. 
 For any $p \geq 1$, and any $f \in \ell_p(\XX,\mu)$, 
  \begin{equation}
  \label{intertwiningp.ineq}
  \nor{\pare{\bLL (K_{q'})_{\bX \XX} - (K_{q'})_{\bX \XX} \LL} f}_{p,\bX} 
  \leq 2 q'  \pare{\frac{\alpha }{\beta}}^{1/p^*} \frac{1}{\mu(\bX)^{1/p}} \nor{f}_{p,\XX}
  \, . 
  \end{equation}

 \end{prop}
  \begin{proof} 
  Writing that $\bLL$ is a Schur complement and using that $K_{q'}\LL=\LL K_{q'}$, we get

  \begin{align}
 \bLL (K_{q'})_{\bX \XX} - (K_{q'})_{\bX \XX} \LL
 &
 =   \LL_{\bX \bX} (K_{q'})_{\bX \XX} 
 + \LL_{\bX \brX} \pare{- \LL_{\brX \brX}}^{-1} \LL_{\brX \bX} (K_{q'})_{\bX \XX}
  - ( \LL K_{q'})_{\bX \XX}
\nonumber
\\
&
  = -  \LL_{\bX \brX} (K_{q'})_{\brX \XX} 
   + \LL_{\bX \brX} \pare{- \LL_{\brX \brX}}^{-1} \LL_{\brX \bX} (K_{q'})_{\bX \XX}
\nonumber
\\
&
= \LL_{\bX \brX} \pare{- \LL_{\brX \brX}}^{-1} 
\pare{\LL_{\brX \brX}(K_{q'})_{\brX \XX}  + \LL_{\brX \bX} (K_{q'})_{\bX \XX} }
\nonumber
\\
& 
= \LL_{\bX \brX} \pare{- \LL_{\brX \brX}}^{-1} \pare{\LL K_{q'}}_{\brX \XX}
\,. 
\label{VT.exp}
 \end{align}

It has already been proven in Lemma \ref{Opernorm.lem} (see \eqref{Opernorm2.eq}), that 
for any $f \in \ell_p(\brX,\mu_{\brX})$,
\[
\nor{ \LL_{\bX \brX} \pare{- \LL_{\brX \brX}}^{-1} f}_{p,\bX} 
\leq \pare{\frac{\alpha}{\beta}}^{1/p^*}  \pare{\frac{\mu(\brX)}{\mu(\bX)}}^{1/p}
 \nor{ f}_{p,\brX}
\]
 To get \eqref{intertwiningp.ineq}, it just remains to prove that for any $f \in \ell_p(\XX,\mu)$,
 \begin{equation}
 \nor{\LL K_{q'} f}_{p,\brX}  \leq \frac{2q'}{\mu(\brX)^{1/p} } \nor{f}_{p,\XX} 
 \, .
 \end{equation}
But this is a direct  consequence of the equality $\LL K_{q'} =q' (K_{q'}-\Id)$, and of the inequality \eqref{Kq-norm.eq}.
\end{proof} 

\section{The multiresolution scheme.}

We assume in this section that we have at our disposal a decreasing sequence of proper 
subsets of $\XX$: 
\[ \XX_0 = \XX \supsetneq \XX_1 \supsetneq \cdots \supsetneq \XX_k \supsetneq \emptyset 
\, , 
\]
and a sequence of non negative real parameters $\acc{q'_i; i = 0, \cdots k-1}$. We will discuss later the choice of such sequences.  
To stick to the previous notations, we define for $i \geq 0$, 
\begin{itemize}
\item $\bX_i=\XX_{i+1}$;
\item  $\brX_i = \XX_i \setminus \XX_{i+1}$;
\item $\LL_0 = \LL$, and for $i \geq 0$, $\LL_{i+1} = S_{\LL_i}((\LL_i)_{\brX_i,\brX_i})$. 
$\LL_i$ is a  Markov generator on $\XX_i$ which is symmetric with respect to $\mu$.
\item $\alpha_i = \max_{x_i \in \XX_i} (-\LL_i)(x_i,x_i)$;
\item $1/\beta_i = \max_{\bx_{i} \in \XX_{i+1}} E_{\bx_i}\cro{H^{(i)+}_{\XX_{i+1}}-\tau_1^{(i)}}$, where $H^{(i)+}_A$ is the return time on the set $A$ of a Markov process $X^{(i)}$ with
generator $\LL_i$, and $\tau_1^{(i)}$ is its first jump time (exponential time with parameter $\alpha_i$);
\item $1/\gamma_i = \max_{\brx_i \in \brX_i} E_{\brx_i}\cro{H^{(i)}_{\XX_{i+1}}}$;
\item for any $x_i, y_i \in \XX_i$, $K_i(x_i,y_i) = q'_i (q'_i \Id_{\XX_i} - \LL_i)^{-1}(x_i,y_i)
= \P_{x_i} \cro{X^{(i)}(T_{q'_i}) = y_i}$, where $T_{q'_i}$ is an exponential random variable with parameter
$q'_i$ independent of $X^{(i)}$. 
\end{itemize}
Given a signal $f_i$ defined on $\XX_i$, we can define its approximation coefficients at
scale $i+1$:
\[ \forall x_{i+1} \in \XX_{i+1} \, , \,\, 
f_{i+1}(x_{i+1})= \bar{f}_i (x_{i+1}) = \bK_{i} (f_i) (x_{i+1}) \, ,   \mbox{ where }  \bK_i := (K_i)_{\XX_{i+1}, \XX_i} \, , 
\]
 and its detail coefficients at scale $i+1$: 
 \[ \forall \brx_i \in \brX_{i} \, , \,\, 
 g_{i+1}(\brx_i) = \brf_i(\brx_i) =  \brK_{i}( f_i)(\brx_i)- f_i(\brx_i)  \, \mbox{ where } \brK_i = (K_i)_{\brX_{i}, \XX_i} \, , 
\]
so that 
\[ f_i = \bR_i f_{i+1} + \brR_{i} g_{i+1}  \, ,
\]
with
\[ \bR_i =  \begin{pmatrix} \Id_{\XX_{i+1}} - \frac{1}{q'_i} \LL_{i+1}
 \\ ((- \LL_i)_{\brX_i \brX_i})^{-1} (\LL_i)_{\brX_i \XX_{i+1}}
\end{pmatrix}  
\, , 
\,\, 
\mbox{ and }
\brR_i 
=\begin{pmatrix} ( \LL_i)_{\XX_{i+1} \brX_i} ((- \LL_i)_{\brX_i \brX_i})^{-1} 
\\  q' _i ((\LL_i)_{\brX_i \brX_i})^{-1} - \Id_{\brX_i}
\end{pmatrix}
 \, . 
\]
Iterating the procedure, we get the usual multiresolution scheme: 
\[ \begin{matrix} f_0 = f & \rightarrow & f_1 & \rightarrow & f_2 & \cdots &  \rightarrow & f_k
\\ & \searrow &  & \searrow &   & & \searrow & 
\\[-.1cm]
&                 & g_1 &         &g_2 &              &  & g_k 
\end{matrix}
\]

\subsection{Analysis operator norm}
For $k \geq 1$, let us consider the analysis operator $U_k$:
\[ U_k : \begin{array}{lcl} \ell_p(\XX) & \rightarrow & 
	\ell_p(\XX_k,\mu_{\XX_k}) \times \ell_p(\brX_{k-1}, \mu_{\brX_{k-1}}) \times \cdots \times \ell_p(\brX_{0}, \mu_{\brX_{0}})
		\\
		f & \mapsto & [f_k,g_k,g_{k-1},\cdots, g_1] 
		\end{array}
		\]
The space $\ell_p(\XX_k,\mu_{\XX_k}) \times \ell_p(\brX_{k-1}, \mu_{\brX_{k-1}}) \times \cdots \times \ell_p(\brX_{0}, \mu_{\brX_{0}})$
 is endowed
with the norm 
\[ \nor{[f_k,g_k,g_{k-1},\cdots, g_1]}_p := 
\pare{ \mu(\XX_k) \nor{f_k}_{p,\XX_k}^p +  \sum_{i=1}^k \mu(\brX_{i-1}) \nor{g_i}_{p,\brX_{i-1}}^p}^{1/p}
\, . 
\]
\begin{prop}[Analysis operator norm]
For any $f \in \ell_p(\XX,\mu)$,
\begin{equation}
\label{analyse.eq}
\nor{U_k(f)}_p \leq 2^{1/p^*} (1+k)^{1/p} \nor{f}_p 
\, . 
\end{equation}
\end{prop}
\noindent

{\it Remark:} To maintain localization properties of our wavelet basis we do not use any orthogonalization procedure. Then, working with a general graph, we do not expect a conditioning of the multilevel analysis or reconstruction operator that is independent from the graph size. However for $p=\infty$ our analysis operator norm is independent of the graph size, while, for $p<\infty$ it depends only on the number of levels, which is typically only logarithmic in the graph size. Making a careful analysis of the following proof, we will see that at least in the case $p=1$ our estimation is optimal for the analysis operator we proposed. As far as the multilevel reconstruction operator is concerned, its norm is strongly related with the choice of the $q'_i$ that we will discuss at the end of Section \ref{choixq.sec}.

\begin{proof}

\begin{align*}
\nor{U_k(f)}_p^p
& = 
\sum_{x_k \in \XX_k} \mu(x_k) \va{f_k(x_k)}^p + \sum_{i=1}^k \sum_{\brx_{i-1} \in \brX_{i-1}} \mu(\brx_{i-1}) 
\va{g_i(\brx_{i-1})}^p  \, , 
\\
& = \sum_{x_k \in \XX_k} \mu(x_k) \va{(\bK_{k-1} ... \cdots \bK_0 )(f)(x_k)}^p 
\\
& \hspace*{1cm} 
+ \sum_{i=1}^k \sum_{\brx_{i-1} \in \brX_{i-1}} \mu(\brx_{i-1})  
\va{((\brK_{i-1}-\Id_{\brX_{i-1} \XX_{i-1}}) \bK_{i-2} \cdots \bK_0)(f)(\brx_{i-1})}^p \, , 
\\
& \leq \sum_{x_k \in \XX_k} \mu(x_k) \va{(\bK_{k-1} ... \cdots \bK_0 )(f)(x_k)}^p
\\
& \hspace*{1cm} 
+ 2^{p-1}
\sum_{i=1}^k \sum_{\brx_{i-1} \in \brX_{i-1}} \mu(\brx_{i-1})  
\left( \va{(\brK_{i-1} \bK_{i-2} \cdots \bK_0)(f)(\brx_{i-1})}^p \right.\\
  &\hspace*{6cm}    + \va{( \bK_{i-2} \cdots \bK_0)(f)(\brx_{i-1})}^p\Bigr)
\, , 
\\
& 
\leq 2^{p-1} (a_k + b_k) 
\, , 
\end{align*}
where
\begin{align*}
a_k & = \sum_{x_k \in \XX_k} \mu(x_k) (\bK_{k-1} \cdots \bK_0 )(\va{f}^p)(x_k)
+ \sum_{i=1}^k \sum_{\brx_{i-1} \in \brX_{i-1}} \mu(\brx_{i-1})   
(\brK_{i-1} \bK_{i-2} \cdots \bK_0)(\va{f}^p)(\brx_{i-1})
\\
b_k & =  \sum_{i=1}^k \sum_{\brx_{i-1} \in \brX_{i-1}} \mu(\brx_{i-1})   
( \bK_{i-2} \cdots \bK_0)(\va{f}^p)(\brx_{i-1}) \, . 
\end{align*}
Note that 
\begin{align*}
& \sum_{x_k \in \XX_k} \mu(x_k) (\bK_{k-1}\cdots \bK_0 )(\va{f}^p)(x_k)
+ \sum_{\brx_{k-1} \in \brX_{k-1}} \mu(\brx_{k-1})   
(\brK_{k-1} \bK_{k-2} \cdots \bK_0)(\va{f}^p)(\brx_{k-1})
\\ 
& = \sum_{x_{k-1} \in \XX_{k-1}} \mu(x_{k-1}) (K_{k-1} \bK_{k-2} \cdots \bK_0 )(\va{f}^p)(x_{k-1})
\\
 & = \sum_{x_{k-1} \in \XX_{k-1}} \mu(x_{k-1}) (\bK_{k-2} \cdots \bK_0 )(\va{f}^p)(x_{k-1})
\end{align*}
by symmetry. Therefore, $a_k=a_{k-1}=a_1=\nor{f}_{p,\XX}^p$. Concerning $b_k$, using symmetry,
it holds 

\begin{align*}
&\sum\limits_{i=1}^k\sum_{\brx_{i-1} \in \brX_{i-1}} \mu(\brx_{i-1})   ( \bK_{i-2} \cdots \bK_0)(\va{f}^p)(\brx_{i-1})\\
&=\sum\limits_{x_0\in\XX_0}\mu(x_0)|f(x_0)|^p
\sum\limits_{i=1}^k\sum\limits_{x_{1}\in\XX_1}\cdots\!\!\!\!\!\sum\limits_{x_{i-2}\in\XX_{i-2}}\sum\limits_{\brx_{i-1}\in\brX_{i-1}}
K_0(x_0,x_1)\cdots K_{i-3}(x_{i-3},x_{i-2})K_{i-2}(x_{i-2},\brx_{i-1})\\
&\leq k\nor{f}_{p,\XX}^p \, , 
\end{align*}
This ends the proof of \eqref{analyse.eq}. 
\end{proof}

We finally give an example for which one can check that in the case $p=1$ our estimation \eqref{analyse.eq} is optimal. Consider the graph of size $n=l+1$ with vertex set $\XX=\{0,1,\dots,l\}$ and transition rates $$w(x,y)=\left\{\begin{array}{ll}
\varepsilon^{2x+1}&\mbox{ if } y=x+1,\\
\varepsilon^{2x}&\mbox{ if } y=x-1,\\
0&\mbox{ otherwise,}
\end{array}\right.$$
together with the signal $$f(x)=\varepsilon^{2x-l},\qquad 0\leq x\leq l,$$ for $\varepsilon\ll 1$. Choosing $q_i=q'_i=\varepsilon^{i+\frac{3}{2}}$ for $0\leq i\leq k$, the random subsets $\XX_0\supset\XX_1\supset \XX_2\supset\cdots $ sampled through Wilson's algorithm are typically given by $$\XX_i=\{i,\dots,l\},\qquad i\leq k, $$ and the inequalities of the previous proof are turned into equalities when $p=1$ and in the limit $\varepsilon\rightarrow 0$. 

\subsection{Approximation error}
\label{jackson.sec}
Given the coefficients $[f_k,g_k,g_{k-1},\cdots, g_1]$,
the reconstruction of $f=f_0$ is 
\begin{align*}
 f=f_0 & = \bR_0 f_1 + \brR_0 g_1
 \\
	  &  = \bR_0 \bR_1 f_2 + \bR_0 \brR_1 g_2 + \brR_0 g_1
\\
& = \bR_0 \bR_1 \cdots \bR_{k-1} f_k
+ \sum_{j=0}^{k-1} (\bR_0 \cdots \bR_{j-1})   \brR_j g_{j+1}
 \, .
\end{align*}
The approximation of $f$ at scale $k$ is thus $\bR_0 \bR_1 \cdots \bR_{k-1} f_k$ and 
we have the following Jackson's type inequality:

\begin{prop}[Jackson's inequality]
\label{jackson.prop}
For $i \in \acc{0, \cdots, k-1}$ let $\va{\bR_i}_p$   ($\va{\brR_i}_p$ respectively) denote the norm operator of 
$\bR_i: \ell_p(\XX_{i+1},\mu_{\XX_{i+1}}) \mapsto \ell_p(\XX_i,\mu_{\XX_i})$ 
($\brR_i: \ell_p(\brX_{i},\mu_{\brX_i}) \mapsto \ell_p(\XX_i,\mu_{\XX_i})$ respectively). We set 
$$  \br_{i,p}  := \pare{\frac{\mu(\XX_i)}{\mu(\XX_{i+1}) }}^{1/p}  \va{\bR_i}_p \, , 
\brr_{i,p}  := \pare{\frac{\mu(\XX_i)}{\mu(\brX_{i}) }}^{1/p}  \va{\brR_i}_p \, ,
\mbox{ and }
e_{i,p}  := \pare{\frac{\mu(\XX_{i+1})}{\mu(\XX_{i}) }}^{1/p}  \va{ \LL_{i+1} \bK_{i} -  \bK_{i}  \LL_{i}}_p  \, . 
$$ 
For any $p \geq 1$ and any $f \in \ell_p(\XX,\mu)$,
\begin{align}
\nonumber
 \nor{f-\bR_0 \bR_1 \cdots \bR_{k-1} f_k}_{p,\XX}
 & \leq  \acc{\sum_{j=0}^{k-1} \br_{0,p}  \cdots \br_{j-1,p}  \brr_{j,p}  \frac{1}{q'_j}} \nor{\LL f}_{p,\XX}
\\
& \hspace*{2cm} 
+ \acc{\sum_{j=0}^{k-1}  \br_{0,p}  \cdots \br_{j-1,p} \brr_{j,p}  \frac{1}{q'_j}
 \sum_{l=0}^{j-1} e_{l,p}  } \nor{f}_{p,\XX}
 \, . 
 \label{jackson.ineq}
\end{align}
\end{prop}

Note that it has been proven in Propositions \ref{bR-norm.prop}, \ref{brR-norm.prop} and  \ref{intertwiningp.prop} that 
\[ \br_{i,p}  \leq \cro{\pare{1 + 2 \frac{\alpha_{i+1}}{q'_i}}^p + \frac{\alpha_i}{\beta_i}}^{1/p}  
\, , 
\, \, 
\brr_{i,p}  \leq \cro{\pare{\frac{\alpha_i}{\beta_i}}^{p/p^*} 
+ \pare{1 + \frac{q'_i}{\gamma_i}}^p}^{1/p}  
\, , 
\mbox{ and } e_{i,p}  \leq 2 q'_{i} \pare{\frac{\alpha_{i}}{\beta_{i}}}^{1/p^*} \, . 
\]
Moreover, we would like to stress the fact that the term involving $\nor{f}_{p,\XX}$ in \eqref{jackson.ineq} is linked to the error in the intertwining relation, and would disappear if 
 this relation was exact.\\

In case  the intertwining relation error is small, we can interpret our inequality (\ref{jackson.ineq}) as a Jackson type inequality.
Indeed in the setting of approximation theory, Jackson type inequalities relate the smoothness of a function with the approximation properties of a multiscale basis. Let us give some more details on these classical results in the setting of wavelet analysis. Assume $\phi$ is a real valued scaling function, in the Schwartz class or compactly supported, i.e a function such that one can find a wavelet $\psi$ such that $\{\phi(.-k),k\in\Z\}\cup\{x\mapsto 2^{j/2}\psi(2^j.-k),j\in\N,k\in\Z\} $ is an orthonormal basis of $L^2(\R)$, a so-called wavelet basis. The construction of such kind of basis has a deep relationship with the discrete scheme described in the introduction \ref{sec:intro} as it is explained in \cite{MAL} (see also \cite{DAUB} for details).

We can define for $j\geq 0$, $V_j:=Span\{\phi_{j,k}=2^{j/2}\phi(2^j.-k),k\in\Z\}$ and $P_j f$ the orthogonal projector on $V_j$.

The remarkable properties of wavelet basis make it possible to analyse functions in other functional spaces than $L^2(\R)$. Suppose the function $\psi$ of our basis has $n$ vanishing moments with $n\in\N^*$ (by construction $\psi$ has always at least one vanishing moment).  Let $1\leq p\leq \infty$ and $W^{n,p}$ denote the Sobolev space of functions $f$ in $L^p$ whose derivative of order $n$ belongs to $L^p$. Let $|f|_{W^{n,p}}=\|\frac{d^n f}{dx^n}\|_{L^p}$.

We have indeed the following result in the classical setting (see \cite{COHEN} for a detailed proof).

\begin{theo}[Classical Jackson's inequality]

There exists a constant $C>0$ such that for all $f\in W^{n,p}$ and $ j\geq 0$,
\begin{equation}\label{class.jackson.ineq}
\|f-P_j f\|_{L^p}\leq C2^{-nj}|f|_{W^{n,p}} \, . 
\end{equation}
\end{theo}

 As one can see, our Jackson's inequality (\ref{jackson.ineq}) relates also the smoothness of the signal $f$ on the graph (measured by $\nor{\LL f}_{p,\XX}$) provided that the error in the intertwining equation is small. The smallest is $\nor{\LL f}_{p,\XX}$, the closer to $f$ should be the approximation part. As in the classical inequality (\ref{class.jackson.ineq}), the constant 
  in front of the Laplacian term depends on the level of approximation and gets worse when this level gets coarser.

We can now prove Proposition \ref{jackson.prop}.
\begin{proof}

\begin{align*}
\nor{f-\bR_0 \bR_1 \cdots \bR_{k-1} f_k}_{p,\XX_0}
& \leq \sum_{j=0}^{k-1} \nor{(\bR_0 \cdots \bR_{j-1})   \brR_j g_{j+1}}_{p,\XX_0}
\, , 
\\
& \leq \sum_{j=0}^{k-1} \prod_{i=0}^{j-1} \va{\bR_i}_p   \va{\brR_j}_p \nor{g_{j+1}}_{p,\brX_j}
\, , 
\\
& \leq \sum_{j=0}^{k-1} \prod_{i=0}^{j-1} \va{\bR_i}_p 
	  \va{\brR_j}_p  \pare{\frac{\mu(\XX_j)}{\mu(\brX_{j})}}^{1/p} \frac{\nor{\LL_j f_j}_{p,\XX_j}}{q'_j} 
\, , 
\end{align*}
by Proposition \ref{brf.prop}. Since $\mu(\XX_0)=1$, this leads to 
\begin{align*}
\nor{f-\bR_0 \bR_1 \cdots \bR_{k-1} f_k}_{p,\XX_0}
\leq \sum_{j=0}^{k-1} \br_0 \cdots \br_{j-1} \brr_j \frac{\mu(\XX_j)^{1/p}}{q'_j} \nor{\LL_j f_j}_{p,\XX_j} 
\, . 
\end{align*}
Now, 

\begin{align*}
\LL_j f_j & = \LL_j \bK_{j-1} f_{j-1} 
\\
 & = \LL_j \bK_{j-1} \cdots \bK_0 f_0
 \\
 & = \bK_{j-1} \cdots \bK_0 \LL_0 f_0 + 
 \sum_{l=0}^{j-1} \bK_{j-1} \cdots \bK_{l+1} \pare{\LL_{l+1} \bK_{l} - \bK_{l} \LL_{l}} \bK_{l-1} \cdots \bK_0 f_0
 \, .
 \end{align*}
 For any $l \in \acc{0,\cdots, j-1}$, $\bK_l$ is an operator from $\ell_p(\XX_l,\mu_{\XX_l})$ to
 $\ell_p(\XX_{l+1},\mu_{\XX_{l+1}})$, and  \eqref{Kq-norm.eq} states that  
 $\nor{\bK_l}_p \leq \pare{\frac{\mu(\XX_l)}{\mu(\XX_{l+1})}}^{1/p}$. Hence, 
 \begin{align*}
\nor{\LL_j f_j}_{p, \XX_j}
& \leq \pare{\frac{\mu(\XX_0)}{\mu(\XX_{j})}}^{1/p} \nor{\LL_0 f_0}_{p,\XX_0}
+ \sum_{l=0}^{j-1} \pare{\frac{\mu(\XX_{l+1})}{\mu(\XX_{j})}}^{1/p} \va{\LL_{l+1} \bK_{l} - \bK_{l} \LL_{l}}_p 
\pare{\frac{\mu(\XX_0)}{\mu(\XX_{l})}}^{1/p}\nor{f_0}_{p,\XX}
\\
&  = \pare{\frac{\mu(\XX_0)}{\mu(\XX_{j})}}^{1/p} \cro{\nor{\LL_0 f_0}_{p,\XX_0}
+ \sum_{l=0}^{j-1} e_{l}  \nor{f_0}_{p,\XX}}
\, , 
\end{align*}
This gives \eqref{jackson.ineq}.
\end{proof}


\section{About the choice of the parameters $q$, $q'$.}
\label{choixq.sec}
The aim of this section is to give a guideline in the choice of the various parameters
involved  in the multiresolution scheme. The requirements we would like to achieve are the 
following ones:
\begin{enumerate}
\item At each step of the downsampling, we would like to keep a fixed proportion of points
in the current set. This would ensure that the number of steps $k$ in the multiresolution 
scheme is of order $\log(\va{\XX})$. In view of Proposition \ref{card.prop}, this could be 
achieved by taking for all $j \geq 0$, $q_j \in [\theta_1 \alpha_j, \theta_2 \alpha_j]$ for some $\theta_1$,  $\theta_2$. 
With this choice and the choice $\XX_{j+1}=\rho(\Phi^{(j)})$ (where $\Phi^{(j)}$ is 
the random  s.o.f on $\XX_j$ associated to $\LL_j$) , one has for instance, 
\[\forall j \geq 0\, , \, \, 
\va{\XX_j} \frac{q_j}{q_j +\alpha_j} \leq \E_{q_j} \cro{\va{\XX_{j+1}} | {\XX_j}}  
\leq \va{\XX_j} - \va{\RR_{j,r}} \pare{1-\frac{r \alpha_j }{q_j+2 \alpha_j}} \, ,
\]
where $ \RR_{j,r} = \acc{x \in \XX_j; w_j(x) \geq r \alpha_j}$. 
 \item To ensure numerical stability of the reconstruction operator, we would like to 
 control at each step the norms of $\bR_j$ and $\brR_j$. Concerning $\bR_j$ (see
 \eqref{bR-norm.eq}), this requires $\alpha_{j+1}/q'_j$ and $\alpha_j/\beta_j$ to be
 small. As for $\brR_j$ (see \eqref{brR-norm.eq}), one has moreover to ensure 
 that $q'_j/\gamma_j$ is small.  Note that once $\XX_j$ has been fixed, 
 $\alpha_{j+1}$, $\beta_j$ and $\gamma_j$  only depend on the choice of $\XX_{j+1}$, hence of 
 $q_j$. Since 
 $\alpha_{j+1}/q'_j$ and $q'_j/\gamma_j$ should be small, the product 
 $\frac{\alpha_{j+1}}{q'_j} \frac{q'_j}{\gamma_j}= \frac{\alpha_{j+1}}{\gamma_j}$ should
 also be small, and one possible choice for $q_j$ is to minimize 
 $\alpha_{j+1}/\gamma_j$. 
 \item In order to get a good approximation error, one would also like the error in the 
 intertwining problem between $\LL_j$ and $\LL_{j+1}$ to be small. Referring to 
 \eqref{jackson.ineq}, this is achieved if at each step $q'_j (\alpha_j/\beta_j)^{1/p^{*}}$ is 
 small with respect to the natural unit $\alpha_j$.
 \end{enumerate}
 To sum up, $\XX_j$ being chosen, one would like to choose $(q_j,q'_j)$ such that
 \begin{description}
 \item[(C1)] $q_j \in [\theta_1 \alpha_j, \theta_2 \alpha_j]$;
 \item[(C2)] $\alpha_j/\beta_j$ is small;
 \item[(C3)] $\alpha_{j+1}/q'_j$ is small;
 \item[(C4)] $q'_j/\gamma_j$ is small; 
 \item[(C5)] $q'_j \alpha_j/\beta_j $ is small (consider the case $p=+\infty$)  with respect to $\alpha_j$.
 \end{description}
 In this respect, we need some estimates on $\alpha_{j+1}$, $\gamma_j$ and 
 $\beta_j$ in terms of $q_j$. The next sections are devoted to this task. 
 
 \subsection{Estimate on $\bar{\alpha}$.}
 \begin{prop}
 \label{alphabar.prop}
 Assume that $\bX = \rho(\Phi)$ and let $\brX= \XX \setminus \bX$.  Then, for any 
 $m \in  \acc{1, \cdots, \va{\XX}+1}$,  
 \begin{equation}
 \label{baralpham.eq} 
 \E_q\cro{\ba \ind_{\va{\bX}=m}} \geq
  \E_q \cro{\frac{1}{\va{\bX}} \sum_{\bx \in \bX} \bw(\bx) \ind_{\va{\bX}=m} }
  =  q \frac{\va{\XX}-m+1}{m}  \P_q\cro{\va{\bX}=m-1} \, ;
  \end{equation}
 \begin{equation}
 \label{baralpha.eq}
   \E_q\cro{\ba } \geq
  \E_q \cro{\frac{1}{\va{\bX}} \sum_{\bx \in \bX} \bw(\bx) }
  = q \E_q \cro{\frac{\va{\brX}}{\va{\bX} + 1}} \, . 
 \end{equation}
 \end{prop}
 Note that \eqref{baralpham.eq} is trivially true for $m=\va{\XX}+1$. It is also trivially true
 for $m=1$, since in this case $\bLL$ is a scalar, which is equal to $0$ since 
 $\bLL$ is a Markov generator. 
 
 Assuming that $\ba = \max_{\bx \in \bX} \bw(\bx)$ is not too far from  the mean 
 of the $(w(\bx), \bx \in \bX)$, Proposition \ref{alphabar.prop} can be used to get an 
 idea of the dependence of $\ba$ with respect to $q$ (or $m$). Indeed, the term
 $\frac{\va{\brX}}{\va{\bX} + 1}$ can be numerically estimated as a function of 
 $q$, since there is a coupling allowing
  to sample the random forest for all the values of $q$ at the same time. Getting an idea
  of $\ba$ as a function of $q$ from the simulations, is more  time consuming since
  it requires the computation of a Schur complement for all values of $q$, and this cannot been
  done inductively since $\rho(\Phi)$ is not an increasing set w.r.t. $q$. 
  
  \begin{proof}
 \eqref{baralpha.eq} is a direct consequence of \eqref{baralpham.eq} after 
 summing over $m$. 
\[ 
 \E_q \cro{\frac{1}{\va{\bX}} \sum_{\bx \in \bX} \bw(\bx) \ind_{\va{\bX}=m} } 
 = \frac{1}{m} \sum_{x \in \XX} \E_q \cro{ \ind_{x \in \bX} \bw(x) \ind_{\va{\bX}=m}}
 \, . 
 \]
Hence, the result is a consequence of the following identity: for any $x \in \XX$, and 
any $m \in \acc{1, \cdots, n+1}$,
\begin{equation}
\label{barwx.eq}
\E_q \cro{ \ind_{x \in \bX} \bw(x) \ind_{\va{\bX}=m}} 
= q \P_q \cro{\va{\bX}=m-1; x \notin \bX} \, . 
\end{equation} 
 The rest of the proof is devoted to proving \eqref{barwx.eq}. Remind that 
 \begin{align*}
  \bw(x) & = -\bLL(x,x)= \sum_{y \in \bX, y \neq x} \bLL(x,y)
 \\
 & =  \sum_{y \in \bX, y \neq x} \LL(x,y) + \LL_{\bX \brX} (-\LL_{\brX \brX})^{-1} \LL_{\brX \bX} (x,y)
 \\
 &  = \sum_{y \in \bX, y \neq x} w(x,y) 
      + \sum_{\begin{array}{l} \scriptstyle{y \in \bX, y \neq x} \\ \scriptstyle{z_1, z_2 \notin \bX} 			\end{array}}           
      w(x,z_1) (-\LL_{\brX,\brX})^{-1}(z_1,z_2) w(z_2,y) 
 \end{align*}         
Remind also from Lemma \ref{green.lem} (see \eqref{green1.eq}) that for $z_1, z_2 \in \brX$,
$(-\LL_{\brX,\brX})^{-1}(z_1,z_2)= G_{\bX} (z_1,z_2)$. This gives an explicit expression of 
$\bw(x)$ as a function of $\bX$. Therefore,
\[ 
\E_q \cro{ \ind_{x \in \bX} \bw(x) \ind_{\va{\bX}=m}}  
= \hspace{-.5cm} \sum_{R  \owns x, \va{R}=m} \P_q \cro{\bX=R} 
\acc{ \sum_{y \in R, y \neq x} w(x,y) 
      + \hspace{-.5cm} \sum_{\sous{y \in R, y \neq x}{z_1, z_2 \notin R}}           
      w(x,z_1) G_{R} (z_1,z_2) w(z_2,y) }
\]
Now, 
\begin{equation}
\label{XbarR.eq}
 \P_q \cro{\bX=R}= \frac{1}{Z(q)}  \sum_{\phi  \, s.o.f, \rho(\phi)=R} q^m w(\phi)
= \frac{q^m}{Z(q)} Z_R(0) \, ,
\end{equation}
 where we set 
$Z_R(0):=\sum_{\phi  \, s.o.f, \rho(\phi)=R} w(\phi)$. In addition, it is proven in \cite{LA} (see Lemma 3.1, or Appendix B of the preprint)
 that for any subset $R$ of $\XX$, and any $z_1, z_2 \notin R$,
\begin{equation}
\label{greenR.eq} G_R(z_1,z_2) = E_{z_1}\cro{\int_0^{H_R^+} \ind_{X(s)=z_2} \, ds }
   = \frac{1}{Z_R(0)} \sum_{\sous{\phi \, s.o.f. \rho(\phi) = R \cup \acc{z_2}}{z_1 \leadsto_{\phi} z_2}} 
   w(\phi)
 \, ,
 \end{equation}
 where $z_1 \leadsto_{\phi} z_2$ means that $z_1$ is contained in the tree of $\phi$, whose root
 is $z_2$. 
 Hence,  
 \[ Z(q) q^{-m} \E_q \cro{ \ind_{x \in \bX} \bw(x) \ind_{\va{\bX}=m}}  = T_1 + T_2 \, , 
 \]
 where 
 \begin{itemize}
 \item 
 $ T_1  \begin{array}[t]{l} := \sum_{R \owns x; \va{R}=m} Z_R(0) \sum_{y \in R; y \neq x} w(x,y)
 \\ = \sum_{\sous{\phi, \rho(\phi)\owns x}{\va{\rho(\phi)}=m}} 
  \sum_{y \in \rho(\phi); y \neq x} w(\phi) w(x,y)
  \\
  = \sum_{(\phi,y) \in \SS_1(m,x)} w(\phi) w(x,y)  
 \end{array}
 $
\\
where we set $\SS_1(m,x):= \acc{(\phi,y) \mbox{ such that } 
\rho(\phi)\owns x, \va{\rho(\phi)}=m, y \in \rho(\phi) , y \neq x} $. 
\item
$ T_2  \begin{array}[t]{l} := \sum_{R \owns x; \va{R}=m} 
 \sum_{\sous{y \in R; y \neq x}{z_1, z_2 \notin R}}  
 \sum_{\sous{\phi \, s.o.f. \rho(\phi) = R \cup \acc{z_2}}{z_1 \leadsto_{\phi}  z_2}} w(x,z_1) w(\phi) w(z_2,y)
 \\
 = \sum_{(R,y,z_1,z_2,\phi) \in \SS_2(m,x)} w(x,z_1) w(\phi) w(z_2,y)
 \end{array}
 $
\\
where we set 
\[ \SS_2(m,x):= \acc{(R,y,z_1,z_2,\phi) \mbox{ such that } 
\begin{array}{l} R \owns x, \va{R}=m, y \in R , y \neq x, 
\\ z_1 \notin R, z_2 \notin R, \rho(\phi) = R \cup \acc{z_2} , z_1 \leadsto_{\phi} z_2
\end{array}} \, .
\] 
 \end{itemize}
 
 Let us compute $T_1$. To any $(\phi,y) \in \SS_1(m,x)$, we associate the forest 
 $\phi_1 = \phi \cup \acc{(x,y)}$, so that $w(\phi_1) = w(\phi) w(x,y)$.  $\phi_1$ is the forest 
 obtained from $\phi$ by hanging the tree with root $x$ to the root $y$. Hence 
 $\va{\rho(\phi_1)}= m-1$, and the distance $d(x,\rho(\phi_1))$ from $x$ to the roots of 
 $\phi_1$, is equal to 1. Otherwise stated,  $\SS_1(m,x)$ is
 sent by this operation to 
 \[ \FF_1(m,x) := \acc{\phi_1 \mbox{ such that } \va{\rho(\phi_1)}= m-1 ,
 d(x,\rho(\phi_1))=1} \, .
 \]
 Note that this correspondence is one to one. Starting from $\phi_1 \in  \FF_1(m,x)$, we recover
 $y$ as the root of the tree containing $x$, and $\phi$ is then obtained by cutting the edge $(x,y)$. 
 Thus,
 \[ T_1 = \sum_{\phi_1 \in \FF_1(m,x)} w(\phi_1) \, . 
 \]
 
 Let us now turn our attention to $T_2$. To any $(R,y,z_1,z_2,\phi) \in \SS_2(m,x)$, 
 we associate the forest 
 $\phi_2 = \phi \cup \acc{(x,z_1),(z_2,y)}$, so that $w(\phi_2) = w(x,z_1) w(\phi) w(z_2,y)$.  
 $\phi_2$ is the forest 
 obtained from $\phi$ by first hanging the tree with root $x$ to $z_1$, so that $x$ is now in the tree
  with root $z_2$, and then $z_2$ is attached to $y$. Hence 
 $\va{\rho(\phi_2)}= \va{\rho(\phi)}-2=m-1$, and  $d(x,\rho(\phi_2)) \geq 2$ (it could happen that 
 $z_1=z_2$).   Otherwise stated,  $\SS_2(m,x)$ is
 sent  to 
 \[ \FF_2(m,x) := \acc{\phi_2 \mbox{ such that } \va{\rho(\phi_2)}= m-1 ,
 d(x,\rho(\phi_2)) \geq 2} \, .
 \]
 Note that this correspondence is one to one. Starting from $\phi_2 \in  \FF_2(m,x)$, we recover
 $y$ as the root of the tree containing $x$; $z_1$ is the point following $x$ in the path going from
 $x$ to $y$; $z_2$ is the point preceding $y$ in the path going from $x$ to $y$; $R= \rho(\phi_2) 
 \cup \acc{x}$, and $\phi$ is obtained from $\phi_2$ by cutting the edges $(x,z_1)$, $(z_2,y)$. 
 Thus,
 \[ T_2 = \sum_{\phi_2 \in \FF_2(m,x)} w(\phi_2) \, . 
 \]
It is clear that $\FF_1(m,x)$ and $\FF_2(m,x)$ are disjoint sets and that 
$$ \FF_1(m,x) \cup \FF_2(m,x) = \acc{ \phi \mbox{ such that } \va{\rho(\phi)}= m-1; 
x \notin \rho(\phi)} \, . 
$$ 
This leads to 
 \begin{align*} 
\E_q \cro{ \ind_{x \in \bX} \bw(x) \ind_{\va{\bX}=m}}  
& = \frac{q^m}{Z(q)} 
\sum_{\sous{\phi, \va{\rho(\phi)}=m-1,}{ x \notin \rho(\phi)}} w(\phi)
\\
& = q \P_q \cro{ \va{\rho(\Phi)}=m-1, x \notin \rho(\Phi)}
 \end{align*}
 This ends the proof of \eqref{barwx.eq} and of Proposition \ref{alphabar.prop}.
    \end{proof}
   
 \subsection{Estimate on $\beta$}
 \begin{prop}
 \label{beta.prop}
 Assume that $\bX = \rho(\Phi)$ and let $\brX= \XX \setminus \bX$.  Then, for any 
 $m \in  \acc{1, \cdots, \va{\XX}}$,  
 \begin{equation}
 \label{betam.eq} 
 \E_q\cro{\frac{1}{\beta} \ind_{\va{\bX}=m}} \geq
 \E_q \cro{\frac{1}{\va{\bX}} \sum_{\bx \in \bX} \sum_{z \in \brX} P(\bx, z) E_z(H_{\bX})  \ind_{\va{\bX}=m} }
  =  \frac{\va{\XX}-m}{\alpha m}  \P_q\cro{\va{\bX}=m} \, ;
  \end{equation}
 \begin{equation}
 \label{beta.eq}
   \E_q\cro{\frac{1}{\beta} } \geq
   \E_q \cro{\frac{\va{\brX}}{\alpha \va{\bX}}} \, . 
 \end{equation}
 \end{prop}

  \begin{proof}
 \eqref{beta.eq} is a direct consequence of \eqref{betam.eq} after 
 summing over $m$. 
\[ 
 \E_q \cro{\frac{1}{\va{\bX}} \sum_{\bx \in \bX} \sum_{z \in \brX} P(\bx, z) E_z(H_{\bX})  \ind_{\va{\bX}=m} } 
 = \frac{1}{m} \sum_{x \in \XX} \E_q \cro{ \ind_{x \in \bX} \sum_{z \in \brX} P(x, z) E_z(H_{\bX}) \ind_{\va{\bX}=m}}
 \, . 
 \]
Hence, the result is a consequence of the following identity: for any $x \in \XX$, and 
any $m \in \acc{1, \cdots, \va{\XX}}$,
\begin{equation}
\label{betax.eq}
\E_q \cro{ \ind_{x \in \bX}  \sum_{z \in \brX} P(x, z) E_z(H_{\bX})  \ind_{\va{\bX}=m}} 
= \frac{1}{\alpha} \P_q \cro{\va{\bX}=m; x \notin \bX} \, . 
\end{equation} 
 The rest of the proof is devoted to state \eqref{betax.eq}. For $x \in \bX$ and $z \in \brX$, 
 $P(x,z) = w(x,z)/\alpha$, and summing over $z_2$ in \eqref{greenR.eq} leads  to:

 \begin{equation} 
 \forall z \notin R \, , \, \, 
 E_z(H_R) = \frac{1}{Z_R(0)} \sum_{y \notin R} 
 \sum_{\sous{\phi \, s.o.f. \rho(\phi) = R \cup \acc{y}}{z \leadsto_{\phi} y}} w(\phi)
 \end{equation}
  Therefore, using \eqref{XbarR.eq},
\[
\begin{array}{l}
\E_q \cro{ \ind_{x \in \bX} \sum_{z \in \brX} P(x, z) E_z(H_{\bX}) \ind_{\va{\bX}=m}}  
\\
\hspace*{1cm} 
 =  \frac{q^m}{\alpha Z(q)} 
 \sum_{R  \owns x, \va{R}=m} 
 \sum_{z \notin R; y \notin R} 
 \sum_{\sous{\phi \, s.o.f. \rho(\phi) = R \cup \acc{y}}{z \leadsto_{\phi} y}} w(x,z) w(\phi)
\\
\hspace*{1cm} 
 =  \frac{q^m}{\alpha Z(q)}  
\sum_{(R,y,z,\phi) \in \SS_3(m,x)} w(x,z)w(\phi)   
\, , 
 \end{array}
\]

where
\[ \SS_3(m,x):= \acc{(R,y,z,\phi) \mbox{ such that } 
R \owns x, \va{R}=m, y \notin R, z \notin R; \rho(\phi)=R \cup \acc{y}, z \leadsto_{\phi} y} 
\, . 
\] 
  To any $(R,y,z,\phi) \in \SS_3(m,x)$,
 we associate the forest 
 $\phi_3 = \phi \cup \acc{(x,z)}$, so that $w(\phi_3) = w(x,z) w(\phi)$.  
 $\phi_3$ is the forest 
 obtained from $\phi$ by  hanging the tree with root $x$ to $z$, so that $x$ is now in the tree
  with root $y$. Hence 
 $\va{\rho(\phi_3)}= \va{\rho(\phi)}-1=m$, and  $d(x,\rho(\phi_3)) \geq 1$ (it could happen that 
 $z=y$).  $\SS_3(m,x)$ is thus
 sent  to 
$ \FF_3(m,x) := \acc{\phi \mbox{ such that } \va{\rho(\phi)}= m ,
 x \notin \rho(\phi)}$.  
 Note that this correspondence is one to one. Starting from $\phi_3 \in  \FF_3(m,x)$, we recover
 $y$ as the root of the tree containing $x$ in $\phi_3$; $z$ is the point following $x$ in the path going from
 $x$ to $y$ in $\phi_3$;  $\phi$ is obtained from $\phi_3$ by cutting the edge $(x,z)$ so that $x$ is a root of $\phi$, and $R= \rho(\phi) \setminus \acc{y}$. 
 Thus,
 \[ \E_q \cro{ \ind_{x \in \bX} \sum_{z \in \brX} P(x, z) E_z(H_{\bX}) \ind_{\va{\bX}=m}}
 = \frac{q^m}{\alpha Z(q)} \sum_{\phi \in \FF_3(m,x)} w(\phi) 
 = \frac{1}{\alpha} \P_q\cro{\va{\rho(\Phi)}=m; x \notin \rho(\Phi)} 
 \, . 
 \]
 This ends the proof of \eqref{betax.eq} and of Proposition \ref{beta.prop}.
    \end{proof}

\subsection{Estimate on $\gamma$}
 \begin{prop}
 \label{gamma.prop}
 Assume that $\bX = \rho(\Phi)$ and let $\brX= \XX \setminus \bX$.  Then, for any 
 $m \in  \acc{0, \cdots, \va{\XX}}$,  
 \begin{equation}
 \label{gammam.eq} 
 \E_q\cro{\frac{1}{\gamma} \ind_{\va{\bX}=m}}  \geq
 \E_q \cro{\frac{1}{\va{\brX}} \sum_{\brx \in \brX}  E_{\brx}(H_{\bX})  \ind_{\va{\bX}=m} }
  =   \frac{\va{\XX}}{\va{\XX}-m} \frac{1}{q} \P_q\cro{\va{\bX}=m+1} \, ;
  \end{equation}
 \begin{equation}
 \label{gamma.eq}
   \E_q\cro{\frac{1}{\gamma} } \geq
  \frac{1}{q}  \E_q \cro{\frac{\va{\XX}}{\va{\brX} + 1}} \, . 
 \end{equation}
 \end{prop}
  Note that \eqref{gammam.eq} is trivially true for $m=0$, since in this case $1/\gamma=+\infty$. It is also true
 for $m=\va{\XX}$ with the convention $0/0=0$ in the right hand side, since in this case $1/\gamma=0$. 
 
 \begin{proof}
Here again, \eqref{gamma.eq} is a direct consequence of \eqref{gammam.eq} after 
 summing over $m$. 
\[ 
\begin{array}{ll}
 \E_q \cro{\frac{1}{\va{\brX}} \sum_{\brx \in \brX}  E_{\brx}(H_{\bX})  \ind_{\va{\bX}=m} } 
\\[.2cm]
\hspace*{1cm}
 = \frac{1}{\va{\XX}-m} \sum_{x \in \XX} \E_q \cro{ \ind_{x \notin \bX}  E_x(H_{\bX}) \ind_{\va{\bX}=m}}
 \\[.2cm]
\hspace*{1cm}
= \frac{1}{\va{\XX}-m} \sum_{x \in \XX} \sum_{R; x \notin R, \va{R}=m} \P_q(\bX= R) E_x(H_R)
\\[.2cm]
\hspace*{1cm}
 = \frac{q^m}{(\va{\XX}-m) Z(q)} \sum_{x \in \XX} \sum_{R; x \notin R, \va{R}=m}  \sum_{y \notin R}
 \sum_{\sous{\phi \, s.o.f. \rho(\phi) = R \cup \acc{y}}{x \leadsto_{\phi} y}} w(\phi)
\\[.2cm]
\hspace*{1cm}
= \frac{q^m}{(\va{\XX}-m) Z(q)} \sum_{x \in \XX}  \sum_{\phi \, s.o.f. \va{\rho(\phi)} = m+1} w(\phi)
\\[.2cm]
\hspace*{1cm}
= \frac{1}{q(\va{\XX}-m)} \sum_{x \in \XX}  \sum_{\phi \, s.o.f. \va{\rho(\phi)} = m+1} \P_q(\Phi=\phi)
\\[.2cm]
\hspace*{1cm}
= \frac{\va{\XX}}{q(\va{\XX}-m)}  \P_q( \va{\rho(\Phi)} = m+1) \, . 
\end{array}
 \]
 This ends the proof of \eqref{gammam.eq} and of Proposition \ref{gamma.prop}.
    \end{proof}

 \subsection{Discussion on the choice of $\bX$, $q$, $q'$}
 \label{discussionq.sec}
 As was said in the beginning of the section, at each step of the multiresolution scheme, one would want to choose $(q,q')$ so as to fulfill the constraints (C1-5). If (C3) and (C4) are satisfied, then the quantity
 $\ba/\gamma= (\ba/q')(q'/\gamma)$ has to be small. In the same way, if (C3) and (C5) are satisfied, then 
 $\ba /\beta$ is small. One has therefore to choose $\bX$ at each step in such a way that $\ba/\gamma$
 and/or $\ba/\beta$ are small. A possible choice for $q$ is then to minimize the functions $q \mapsto 
 \E_q(\ba/\gamma)$, or $q \mapsto  \E_q(\ba/\beta)$ in some interval $[\theta_1 \alpha; \theta_2 \alpha]$ (in accordance 
 to (C1)). These functions could be estimated using a Monte Carlo method. However, this can be quite costly, 
 since for instance, the computation of $\ba$ needs the computation of a Schur complement. 
 If one can afford such a computation one time, doing it several times to get a Monte Carlo estimation is time consuming. 
 Instead, in the simulations presented in section \ref{numeric.sec}, we used the estimates \eqref{baralpha.eq} and  \eqref{beta.eq}. Denote
 \[ \tilde{\alpha}(q) := q \E_q \cro{\frac{\va{\brX}}{1+\va{\bX}}} \, ; \, \, 
 \frac{1}{\tilde{\beta}(q)} :=  \frac{1}{\alpha} \E_q \cro{\frac{\va{\brX}}{\va{\bX}}}\, ; \, \, 
 \frac{1}{\tilde{\gamma}(q)} := \frac{1}{q} \E_q \cro{\frac{\va{\XX}}{1 + \va{\brX}}} \, . 
 \] 
 These functions are quite easy to estimate by a Monte Carlo method. Indeed, a 
 coalescence-fragmentation process has been proposed in \cite{LA}, allowing to couple the random forests
 sampled according to $\pi_q$ for various values of $q$. In this process, $t=1/q$ is seen as a time variable. 
 Starting from $t_0$, the process begins to sample a random forest whose law is $\pi_{1/t_0}$, and 
 this forest evolves in a Markovian way with fragmentations or coalescences of trees. The marginal of this 
 forest process $(\Phi(t), t \geq t_0)$ is at each time $t$, equal to $\pi_{1/t}$. We refer the reader to \cite{LA} for further details on the 
 construction. Using this coalescence-fragmentation process, it is easy to sample the random function 
 $q \mapsto \va{\rho(\Phi(1/q))}$,  and to estimate the functions 
 $\tilde{\alpha}(q)$, $\tilde{\beta}(q)$, $\tilde{\gamma}(q)$ by empirical means over a sample. One can then 
 choose $q$ as a minimizer of $q \in  [\theta_1 \alpha; \theta_2 \alpha] \mapsto \tilde{\alpha}(q)/\tilde{\beta}(q)$
 or $q \in  [\theta_1 \alpha; \theta_2 \alpha] \mapsto \tilde{\alpha}(q)/\tilde{\gamma}(q)$. 
 It can however be proven that $\tilde{\alpha}(q)/\tilde{\gamma}(q)$ is a decreasing function of $q$; 
 and the reconstruction formulas show that numerical stability depends, just like our Jackson's inequality, 
 on the product of the $\bigl \| \bR_j \bigr \|$  (rather than $\bigl \| \brR_j \bigr \|$), which is associated with $1/\tilde{\beta}$ (rather than $1/\tilde{\gamma}$).  Therefore we decided in the simulations of section \ref{numeric.sec}, to minimize the function 
 $q \in  [\theta_1 \alpha; \theta_2 \alpha] \mapsto \tilde{\alpha}(q)/ \tilde{\beta}(q)$, with $\theta_1=1/8$ and $\theta_2=1$.  
  
  Turning back to the choice of $q'$, we could choose it so that the two terms of the product $(\ba/q')(q'/\beta)=
  \ba/\beta$ are of the same order. This leads to $q' = \theta'_1 \sqrt{\ba \beta}$. With this choice, one get, using 
  \eqref{bR-norm.eq}, \eqref{baralpha.eq} and \eqref{beta.eq}
  \[ \nor{\bar{R}}_{\infty,\infty}  \leq 1 + \frac{2}{\theta'_1} \sqrt{\frac{\ba}{\beta}}
  \simeq 1 + \frac{2}{\theta'_1} \sqrt{ \frac{q}{\alpha} \frac{\va{\brX}}{1+ \va{\bX}} \frac{\va{\brX}}{ \va{\bX}}}
  \leq  1 + \frac{2 \sqrt{\theta_2}}{\theta'_1}  \frac{\va{\brX}}{ \va{\bX}} \leq \frac{\va{\XX}}{\va{\bX}}
  \]
  if we choose $\theta'_1 = 2 \sqrt{\theta_2}$. This value of $\theta'_1$ ensures that the multiplicative term appearing in 
 Jackson's inequality \eqref{jackson.ineq} for $p=+\infty$ behaves nicely:
 \begin{equation}
 \label{cont_explo.eq}
 \prod_{i=0}^{j-1} \pare{1+2\frac{\alpha_{i+1}}{q'_i}} 
 \leq   \prod_{i=0}^{j-1} \frac{\va{\XX_i}}{\va{\XX_{i+1}}} = \frac{\va{\XX_0}}{\va{\XX_{j}}}
 \, . 
 \end{equation}
Another possible choice (made in the simulations results of section \ref{numeric.sec})  is to ensure \eqref{cont_explo.eq} by setting
\[ 
1+2\frac{\alpha_{i+1}}{q'_i}= \frac{\va{\XX_i}}{\va{\XX_{i+1}}}
\equi 
q'_i = 2 \alpha_{i+1}  \frac{\va{\XX_{i+1}}}{\va{\brX_{i}}}
\, . 
\]
  
 \section{A summary of the procedure, and computational issues.}
 \label{resume}
  \subsection{Summary}
 Fix two real numbers $\theta_1,\theta_2$ ($\theta_1=1/8$ and $\theta_2= 1$ in our simulations).
 Starting from $(\XX, \LL, \alpha,\mu)$,
  one iteration of our multiresolution scheme goes through the following steps:
 \begin{enumerate}
 \item Choice of $q$. Make $N$ i.i.d. draws ($N=1$ in our simulations) of the coalescence-fragmentation forest process 
 $(\Phi^{(i)}_t; 
 t \in [1/(\theta_2 \alpha);1/(\theta_1 \alpha)] , i=1, \cdots, N)$. From this, estimate the functions 
\[
 \tilde{\alpha}(q) \simeq q \frac{1}{N} \sum_{i=1}^N \frac{\va{\XX}-\va{\rho(\Phi^{(i)}_{1/q})}}
 {1+ \va{\rho(\Phi^{(i)}_{1/q})}} \, , 
 \]
 and 
\[ \frac{1}{\tilde{\beta}(q)}  \simeq \frac{1}{N \alpha}   \sum_{i=1}^N \frac{\va{\XX}-\va{\rho(\Phi^{(i)}_{1/q})}}{\va{\rho(\Phi^{(i)}_{1/q})}} \, .
 \]
 Choose $q_{opt}$ as the minimizer  in $[\theta_1 \alpha; \theta_2 \alpha]$ of the estimation of 
 $\frac{\tilde{\alpha}(q)}{\tilde{\beta}(q)}$. 
 \item Once $q_{opt}$ is chosen, draw one forest $\Phi$ from $\pi_{q_{opt}}$. Take 
 $\bX = \rho(\Phi)$, $\brX = \XX \setminus \bX$, $\bar{\mu}= \mu(\cdot | \bX)$.
 \item Compute $\bLL = \LL_{\bX,\bX}- \LL_{\bX,\brX} \pare{\LL_{\brX,\brX}}^{-1} \LL_{\brX,\bX}$, and 
 $\ba= \max_{\bx \in \bX} (- \bLL(\bx,\bx))$. 
\item Choose $q'= 2 \ba \va{\brX}/ \va{\bX}$.  Compute $K_{q'}= q' (q' \Id - \LL)^{-1}$. This can be done by 
a polynomial approximation. The scaling functions are then 
\[ x \mapsto \phi_{\bx}(x)  = \frac{K_{q'}(\bx, x)}{\mu(x)} \, ,  
\]
while the wavelets are 
\[ x \mapsto \psi_{\brx}(x) = \frac{(K_{q'}- \Id)(\brx, x)}{\mu(x)}
\, . 
\]
Compute also the reconstruction operators $\bR$ and $\brR$.
\item Go on with $(\bX,\bLL, \bar{\mu}, \ba)$ in place of  $(\XX,\LL, \mu, \alpha)$.
 \end{enumerate}
 
  \subsection{Numerical issues.} 
  \label{numeric}
  The three main computational issues at each step of the algorithm are:
  \begin{enumerate}
  \item the $N$ sampling of the coalescence-fragmentation forest process  
  $(\Phi^{(i)}_t)_{t \in [\frac{1}{\theta_2 \alpha};\frac{1}{\theta_1 \alpha}]; i =1, \cdots, N}$. As proved in \cite{MAR}, the mean time to sample $\Phi_q$ is
  equal to $\sum_{i=0}^{n-1} \frac{1}{q + \lambda_i}$, corresponding to a number of random jumps for the Markov  chain of order 
  $\sum_{i=0}^{n-1} \frac{\alpha}{q + \lambda_i} \leq \frac{\alpha \va{\XX}}{q}$. Hence, the number of operations to sample $N$ realisations of the forest process is of order $N \va{\XX}/\theta_1$.
 \item the computation of the Schur complement. This requires to compute the inverse $(-\LL_{\brX,\brX})^{-1}$. Note that 
 \begin{equation}
 \label{inv-serie.eq}
  (-\LL_{\brX,\brX})^{-1}= \frac{1}{\alpha} (\Id_{\brX}- P_{\brX,\brX})^{-1}  
 			= \frac{1}{\alpha} \sum_{k=0}^{+ \infty} (P_{\brX,\brX})^k \, . 
\end{equation}
Since $(P_{\brX,\brX})^k(\brx,\bry)$ is the probability for the discrete time Markov Chain to go from $\brx$ to $\bry$ in $k$ steps without entering $\bX$, and since $\bX$ is "well spread", this probability is fast decaying in $k$, so that a good approximation of 
$(-\LL_{\brX,\brX})^{-1}$ is obtained by truncating the series in \eqref{inv-serie.eq} to the first terms. 
  
  The computation of  $(-\LL_{\brX,\brX})^{-1}$ is then fast if the matrix $-\LL_{\brX,\brX}$ is sparse. This may be satisfied by $\LL$
  at the beginning of the algorithm, but even if $\LL$ is sparse, this is usually no more the case for $\bLL$ (see Figure \ref{nosparse.fig}). Therefore, after the computation of $\bLL$, the authors of  \cite{HAM,SFV} perform a sparsification method. In our frame, it would be 
 convenient that the sparsification step does not alter  too much  the error bound we get on 
 $\nor{ (K_{q'})_{\bX,\XX} \LL -\bLL  (K_{q'})_{\bX,\XX}}$. 
 Denote by $\bLL_s$ the sparsified version of $\bLL$, and let us consider the interwining error:
 \begin{align*}
&  \nor{(K_{q'})_{\bX,\XX}  \LL  -\bLL_s(K_{q'})_{\bX,\XX} }_{\infty,\infty} 
 \\
 & \qquad  \leq  \nor{(K_{q'})_{\bX,\XX}  \LL -\bLL (K_{q'})_{\bX,\XX} }_{\infty,\infty} + \nor{(\bLL -\bLL_s) (K_{q'})_{\bX,\XX}}_{\infty,\infty}
 \\
& \qquad  \leq   \nor{(K_{q'})_{\bX,\XX}  \LL -\bLL (K_{q'})_{\bX,\XX} }_{\infty,\infty}  +  \nor{\bLL -\bLL_s}_{\infty,\infty}
\\
 &  \qquad \leq    \nor{(K_{q'})_{\bX,\XX}  \LL -\bLL (K_{q'})_{\bX,\XX} }_{\infty,\infty} 
+ \max_{\bx \in \bX} \sum_{\by \in \bX} \va{\bLL(\bx,\by)-\bLL_s(\bx,\by)}
\, . 
\end{align*}
Therefore, the sparsification should for instance satisfy 
 \[   \max_{\bx \in \bX} \sum_{\by \in \bX} \va{\bLL(\bx,\by)-\bLL_s(\bx,\by)} \leq 
 			\nor{(K_{q'})_{\bX,\XX}  \LL -\bLL (K_{q'})_{\bX,\XX} }_{\infty,\infty}  \, . 
			\]
By building $\bar {\mathcal L}_s$ with this constraint only,
and replacing ${\mathcal L}$ by ${\mathcal L}_s$,
we can however locally deteriorate a good intertwining approximation,
i.e., increase a lot a small error
$$
    \epsilon(\bar x)
    = \sum_{x \in {\cal X}} \bigl|\bigl(K_{q'}{\cal L}\bigr)(\bar x, x)
        - \bigl(\bar{\cal L}(K_{q'})_{\bar{\cal X}, {\cal X}}\bigr)(\bar x, x)\bigr|
$$
by allowing this error to grow up to the order of $\epsilon = \max_{\bar x \in \bar{\cal X}} \epsilon(\bar x)$.
Since, at least from a probabilistic point of view (see \cite{ACGM1} and its appendix),
the collection of these local errors $\epsilon(\bar x)$ is more meaningful
than the global error $\epsilon$ we follow a more local and restrictive approach.

For $\theta \geq 1$ ($\theta = 4$ in our simulation, except in Figure~\ref{sparse.fig}
where we show the effect of different choices of $\theta$ in
our sparsification scheme), we set to 0 a maximal number of coefficients on each row $\bar x$ of $\bar {\cal L}$
until the off-diagonal suppressed weight remains below a targeted error level $\epsilon(\bar x) \bar\alpha / 2 \theta \alpha$
(the coefficients of $\cal L$ and $\bar {\cal L}$ are naturally on scale $\alpha$ and $\bar\alpha$ respectively).
But this has to be done in a symmetric way to ensure 
the symmetry of $\bLL_s$.  In addition one has to make $\bLL_s$ a
Markov generator.  To satisfy all these constraints, we proceed in the following way. 
We first sort in non-decreasing order the off-diagonal pairs $(\bar x, \bar y)$ (with $\bar x \neq \bar y$) according to the quantity 
$\max(\bw(\bx,\by), \bw(\by,\bx))$. We then set $\delta_{\bx} = 0$ for all $\bx \in \bX$ and screen the ordered off-diagonal pairs. 
Each time $(\bx,\by)$ is such that $\delta_{\bx} + \bw(\bx,\by) \leq \epsilon(\bar x)$ and $\delta_{\by} + \bw(\by,\bx) \leq \epsilon(\bar y)$, 
we set to $0$ both $\bw(\bx,\by)$ and $\bw(\by,\bx)$. We also
increase $\delta_{\bx}$  by $\bw(\bx,\by)$, and $\delta_{\by}$ by $\bw(\by,\bx)$. We finally adjust the diagonal coefficients $\bLL(\bx,\bx)$ 
by subtracting them the associated removed weight $\delta_{\bx}$.
This ensures 
\begin{align*}
    \max_{\bx \in \bX} \sum_{\by \in \bX} \va{\bLL(\bx,\by)-\bLL_s(\bx,\by)}
    & \leq {\bar\alpha \over \theta \alpha}\nor{(K_{q'})_{\bX,\XX}  \LL -\bLL (K_{q'})_{\bX,\XX} }_{\infty,\infty} \\
    & \leq \nor{(K_{q'})_{\bX,\XX}  \LL -\bLL (K_{q'})_{\bX,\XX} }_{\infty,\infty}  \, . 
\end{align*}
In doing so, we can lose the irreducibility property, i.e., we can get a disconnected reduced graph associated with $(\bX, \bLL_s)$.
This actually does not raise any particular difficulty.
The only use we made of irreducibility was to define without ambiguity our reference measure $\mu$.
When irreducibility is lost, we simply proceed by using the natural restriction of $\mu$ as  reference measure.  

%
%

 \item the computation of $K_{q'}$. As in \cite{SFV}, this can be done by Chebyshev polynomial approximation. Apart from computational gain, 
using polynomial approximation has the advantage to produce filters  $K_{q'}$ with a good space localization as soon as $\LL$ is nearly diagonal.
  \end{enumerate}

\begin{figure} 
 \begin{tabular}{ccc}
 \includegraphics[scale=0.25]{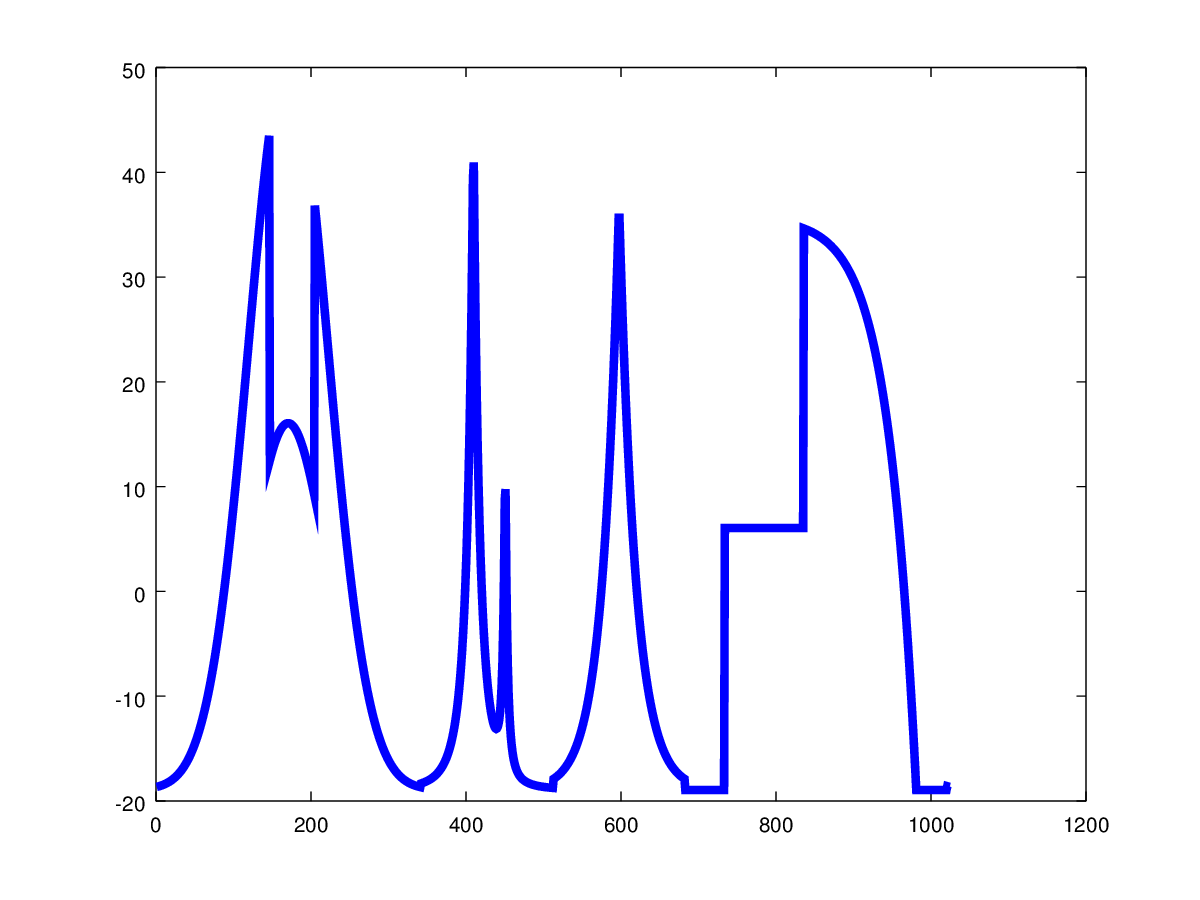}
&
  \includegraphics[scale=0.25]{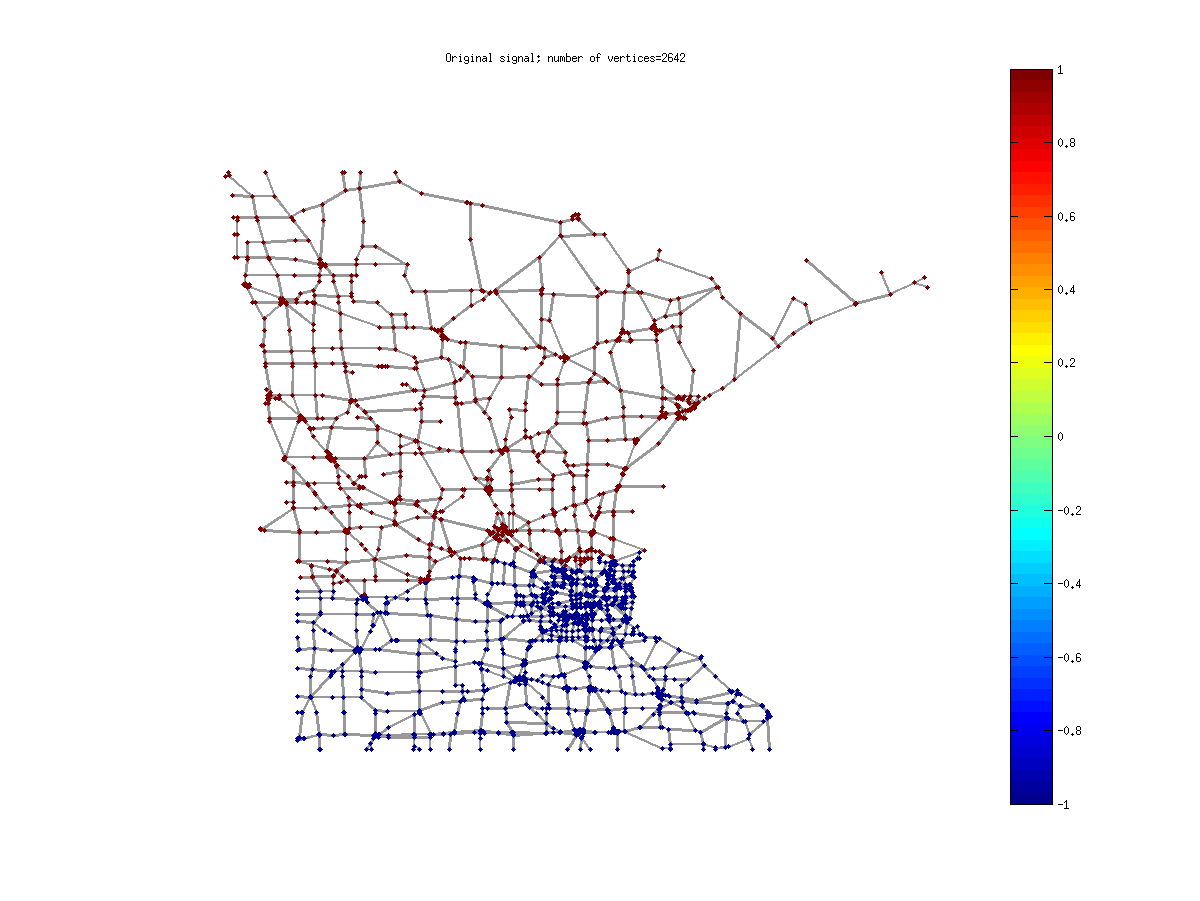}
&
 \includegraphics[scale=0.25]{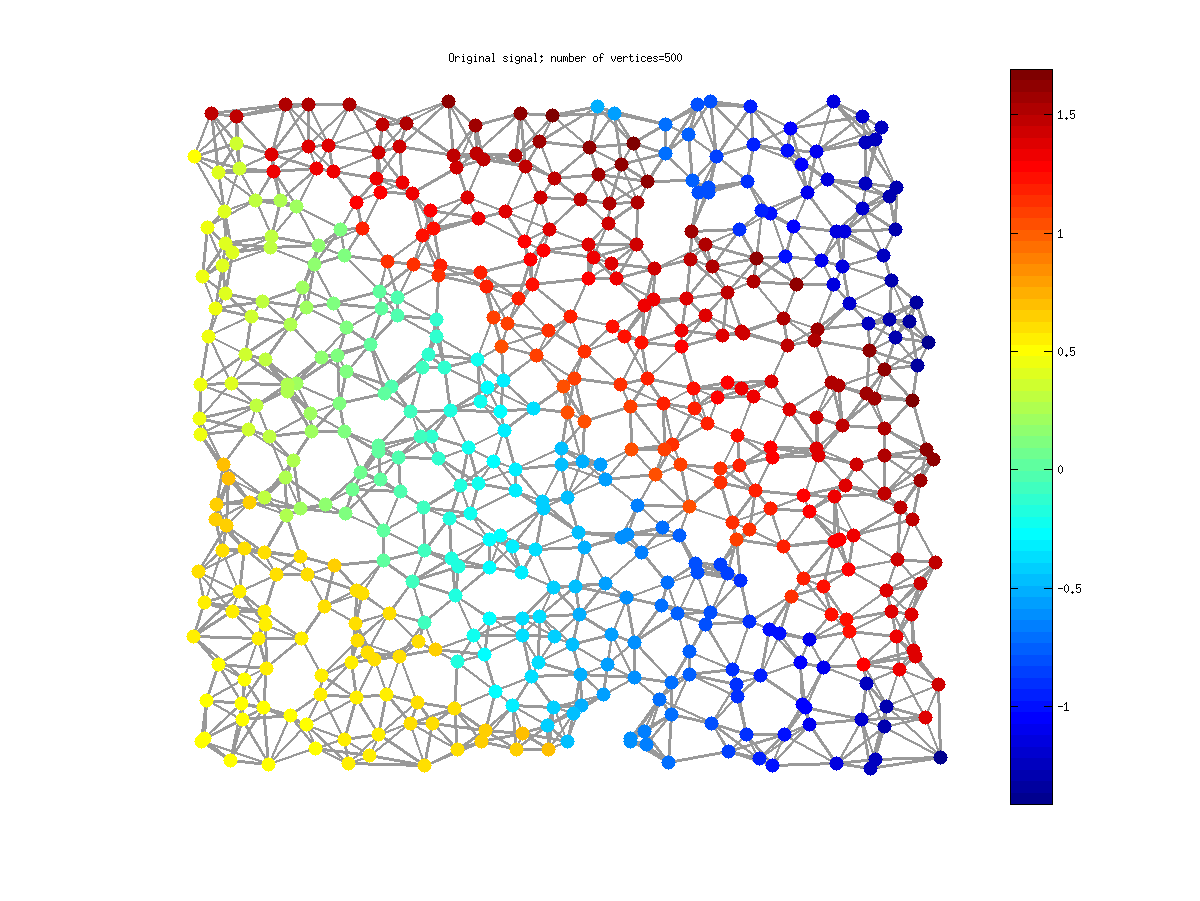} 
  \\
(a) & (b) & (c)
\end{tabular}
\caption{Original signals: (a) Signal on the line presented in \cite{MAL}. (b) Signal on the Minnesota roads network: sign of the first Fourier mode $e_1$. (c) Signal from \cite{SFV} on the sensor graph from the  GSPBox.} 
 \label{signaux}
  \end{figure}

 \section{Numerical Results}
 \label{numeric.sec}
 This section is devoted to numerical illustrations of our multiresolution scheme,  referred to as "the intertwining wavelets multiresolution". 
 We show the results of some downsampling steps on  Minnesota roads network (cf Figures \ref{nosparse.fig} and \ref{sparse.fig}) containing 2642 vertices,
  and use the multiresolution schemes 
 to analyse and compress  the three benchmark signals of Figure \ref{signaux}. 

\subsection{Downsampling of the Minnesota roads network.}
 Figure \ref{nosparse.fig} shows the result of two levels of forest's roots sampling, combined with the weighting procedure through 
 Schur's complement computation without sparsification. It illustrates the loss of sparsity of the weighting procedure. In Figure 
 \ref{sparse.fig}, we used the sparsification method proposed in section \ref{numeric}
with three values of the parameter $\theta$. On these graphs, the width of one edge is proportional to its weight. 
 
 \begin{figure} 
\begin{tabular}{c}
 \includegraphics[height=6cm,width=16cm]{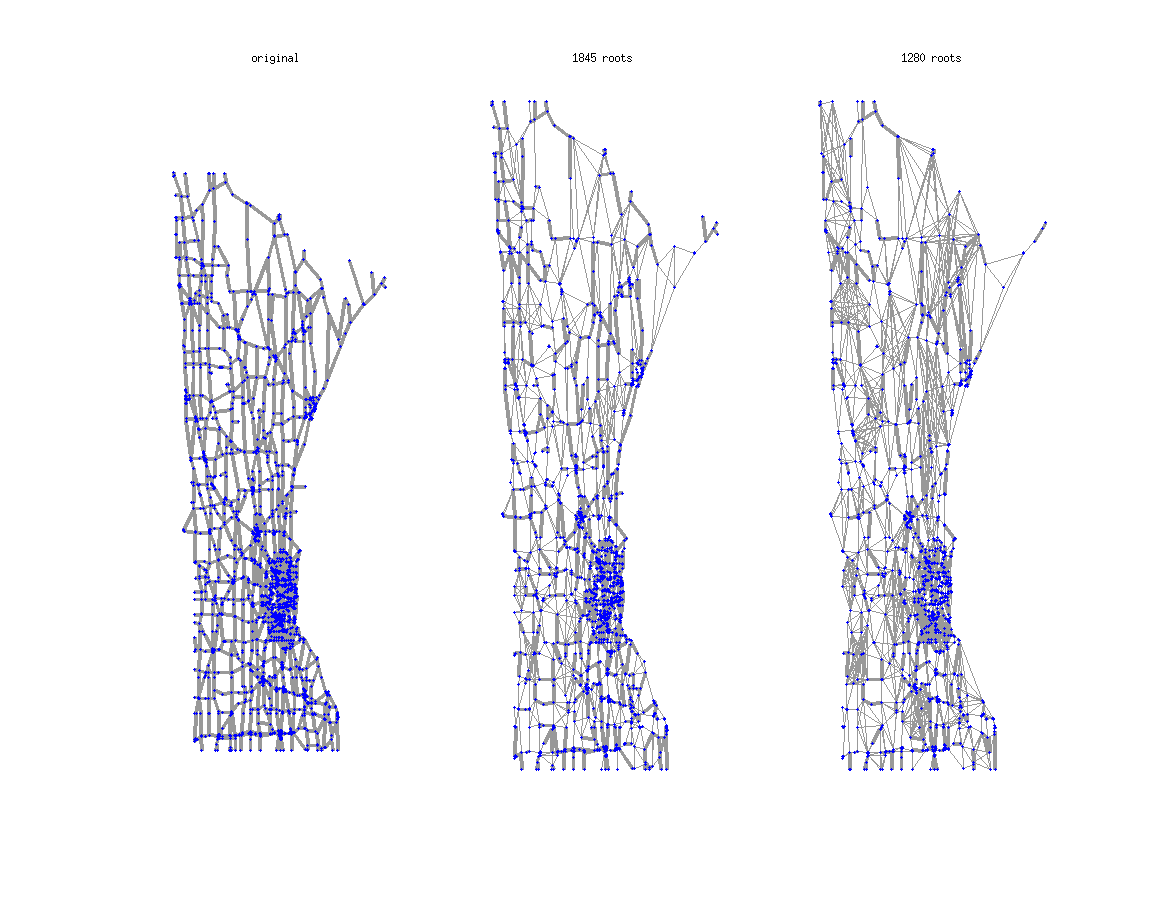}
 \\[-1cm]
(a) \hspace{3.7cm} (b) \hspace{3.7cm} (c) 
\end{tabular}
\caption{ One sequence of Minnesota's subgraphs, without sparsification. (a) Original graph. (b) First level of downsampling. 
(c) Second level of downsampling.}
 \label{nosparse.fig}
  \end{figure}

  \begin{figure} 
 \begin{tabular}{c}
  \includegraphics[height=6cm,width=16cm]{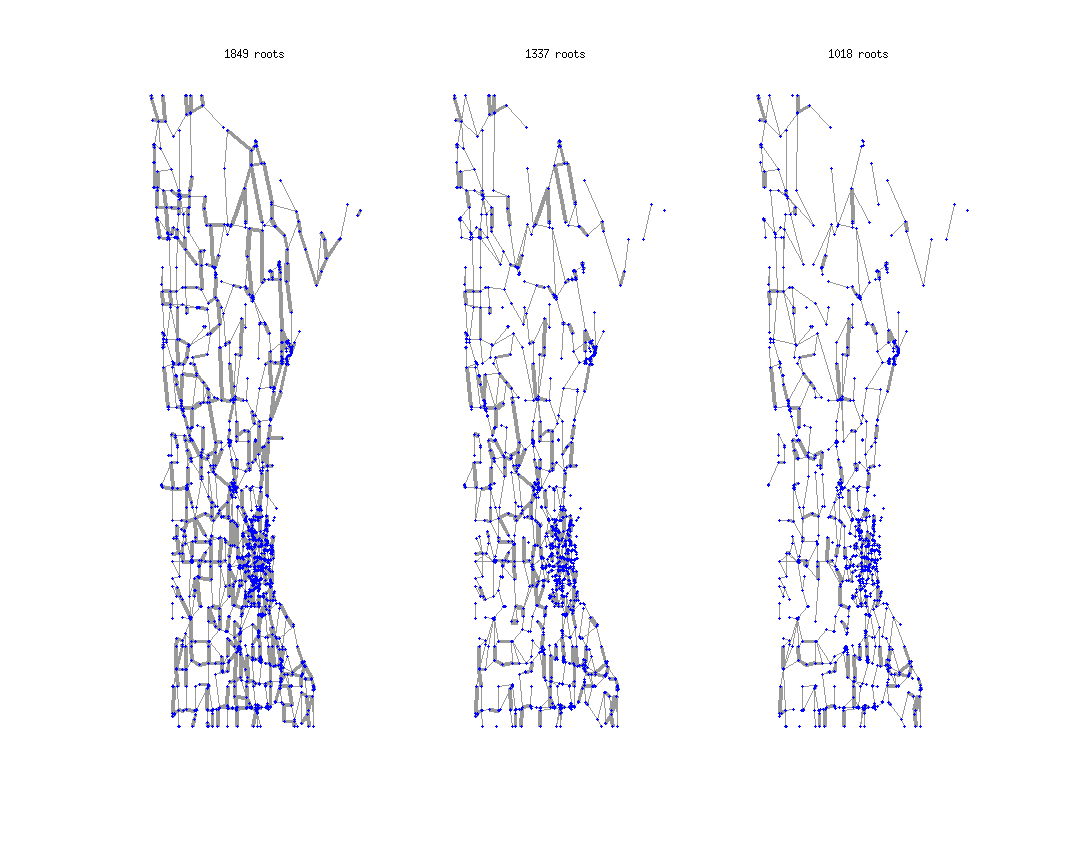}
 \\[-1cm]
  \includegraphics[height=6cm,width=16cm]{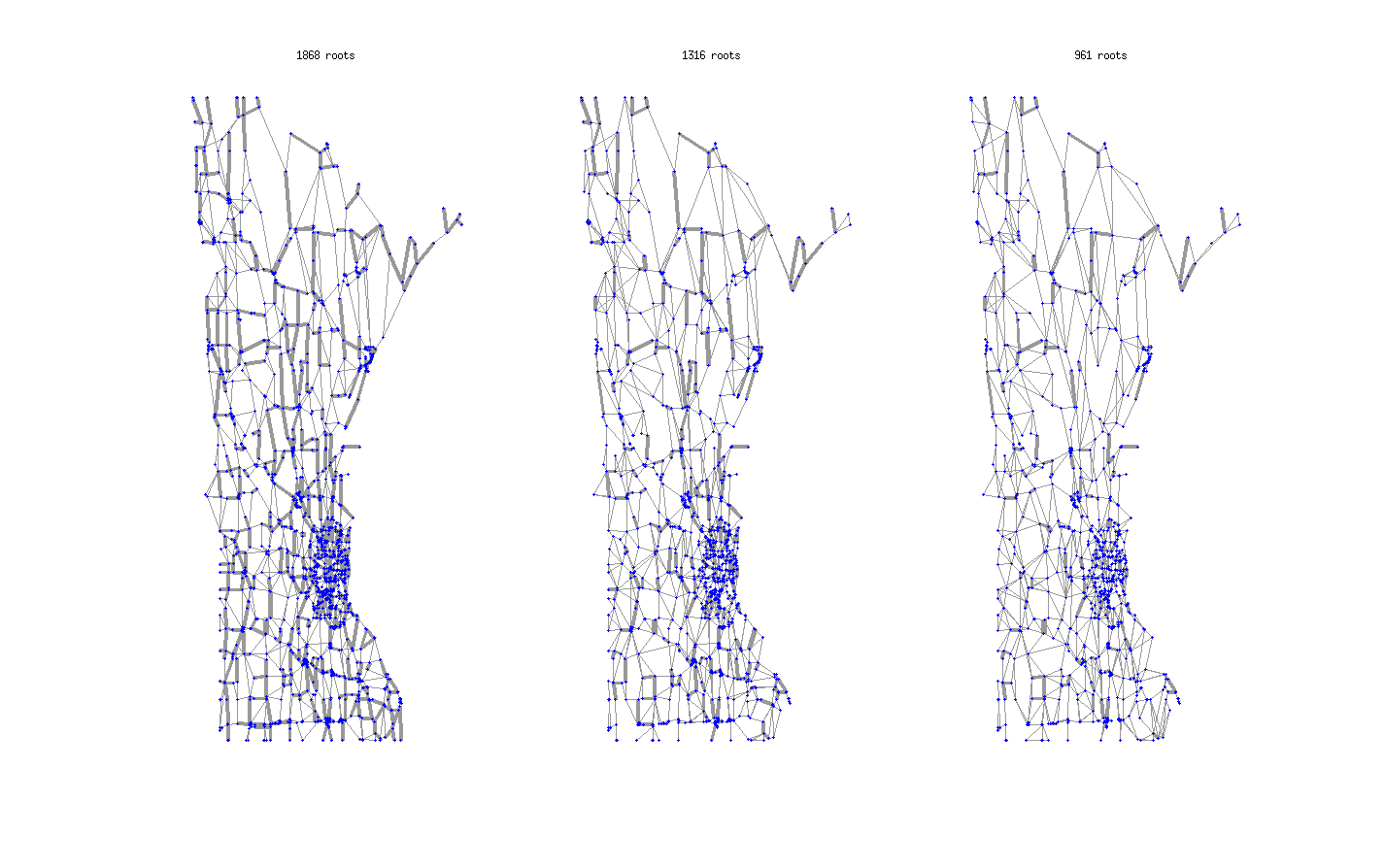}
 \\[-1cm]
   \includegraphics[height=6cm,width=16cm]{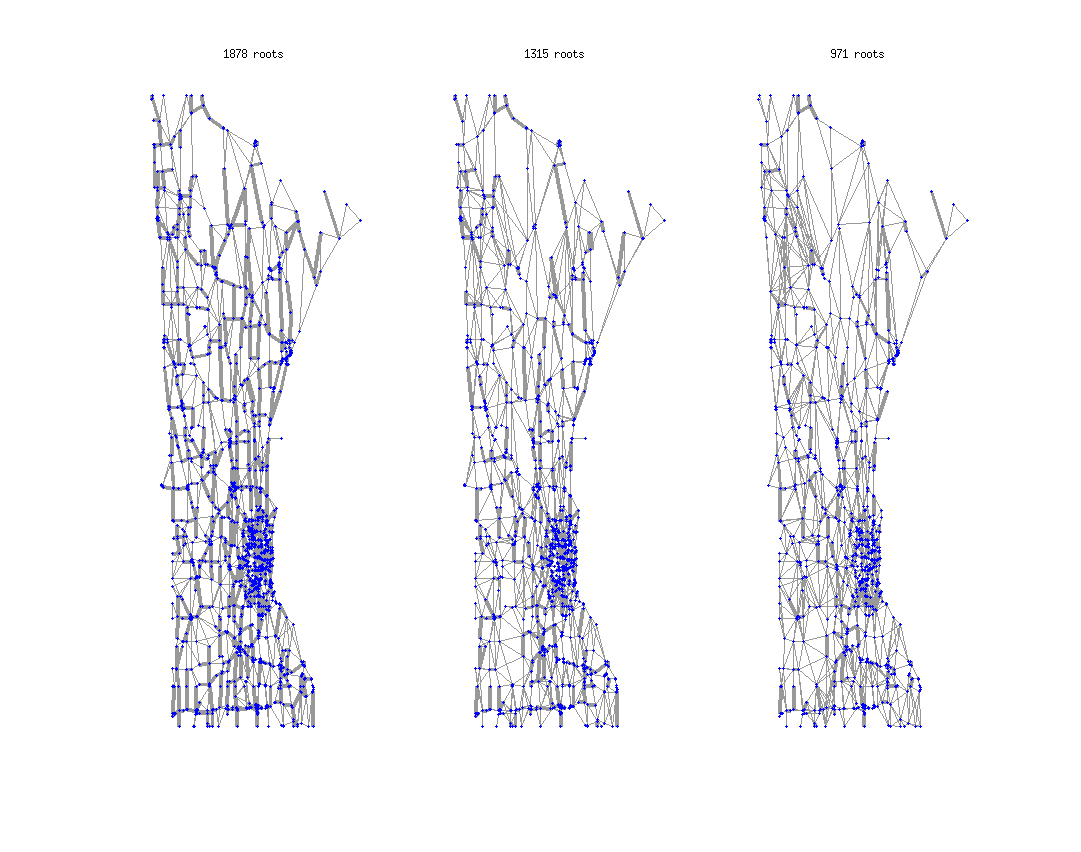}
 \\[-.8cm]
(a) \hspace{3.7cm} (b) \hspace{3.7cm} (c) 
\end{tabular}
\caption{One sequence of Minnesota's subgraphs, using sparsification.
On the first line, $\theta=2$, while $\theta=4$ and $\theta = 8$ on the second and third ones.
(a) First level of downsampling.
(b) Second level of downsampling.
(c) Third level of downsampling.}
 \label{sparse.fig}
  \end{figure}

 \subsection{Analysis.}
\subsubsection{The line.} 
 We analyse the signal of Figure \ref{signaux}(a) using our multiresolution scheme.
The results after one step are presented in Figure \ref{segment_coeff},
where we can see that the big detail coefficients are located at the discontinuities of the original signal.
We also compare our scheme with a classical wavelet scheme involving Daubechies12 wavelets.
The results are given in Figure \ref{segment_comp_coeff}.
After two steps , we end up with 370 approximation coefficients ($f_2$), instead of 256 for the classical wavelets.
In both cases, the number of non vanishing detail coefficients is small.

  \begin{figure}
  \includegraphics[scale=0.6]{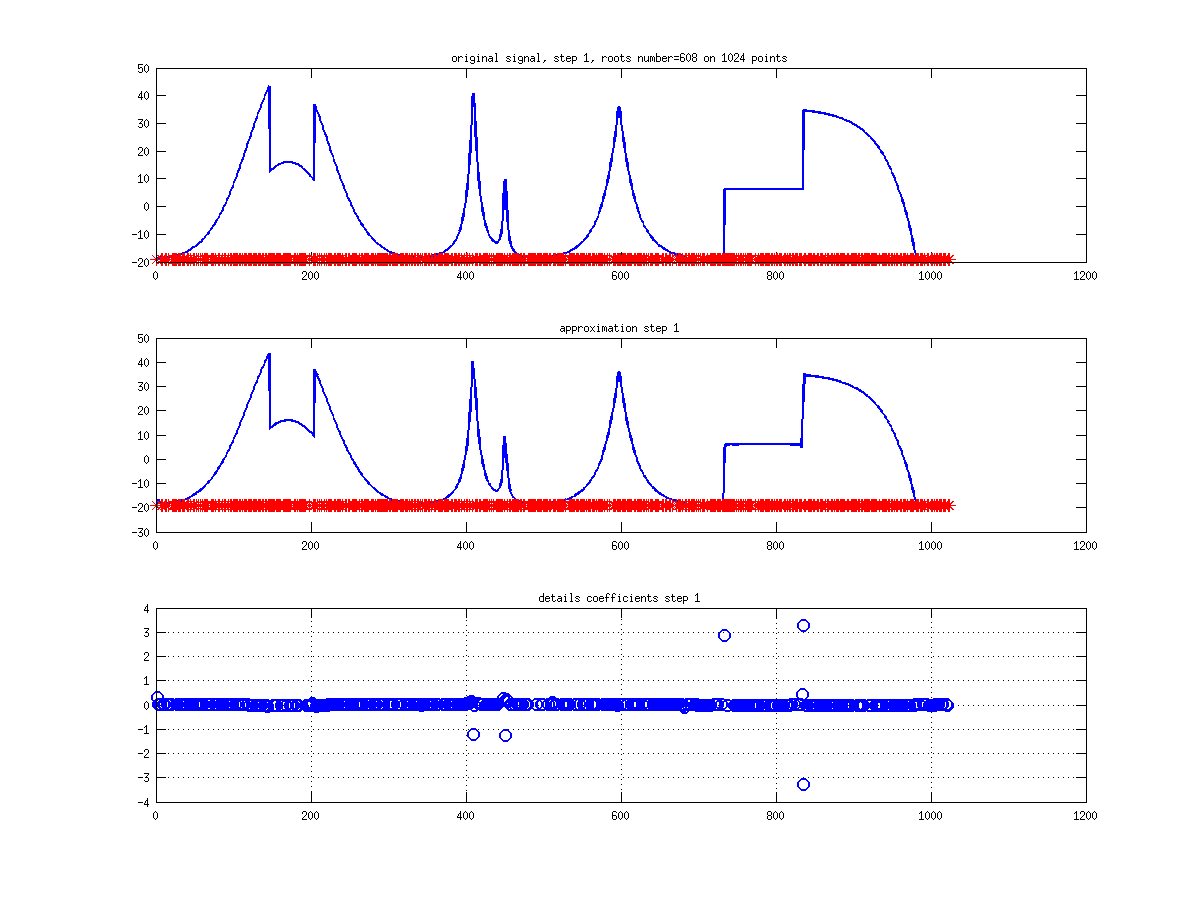}
 \caption{One step of the analysis. The red crosses are the locations of the 608 chosen points in $\bX$, while the blue circles correspond to the points of $\brX$. The first figure is the original signal; the second one gives the approximations coefficients ($f_1$), and the third one the detail coefficients ($g_1$).}
 \label{segment_coeff}
 \end{figure}
 
 \begin{figure}
  \includegraphics[scale=0.6]{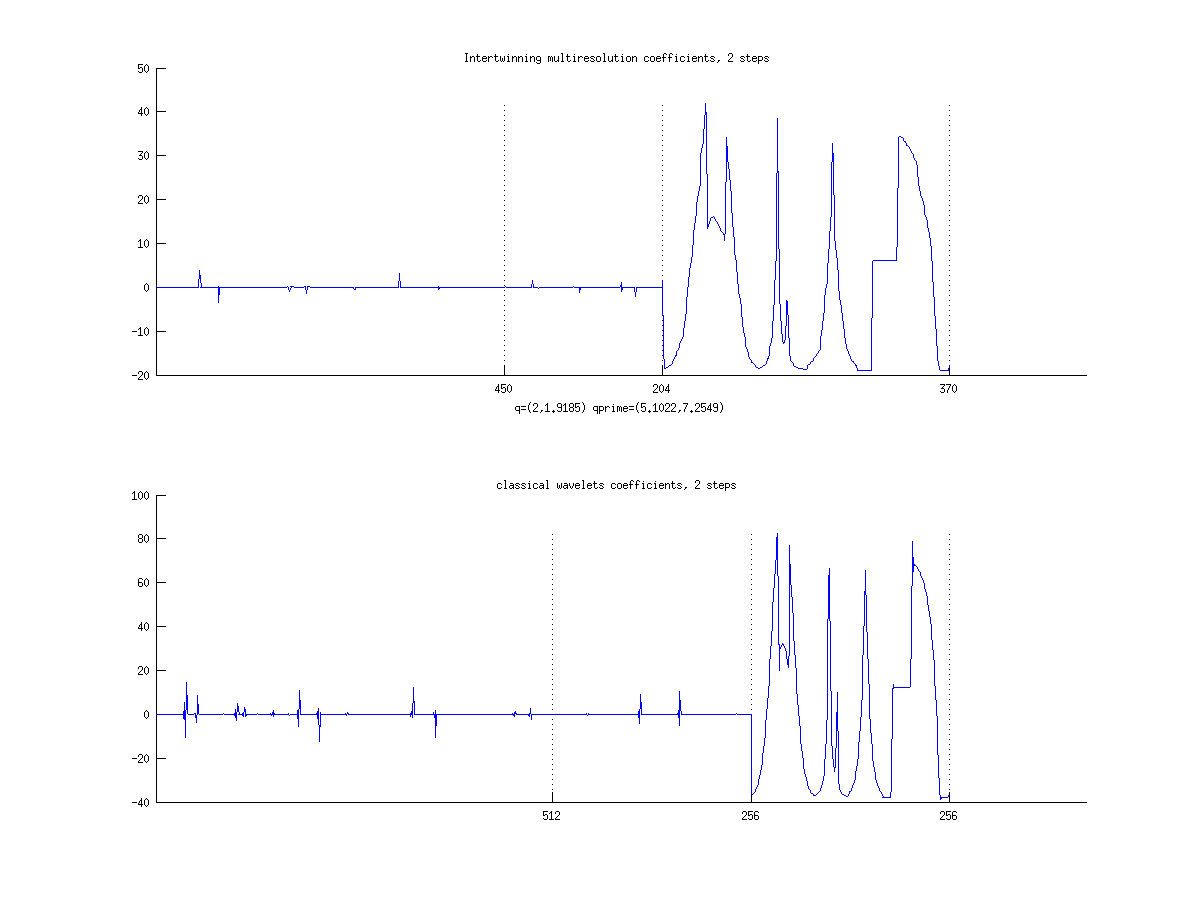}
 \caption{Results of two steps of the multiresolution for the intertwining wavelets (first graph) and Daubechies  12 wavelets. The three parts of a graph 
 give $(g_1,g_2,f_2)$.}
 \label{segment_comp_coeff}
 \end{figure}

 \subsubsection{Minnesota graph} 
 The analysis of the signal of Figure \ref{signaux}(b) after two steps of the interwining wavelets multresolution is presented in  Figure \ref{Minnesota_2steps}. Here again, the big detail coefficients are located at discontinuities of the signal. 
 
 \begin{figure}
 \begin{multicols}{2}
(a) Original signal: sign of $e_1$.

\centerline{\includegraphics[width=6cm,height=4cm]{minnesota_one_step_original_signal}}

(b) Approximation. Size of $f_2$: 1268.

\centerline{\includegraphics[width=6cm,height=4cm]{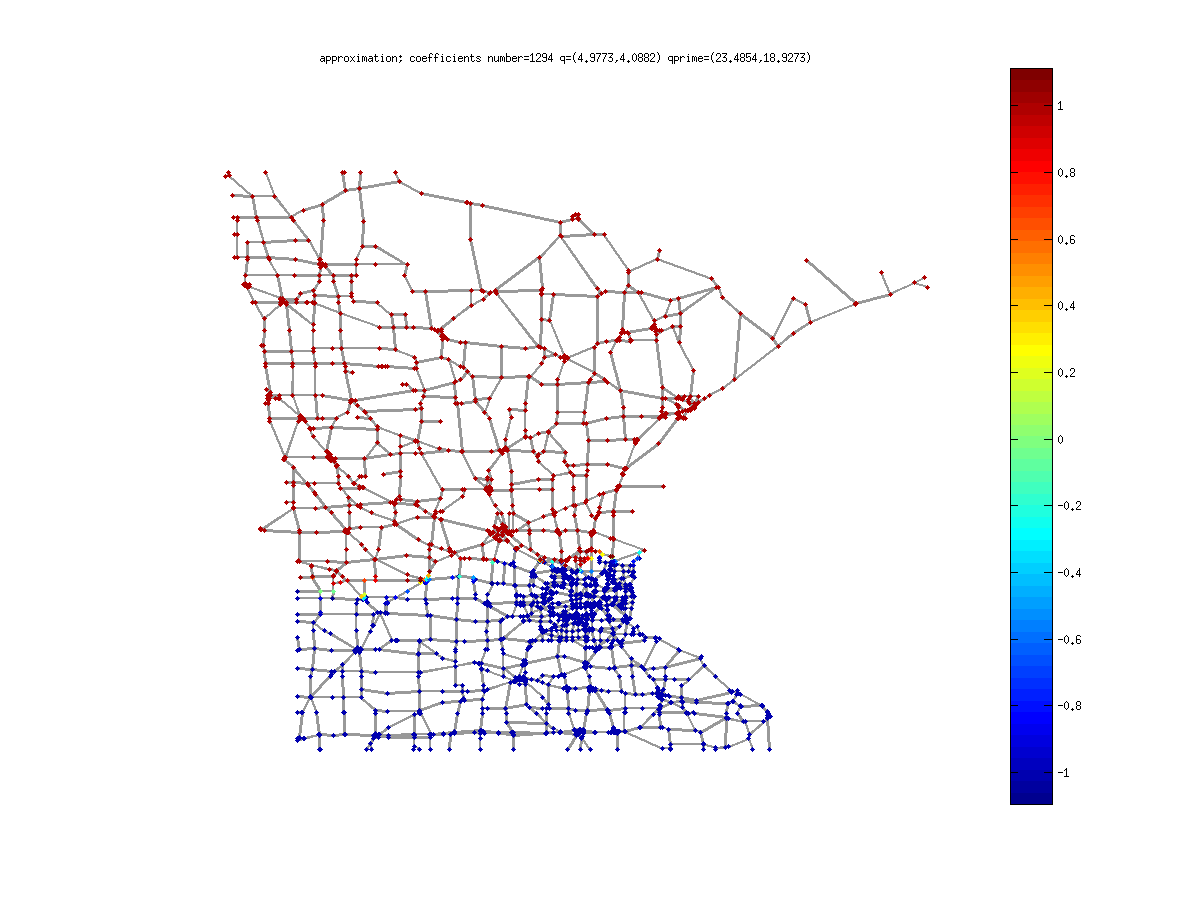}}
\columnbreak

(c) Detail at scale 1.   Size of $g_1$: 800.

\centerline{\includegraphics[width=6cm,height=4cm]{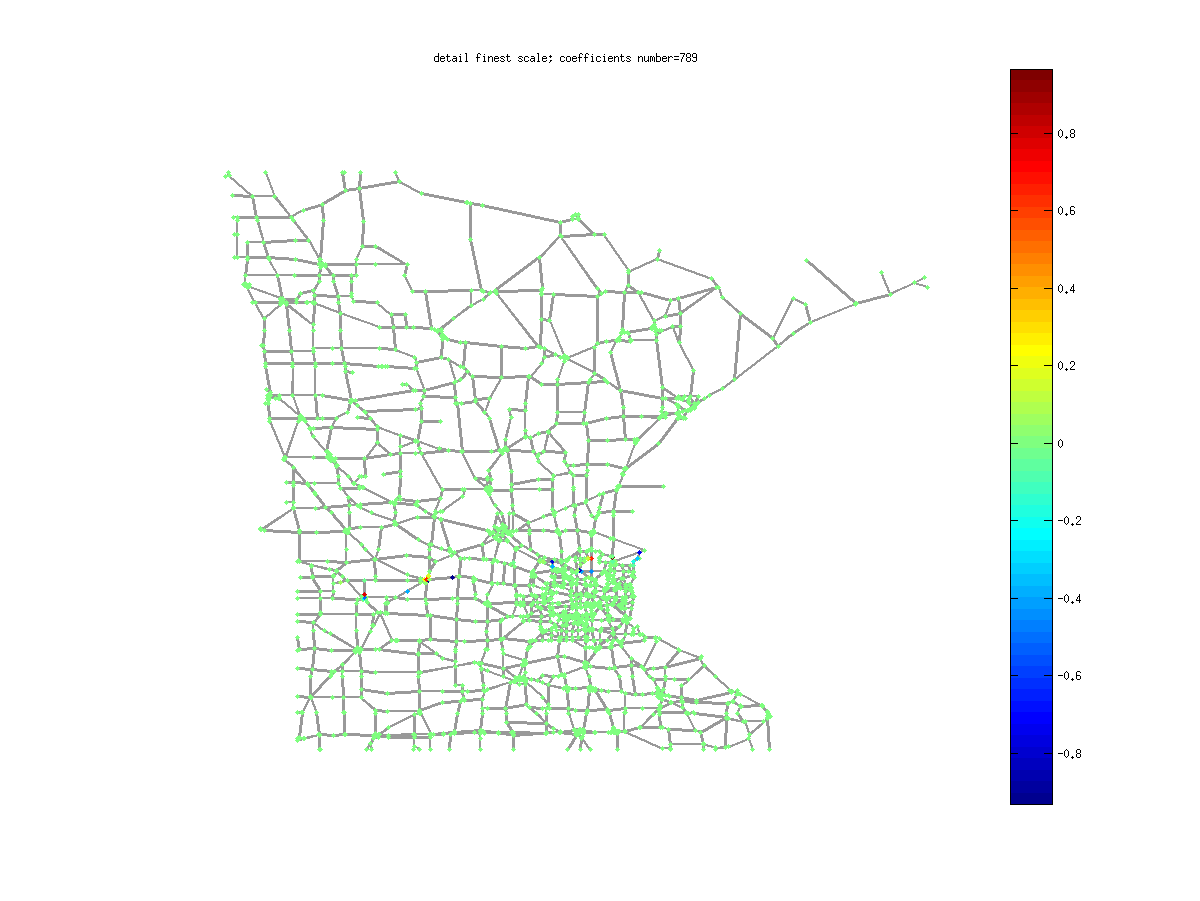}}

(d) Detail at scale 2.  Size of $g_2$: 574.

\centerline{\includegraphics[width=6cm,height=4cm]{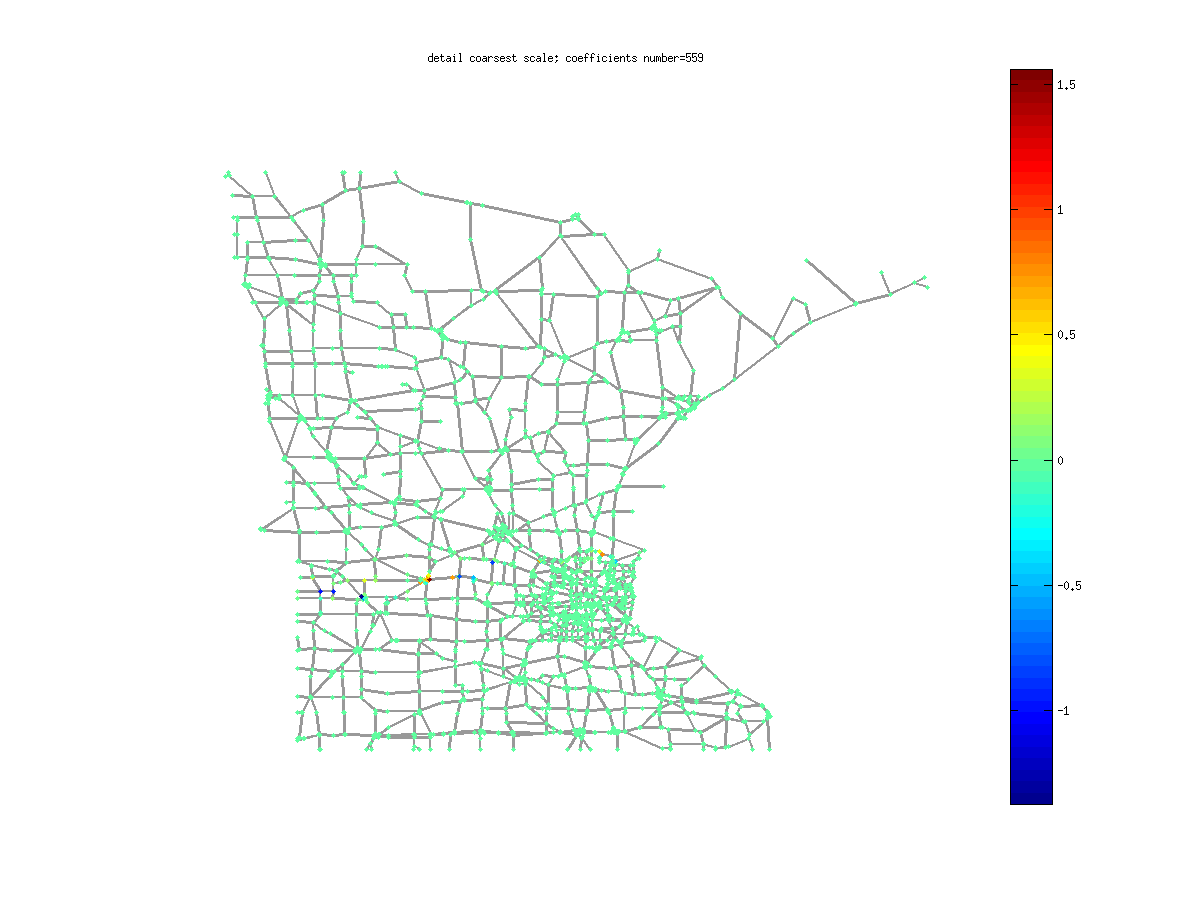}}
\end{multicols}
\caption{Two steps of the intertwining wavelets multiresolution upsampled to the original graph: (a) $f_0$; (b) $\bar{R}_0 \bar{R}_1 f_2$; 
(c) $ \brR_0 g_1$;  (d) $\bar{R}_0 \brR_1 g_2$}.
\label{Minnesota_2steps}
\end{figure}

 \subsubsection{Sensor graph} 
 The analysis of the signal of Figure \ref{signaux}(c) after two steps of the interwining wavelets multresolution is presented in  Figure \ref{sensor_2steps}. 
 
 \begin{figure}
 \begin{multicols}{2}
(a) Original signal.

\centerline{\includegraphics[width=6cm,height=4cm]{sensor_one_step_original_signal}}

(b) Approximation. Size of $f_2$: 206.

\centerline{\includegraphics[width=6cm,height=4cm]{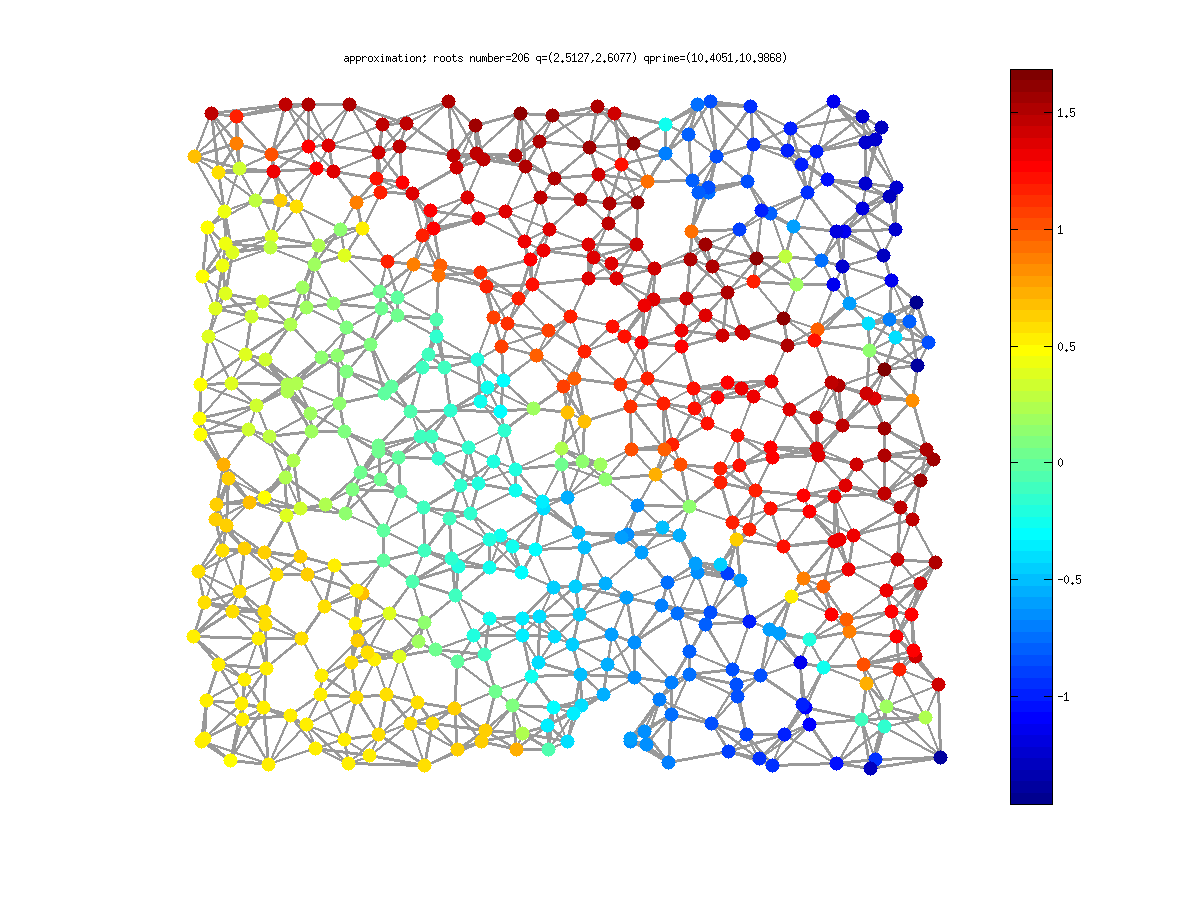}}
\columnbreak

(c) Detail at scale 1.   Size of $g_1$: 188 .

\centerline{\includegraphics[width=6cm,height=4cm]{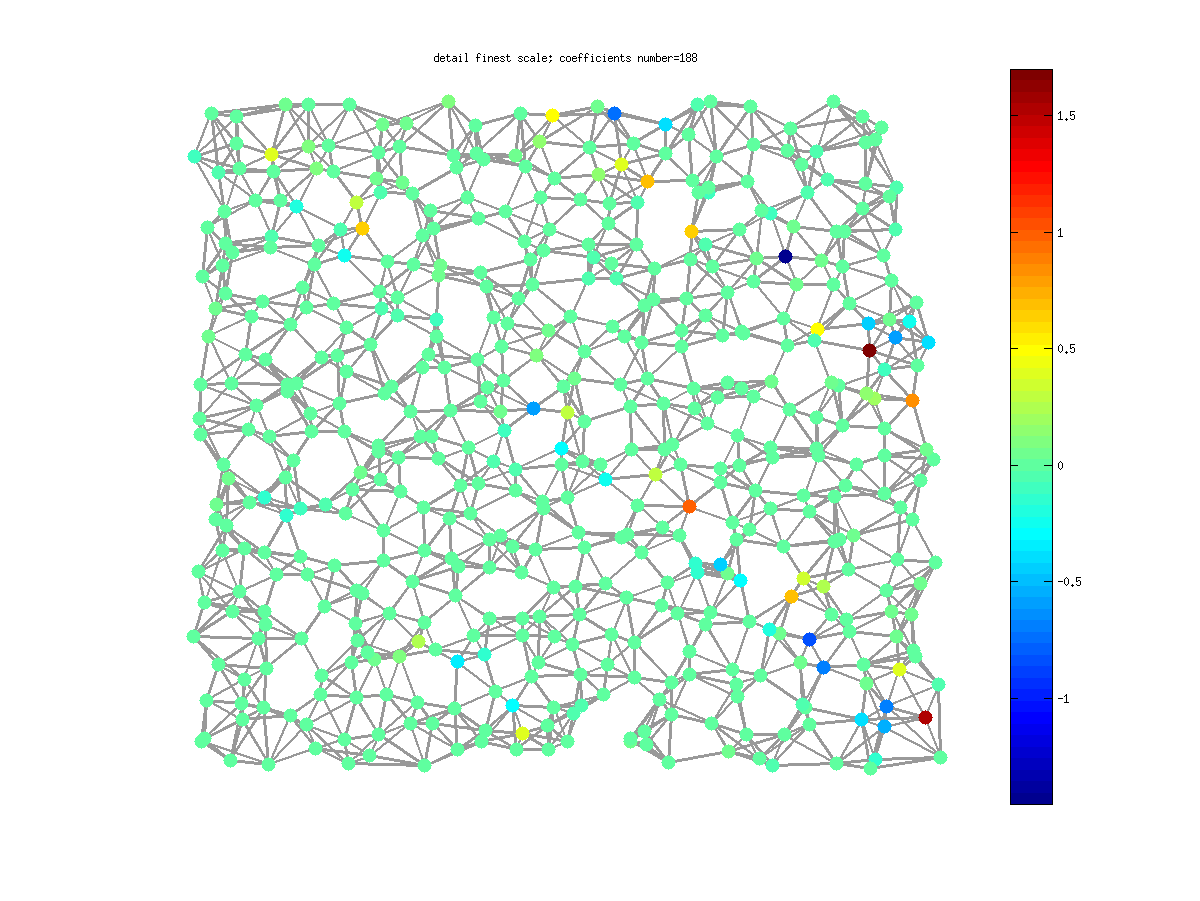}}

(d) Detail at scale 2.  Size of $g_2$: 106.

\centerline{\includegraphics[width=6cm,height=4cm]{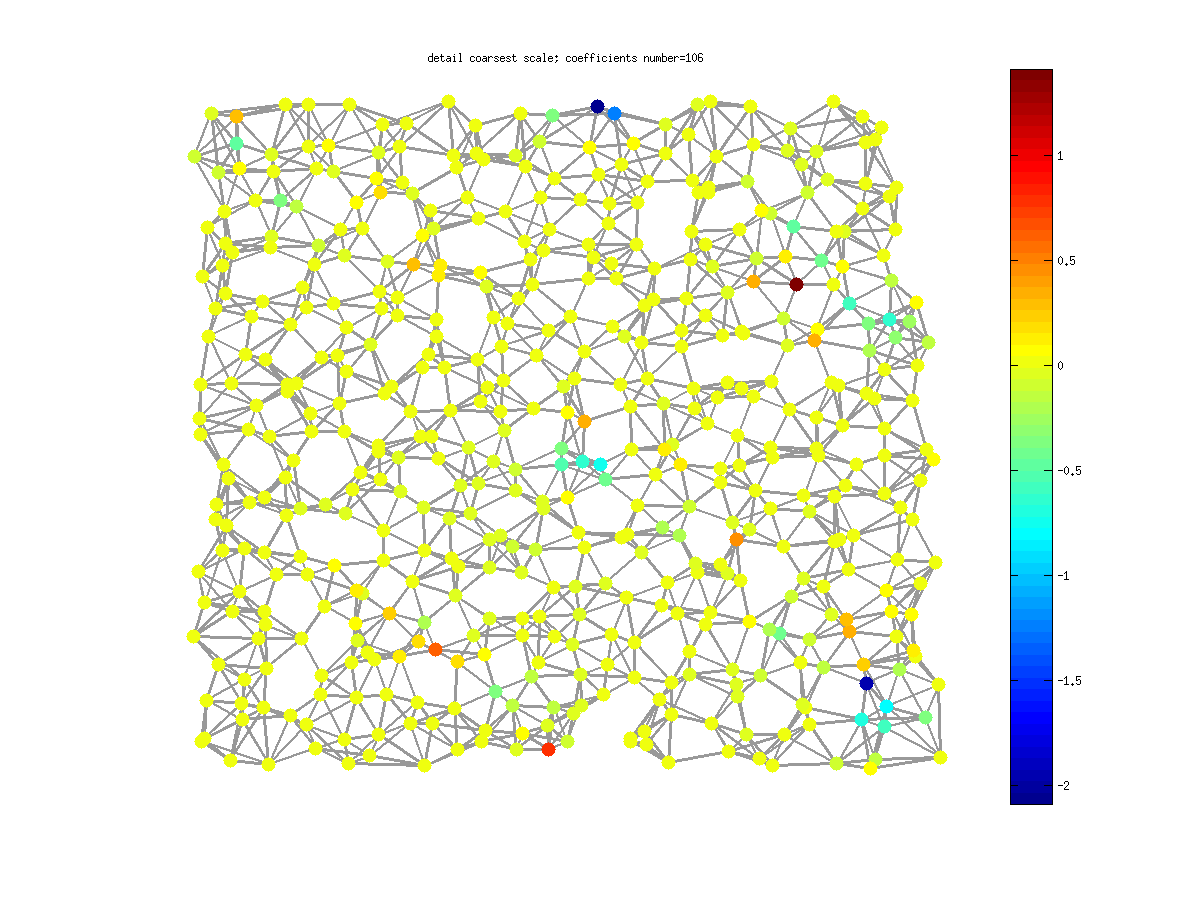}}
\end{multicols}
\caption{Two steps of the intertwining wavelets multiresolution upsampled to the original graph: (a) $f_0$; (b) $\bar{R}_0 \bar{R}_1 f_2$; 
(c) $ \brR_0 g_1$;  (d) $\bar{R}_0 \brR_1 g_2$}.
\label{sensor_2steps}
\end{figure}
 \subsection{Compression.}
 Now, we use the intertwining wavelets  to compress the signals of  Figure \ref{signaux}. Since the  intertwining wavelets
  $\psi_{\brx}$ are unnormalized, we normalize detail coefficients in the compression problem. 
 More precisely, given a signal $f$, the unnormalized coefficients after one multiresolution step are $\baf$ and 
 $\brf$, from which we can reconstruct $f$ by 
 \[ f = \bR \baf + \brR \brf = \sum_{\bx \in \bX} \baf(\bx) \tilde{\xi}_{\bx} + 
 \sum_{\brx \in \brX} \brf(\brx) \tilde{\xi}_{\brx}
 \, , 
 \]
where the dual basis $(\tilde{\xi}_{\bx}, \tilde{\xi}_{\brx})_{\bx \in \bX, \brx \in \brX}$ corresponds to the 
columns of matrices $\bR$ and $\brR$.  Unlike the classical wavelets, the basis $(\xi_x, x \in \XX)$ and its 
dual basis $(\tilde{\xi}_x, x \in \XX)$ are not orthonormal ones. Therefore to compress our signal, we 
 truncate the "normalized coefficients" $\brf(\brx) \nor{ \tilde{\xi}_{\brx}}$. In a similar way, after $k$ 
 multiresolution steps, the unnormalized coefficients are $[f_k, g_k, g_{k-1}, \cdots, g_1]$, 
 from which we can reconstruct $f$ by:
\begin{align}
  f  & =  \bR_0 \cdots \bR_{k-1}f_k + \sum_{j=0}^{k-1} \bR_0 \cdots \bR_{j-1} \brR_j g_j
\\ 
& = \sum_{x_k \in \XX_k} f_k(x_k) \tilde{\xi}^{(k,k)}_{x_k} + 
 \sum_{j=0}^{k-1} \sum_{\brx_{j-1} \in \brX_{j-1}} g_j(\brx_{j-1}) \tilde{\xi}^{(k,j)}_{\brx_{j-1}}
 \end{align}
 where $(\tilde{\xi}^{(k,k)}_{x_k}, x_k \in \XX_k)$ are  the 
columns of the matrix $\bR_0 \cdots \bR_{k-1}$, while $(\tilde{\xi}^{(k,j)}_{\brx_{j-1}}, \brx_{j-1} \in \brX_{j-1})$
are the colums of the matrix $\bR_0 \cdots \bR_{j-1} \brR_{j}$ ($j=0, \cdots k-1$). Given a threshold 
$\epsilon > 0$, the compressed version of $f$ is 
 \[
 f_c = \sum_{x_k \in \XX_k} f_k(x_k) \tilde{\xi}^{(k,k)}_{x_k} + 
 \sum_{j=0}^{k-1} \sum_{\brx_{j-1} \in \brX_{j-1}} g_j(\brx_{j-1}) \
 \ind_{\va{g_j(\brx_{j-1})} \nor{\tilde{\xi}^{(k,j)}_{\brx_{j-1}}} \geq \epsilon} \, \tilde{\xi}^{(k,j)}_{\brx_{j-1}}
 \, . 
 \]
 Another way to compress $f$ is to keep a fixed percentage of the highest (in absolute value) normalized detail coefficients. This is the way we have done our compression experiments.
 
 \subsubsection{The line.}  We compared the compression results of our method with those obtained using classical  Daubechies12 wavelets. 
 We let our algorithm evolve until we get an approximation of size less than 16. In the experiment, this was achieved after 20 steps, and 
 led to  15 approximation coefficients. For classical wavelets, we get 16 approximation coefficients  after 6 steps. 
 For both methods, a given proportion $p$ of the details coefficients are kept to compute the compressed signals $f_c$. 
Figure \ref{segment_compression_error} presents the relative errors $\nor{f-f_c}_2/\nor{f}_2$ in terms of $p$, and shows the good behavior of the intertwining wavelets.  The compressed signal computed  with 10\% of normalized detail coefficients is shown in Figure \ref{segment_compression_signal}.

  \begin{figure}
  \includegraphics[scale=0.6]{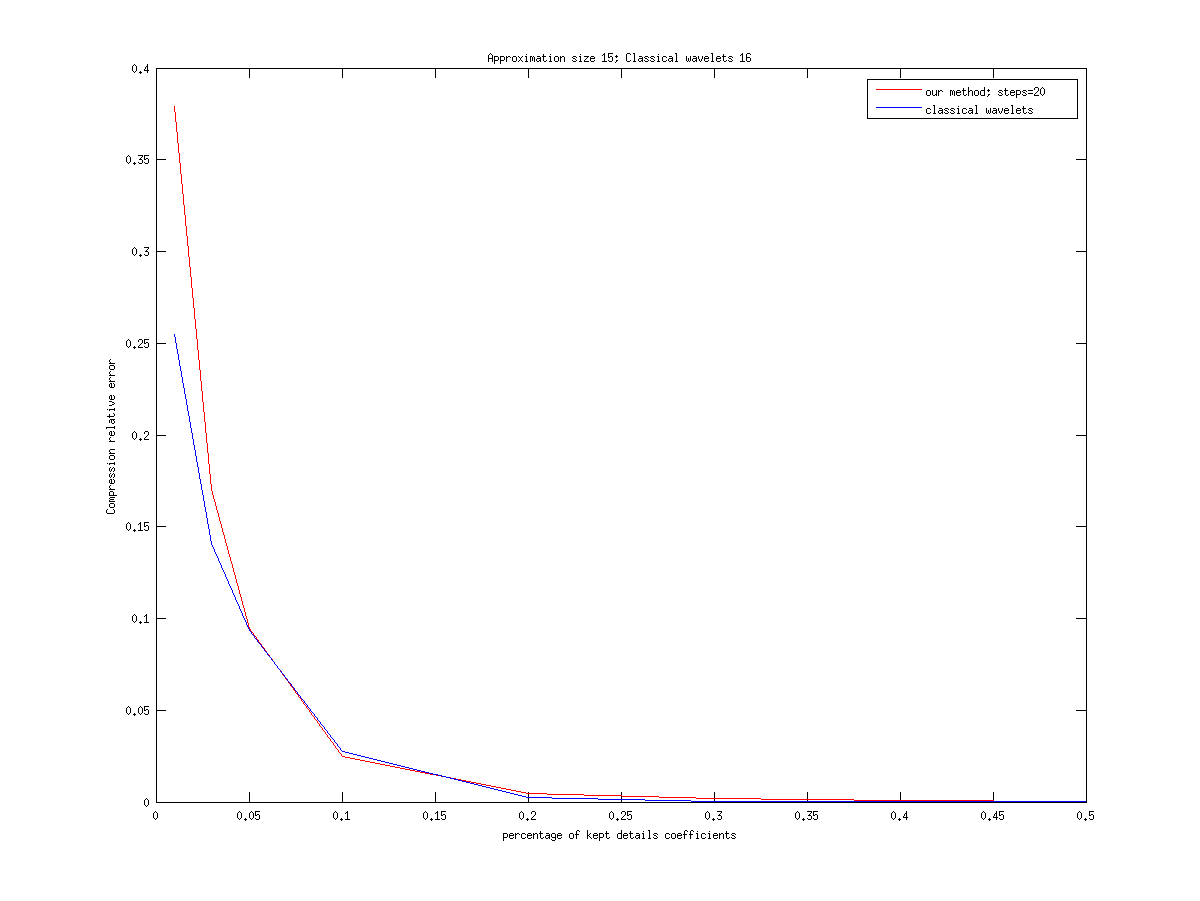}
 \caption{Relative compression error of signal in Figure \ref{signaux}(a) in terms of the percentage of kept normalized detail coefficients. In red, the error using intertwining wavelets. In blue, error using Daubechies12 wavelets. }
 \label{segment_compression_error}
 \end{figure}
  \begin{figure}
  \includegraphics[scale=0.6]{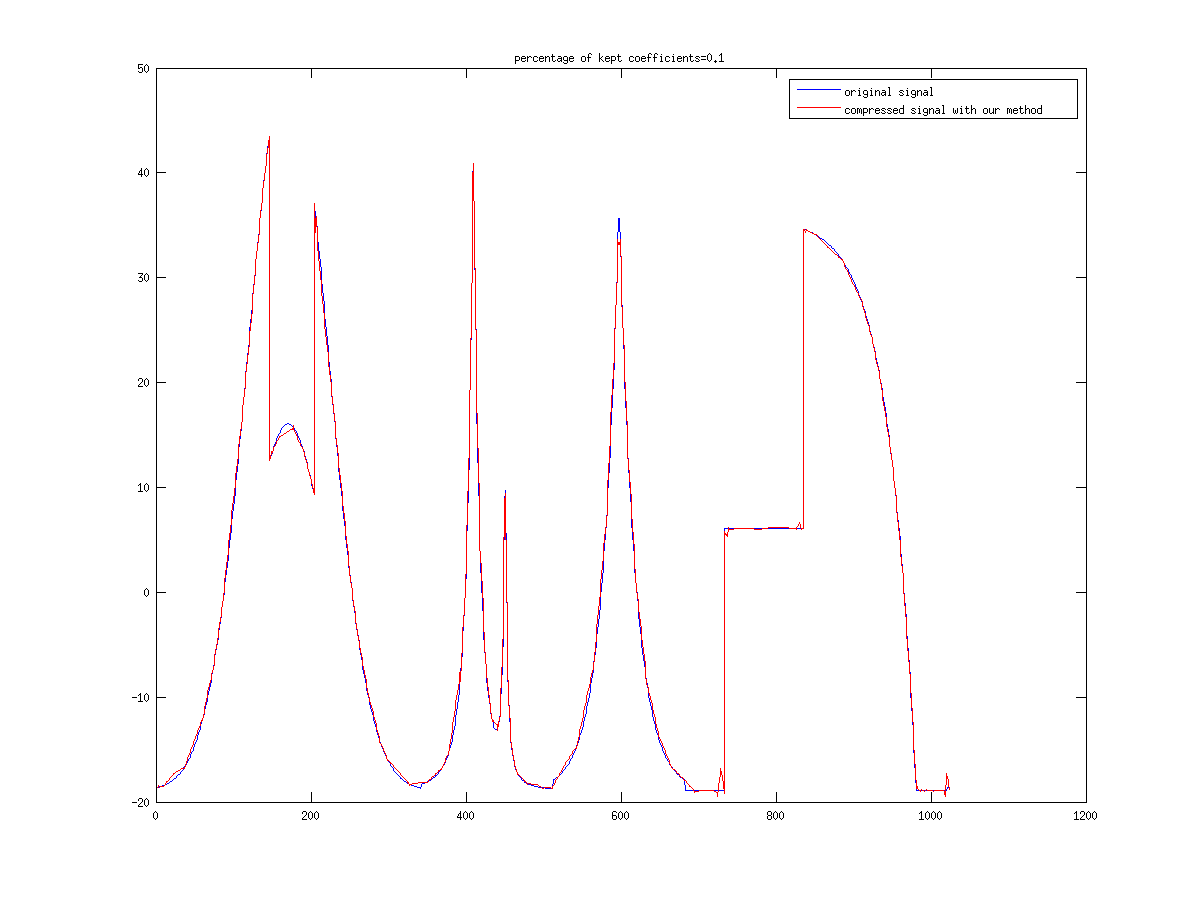}
 \caption{Original signal (blue line) and compressed one (red line) keeping 10\% of normalized detail coefficients.  }
 \label{segment_compression_signal}
 \end{figure}

 \subsubsection{The Minnesota graph.} 
 Figure \ref{minnesota_compression_signal} gives the compressed signal computed  with 10\% of kept coefficients 
  after 3 multiresolution steps.
  \begin{figure}
  \includegraphics[height=12cm, width=8cm]{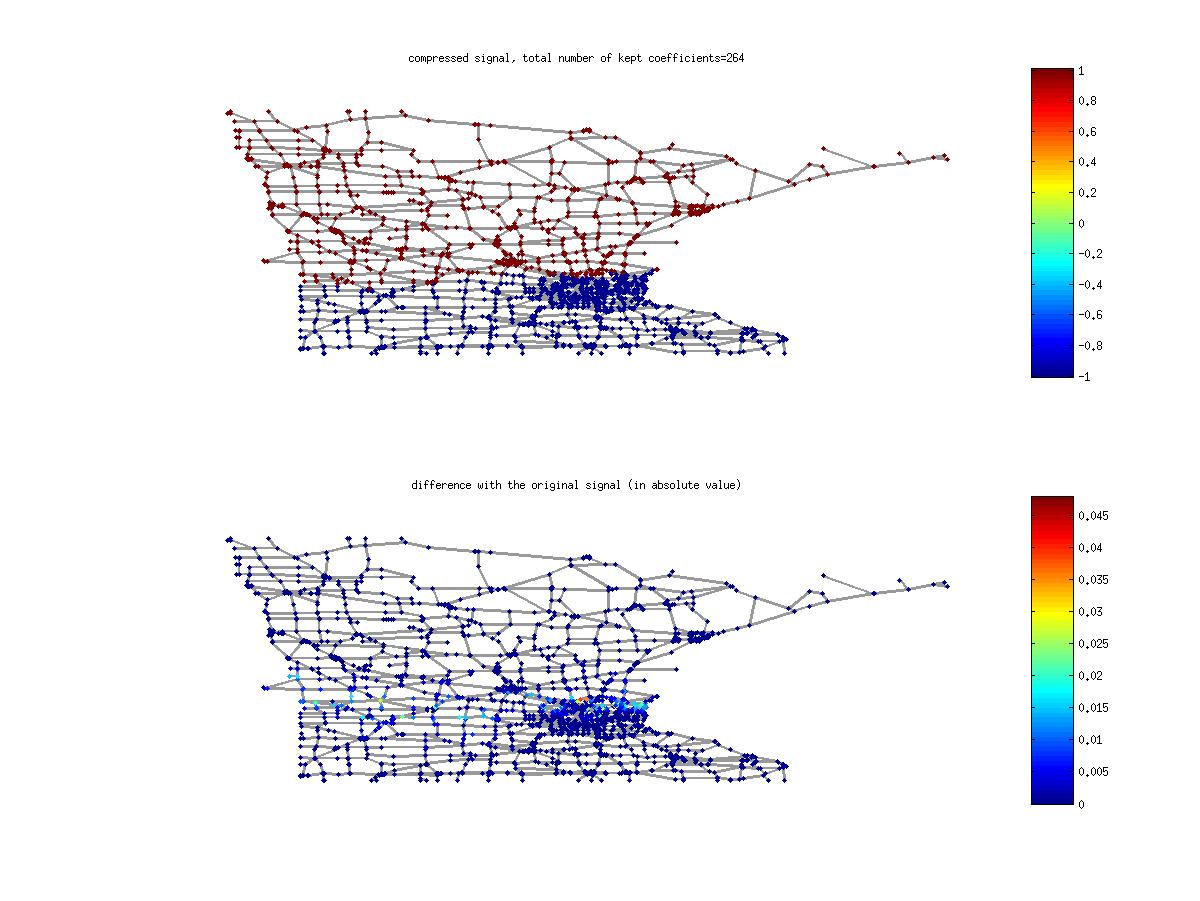}
 \caption{Compression with  10\% of kept coefficients. The  first graph is the compressed signal, the second one is the error with the original one.}
 \label{minnesota_compression_signal}
 \end{figure}
 
  In addition we compare the intertwining wavelets compression  results  with those obtained through the spectral graph wavelets pyramidal algorithm of \cite{HAM,SFV}. For this purpose, we used  the GSPBox\footnote{available at https://lts2.epfl.ch/gsp/} \cite{GSP}. The main features of this pyramidal algorithm are the following ones:
 \begin{enumerate}
 \item Subsampling: $\bX$ is chosen according to the sign of the highest frequency Fourier mode $e_{n-1}$.
 \item Weighting: $\bLL$ is computed by  the Schur complement followed by a sparsification step. 
 \item Approximation coefficients: $\bar{f}(\bx)=g(t \LL)(\bx)$, where $g$ is a low-pass filter, and $t$ is a positive real number. As in 
 \cite{SFV}, we took $g(x)=1/(1+x)$ and $t=2$ to analyse the signal of Figure \ref{signaux}(b). 
 With this choice, $\baf(\bx)=K_{1/2}f(\bx)$.
 \item The signal $\bf$ on $\bX$ is then interpolated on the whole of $\XX$: $\epsilon > 0$ being fixed, the interpolation is defined 
 as 
 \[ f_{int}(x) = (\LL^{\epsilon})^{-1}_{\XX \bX} \bLL^{\epsilon} \baf(x) \, , 
 \]
 where $\LL^{\epsilon}= \epsilon \Id - \LL$, and $\bLL^{\epsilon}$ is the Schur complement of $\LL^{\epsilon}_{\brX, \brX}$ in $\LL^{\epsilon}$. Using Proposition \ref{schur.prop}, one can see at once that $f_{int}|_{\bX}=\baf$.
 \item The error prediction $y=f-f_{int}$ is stored.
 \end{enumerate}
 At the end of one step, the signal $f \in \R^{\XX}$ is encoded by $(\baf,y) \in \R^{\bX \times \XX}$, resulting in a redundant information. Figure \ref{minnesota_compression_error} presents the relative compression  error in terms of the number of kept coefficients for the two methods. More precisely, we ran the spectral graph wavelets pyramidal algorithm for 7 steps, resulting in 20 approximations coefficients among approximately  $2 \times 2642$ stored ones. 
 We also ran our intertwining wavelets multiresolution until getting approximately the same number of approximation coefficients to get a fair comparison. This took 16 steps, resulting in 16 approximation coefficients. 
 We kept then the same number of the biggest coefficients to construct the compressed version of the signal. 
 
   \begin{figure}
  \includegraphics[scale=0.6]{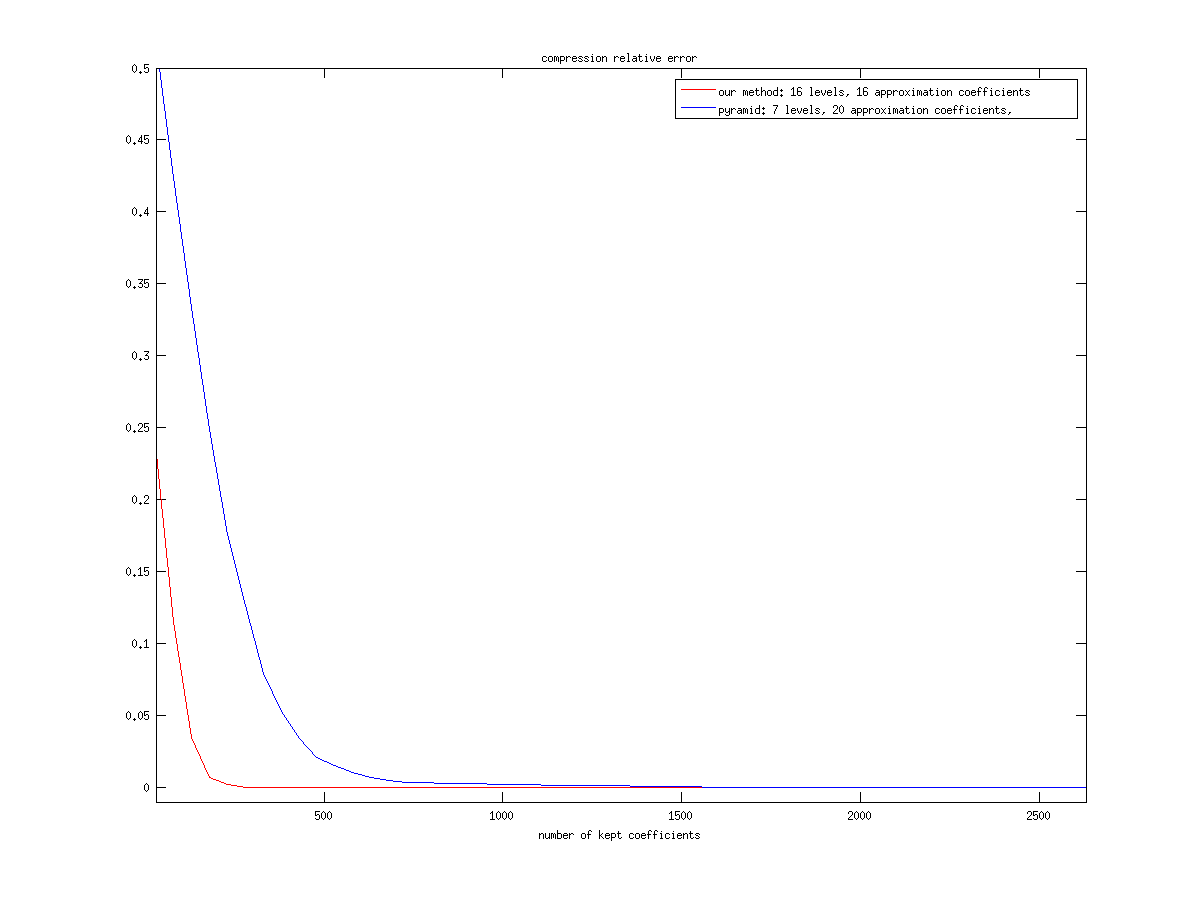}
 \caption{Relative compression error of signal in Figure \ref{signaux}(b), in terms of the number of kept coefficients. In red, the error using intertwining wavelets. In blue, error using the spectral graph wavelets pyramidal algorithm.}
 \label{minnesota_compression_error}
 \end{figure}

\subsubsection{Sensor graph}
For the signal on the sensor graph we compare once again in the same way
our method with the Pyramid algorithm on 3 steps.
In this case we also included the results of our procedure {\em without} sparsification, since they surprisingly show that the sparsification improve the results.  
In Figure \ref{sensor_compression_error}, we compare the relative compression errors  in terms of the number of kept coefficients for the Pyramid algorithm and our method with and without sparsification.   
 \begin{figure}
\begin{tabular}{cc}
  \includegraphics[scale=0.35]{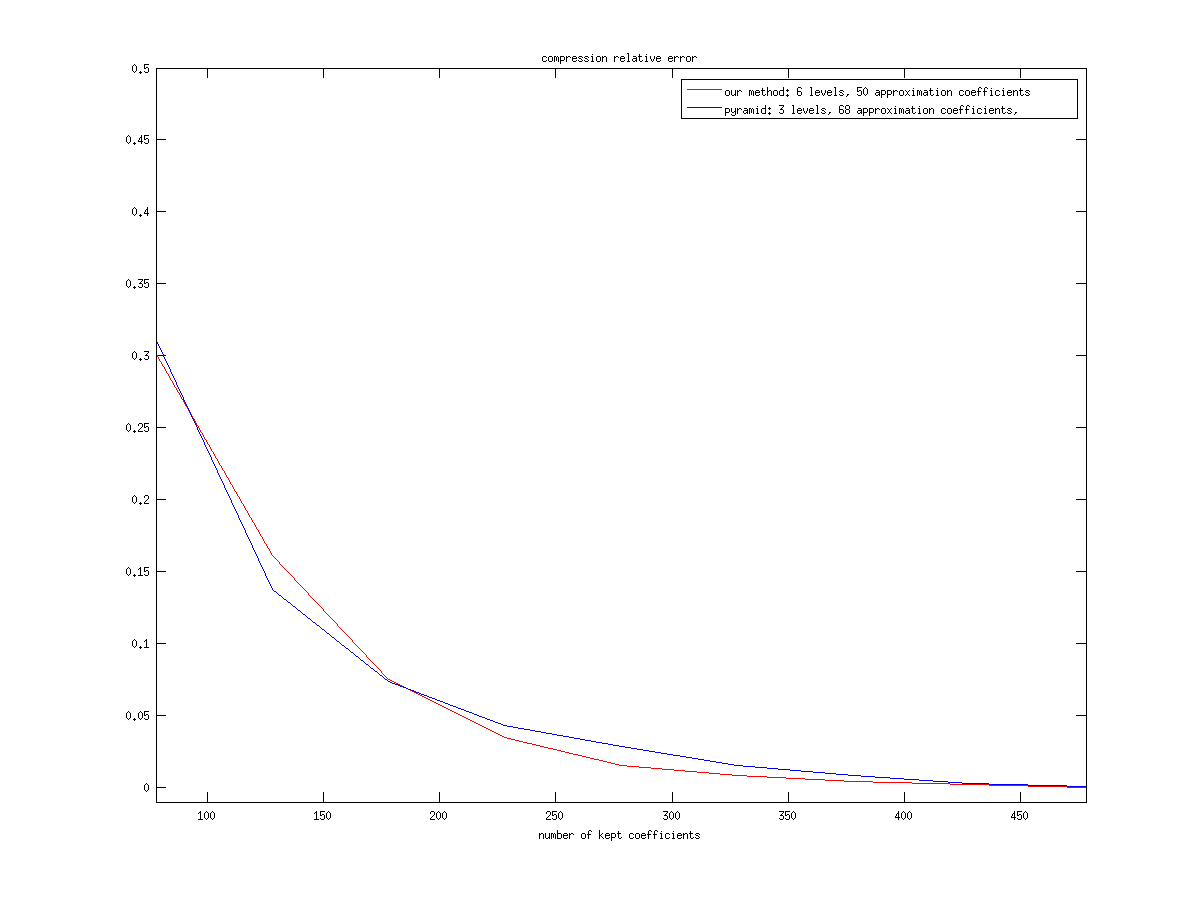}
 &
  \includegraphics[scale=0.35]{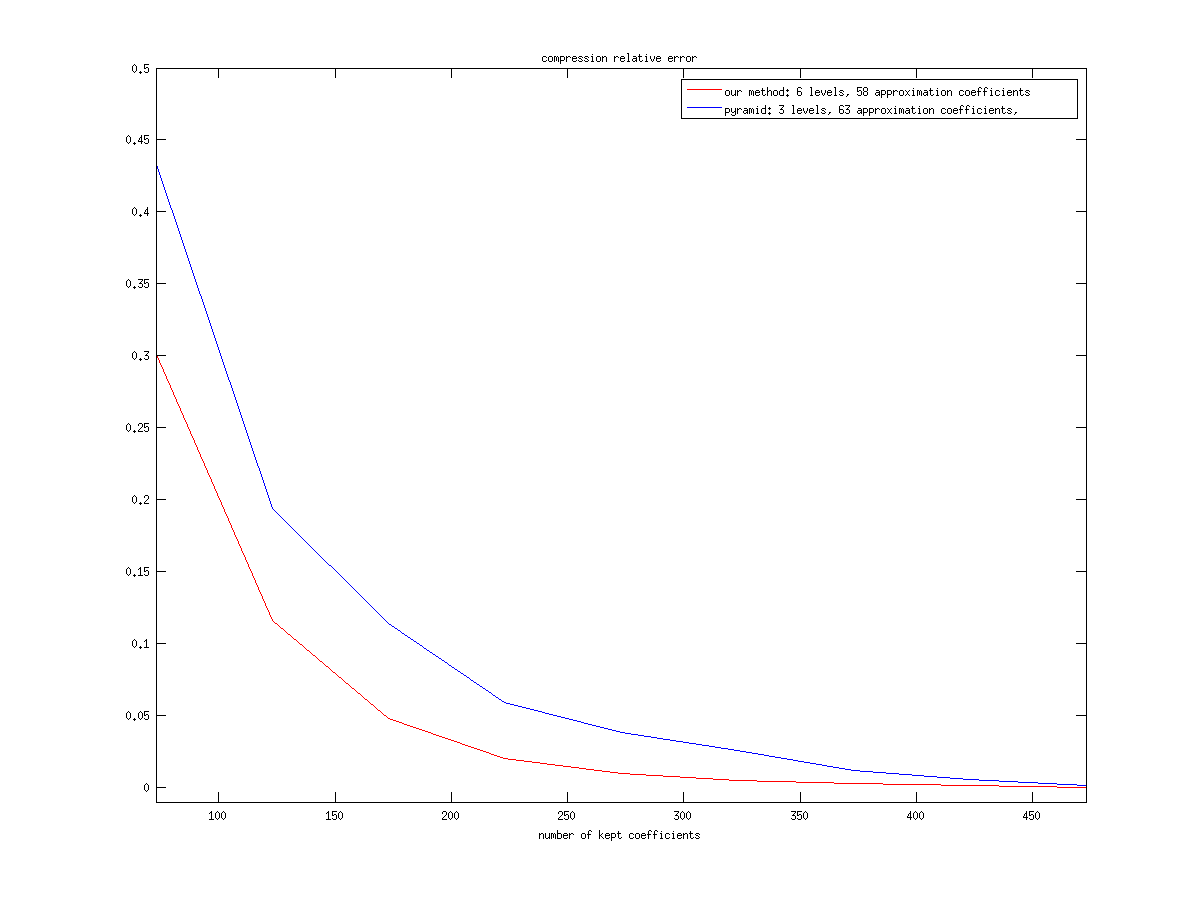}
\\
Without sparsification 
&
With sparsification 
\end{tabular}
 \caption{Relative compression error of signal in Figure \ref{signaux}(c), in terms of the number of kept coefficients. In red, the error using intertwining wavelets. In blue, error using the spectral graph wavelets pyramidal algorithm. In the first graph, we do not sparsify the graph, while we perform sparsification on the second. }
 \label{sensor_compression_error}
 \end{figure}

\section*{Aknowledgements} Marseille's team thanks Leiden university for its hospitality along many cumulated weeks during which we discovered
our wavelets. This was supported by Frank den Hollander's ERC Advanced Grant 267356-VARIS. 
  The four authors are especially grateful to Dominique Benielli (labex Archim\`ede) who made possible the numerical work. 
L. Avena was partially supported by NWO Gravitation Grant 024.002.003-NETWORKS.


\begin{thebibliography}{00}
\bibitem{ACGM1} Avena, Luca; Castell, Fabienne; Gaudilli\`ere, Alexandre; M\'elot, Clothilde.
{\it Approximate and exact solutions of  intertwining equations through random forests.} arXiv:1702.05992v1 [math.PR].

\bibitem{LA} Avena, Luca; Gaudilli\`ere, Alexandre. {\it Two applications of random spanning forests.}
 To appear in Journal of Theoretical Probability. (see  arXiv:1310.1723v4 [math.PR] for a preprint version with a different title).

  
\bibitem{CPY} Carmona, Philippe; Petit, Fr\'ed\'erique; Yor, Marc. {\it Beta-gamma random variables and intertwining relations between certain Markov processes.}  Rev. Mat. Iberoamericana 14 (1998), no. 2, 311--367.

\bibitem{CM} Coifman, Ronald R.; Maggioni, Mauro. {\it Diffusion wavelets.}  Applied and
Computational Harmonic Analysis, vol. 21, no. 1, pp 53--94, 2006.

\bibitem{COHEN} Cohen, Albert: {\it Numerical analysis of wavelet methods}, 
Studies in mathematics and its applications, Elsevier, Amsterdam, 2003

  \bibitem{DAUB}
Daubechies, Ingrid:  {\it Ten lectures on wavelets.} SIAM: Society for Industrial and Applied Mathematics. (1992)

\bibitem{DF}  Diaconis, Persi; Fill, James Allen. {\it Strong stationary times via a new form of duality}. Ann. Probab. 18 (1990), no. 4, 1483--1522.

\bibitem{DMDY} Donati-Martin, Catherine; Doumerc, Yan; Matsumoto, Hiroyuki; Yor, Marc. {\it  Some properties of the Wishart processes and a matrix extension of the Hartman-Watson laws.} 
 Publ. Res. Inst. Math. Sci. 40 (2004), no. 4, 1385--1412.   

 \bibitem{ED}     Elisha, Oren; Dekel, Shai. {\it  Wavelet decompositions of random forests. Smoothness analysis, sparse approximation and applications.}
 J. Mach. Learn. Res. 17 (2016), Paper No. 198, 38 pp. 

 \bibitem{GBC}    Gavish, Matan; Boaz, Nadler; Coifman, Ronald R.  {\it  Multiscale wavelets on trees, graphs and high dimensional data: theory and applications to semi supervised learning.} 
Proceedings of the 27th International Conference on Machine Learning (2010), 367--374. 
 
 
  \bibitem{HAM}
 Hammond, David K.; Vandergheynst, Pierre; Gribonval, R\'emi. {\it Wavelets on graphs via spectral graph theory}. Applied and Computational Harmonic Analysis (2009), Elsevier 30, no 2, 129--150.
 
   \bibitem{MAL}
Mallat, St\'ephane.  {\it A wavelet tour of signal processing.} Academic Press, (2008)

 \bibitem{MY} Matsumoto, Hiroyuki; Yor, Marc. {\it An analogue of Pitman's $2M-X$ theorem for exponential Wiener functionals. I. A time-inversion approach. } Nagoya Math. J. 159 (2000), 125--166. 
 
 \bibitem{MAR} Marchal, Philippe. {\it Loop-erased random walks, spanning trees and Hamiltonian cycles.} Elect. Comm. Probab. 5 (2000), 39--50.

\bibitem{NO} Narang, Sunil K.; Ortega, Antonio. {\it Perfect reconstruction two-channel wavelet filterbanks for graph structured data.}  IEEE Transactions on Signal Processing 60 (6) (2012), pp 2786--2799.

\bibitem{NO2} Narang, Sunil K.; Ortega, Antonio. {\it Compact support biorthogonal wavelet
filterbanks for arbitrary undirected graphs.} IEEE Transactions on Signal Processing 61 (19) (2013), 
4673--4685.

\bibitem{ND} Nguyen, Ha Q.; Do, Minh N. {\it Downsampling of signal on graphs via maximum spanning tree}. IEEE Transactions on Signal Processing 63 (1) (2015), 182--191. 
 
\bibitem{GSP} Perraudin, Nathana\"el;  Paratte, Johan;  Shuman, David I.; Kalofolias, Vassilis;  Vandergheynst, Pierre;  Hammond, David K. {\it GSPBOX: A toolbox for signal processing on graphs.} ArXiv e-prints, Aug. 2014.  http://arxiv.org/abs/1408.5781. 

 \bibitem{RP}  Rogers, L. C. G.; Pitman, J. W. {\it Markov functions.} Ann. Probab. 9 (1981), no. 4, 573--582.
 
 \bibitem{SFV} 
 Shuman, David I.; Faraji, Mohamad J.; Vandergheynst, Pierre. {\it A multiscale pyramid transform for graph signals.}   IEEE Transactions on Signal Processing 64 (8) (2016) , 2119--2134.

 
 \bibitem{SCHU}
 Shuman, David I.; Narang, Sunil K.; Frossard, Pascal; Ortega, Antonio; Vandergheynst, Pierre.
{\it  The Emerging Field of Signal Processing on Graphs: Extending High-Dimensional Data Analysis to Networks and Other Irregular Domains.}  IEEE Signal Processing Magazine 30 (2013), no. 3, 83-98.

 
\bibitem{TB}
Tremblay, Nicolas; Borgnat, Pierre. {\it  Subgraph-based filterbanks for graph signals.} IEEE Transactions on Signal Processing 64 (15) (2016), 3827-3840. 

 \bibitem{War} Warren, Jon. {\it Dyson's Brownian motions, intertwining and interlacing. } 
 Electron. J. Probab. 12 (2007), no. 19, 573--590.

\bibitem{Wi}
	D. Wilson,
	{\it Generating random spanning trees more quickly than the cover time,}
	Proceedings of the twenty-eighth annual acm symposium on the theory of computing (1996),
	296--303.

\bibitem{Zhang} Zhang, Fuzhen (Editor). {\it The Schur complement and its applications.} Numerical Methods and Algorithms. Springer (2005). 

\end{thebibliography}
\end{document}